\documentclass[a4paper,twocolumn,11pt]{quantumarticle}
\pdfoutput=1
\usepackage[utf8]{inputenc}
\usepackage[english]{babel}
\usepackage[T1]{fontenc}

\usepackage{amsmath,amsfonts,amssymb,bm}
\usepackage{graphicx,color}
\usepackage{enumerate}
\definecolor{darkblue1}{rgb}{0.18,0.19,0.57}
\usepackage{hyperref}

\usepackage{theorem}

\newtheorem{proposition}{Proposition}
\newtheorem{lemma}[proposition]{Lemma}
\newtheorem{theorem}[proposition]{Theorem}

\newenvironment{proof}{\noindent \textbf{{Proof~} }}{\hfill $\blacksquare$}

\usepackage{tikz}

\newenvironment{mytikz2}{\begin{tikzpicture}[x=0.4pt,y=0.4pt,yscale=-1,xscale=1,baseline={([yshift=+0ex]current bounding box.center)}]}{\end{tikzpicture}}

\newcommand{\nc}{\newcommand}
\nc{\ket}[1]{|#1\rangle}
\nc{\bra}[1]{\langle#1|}
\nc{\ketbra}[2]{|#1\rangle\!\langle#2|}
\nc{\braket}[2]{\langle#1|#2\rangle}
\nc{\braoprket}[3]{\langle#1|#2|#3\rangle}
\nc{\opr}[1]{\operatorname{#1}}
\nc{\avg}[1]{\langle#1\rangle}
\nc{\ketbrasame}[1]{|#1\rangle\!\langle#1|}
\nc{\tr}{\operatorname{tr}}
\nc{\E}{\mathbb{E}}
\nc{\var}{\operatorname{Var}}

\usepackage{quantikz}

\nc{\hk}[1]{\textcolor{violet}{\textbf{[hk: #1]}}}
\nc{\ch}[1]{\textcolor{blue}{\textbf{[ch: #1]}}}
\usepackage[normalem]{ulem}
\nc{\hknew}[1]{\textcolor{violet}{#1}}

\begin{document}
\title{Predicting quantum learnability from landscape fluctuation}

\author{Hao-Kai Zhang}
\email{zhk20@mails.tsinghua.edu.cn}
\affiliation{These two authors contributed equally to this work}

\affiliation{Institute for Advanced Study, Tsinghua University, Beijing 100084, China}

\author{Chenghong Zhu}
\affiliation{These two authors contributed equally to this work}
\affiliation{Thrust of Artificial Intelligence, Information Hub, The Hong Kong University of Science and Technology (Guangzhou), Guangdong 511453, China}

\author{Xin Wang}
\email{felixxinwang@hkust-gz.edu.cn}
\affiliation{Thrust of Artificial Intelligence, Information Hub, The Hong Kong University of Science and Technology (Guangzhou), Guangdong 511453, China}


\begin{abstract}
The conflict between trainability and expressibility is a key challenge in variational quantum computing and quantum machine learning. Resolving this conflict necessitates designing specific quantum neural networks (QNN) tailored for specific problems, which urgently needs a general and efficient method to predict the learnability of QNNs without costly training. In this work, we demonstrate a simple and efficient metric for learnability by comparing the fluctuations of the given training landscape with standard learnable landscapes. This metric shows surprising effectiveness in predicting learnability as it unifies the effects of insufficient expressibility, barren plateaus, bad local minima, and overparametrization. Importantly, it can be estimated efficiently on classical computers via Clifford sampling without actual training on quantum devices. We conduct extensive numerical experiments to validate its effectiveness regarding physical and random Hamiltonians. We also prove a compact lower bound for the metric in locally scrambled circuits as analytical guidance. Our findings enable efficient predictions of learnability, allowing fast selection of suitable QNN architectures for a given problem without training, which can greatly improve the efficiency especially when access to quantum devices is limited. 
\end{abstract}

\maketitle

\section{Introduction}

The rapid development of quantum computing~\cite{Preskill2018, Arute2019, Zhong2020, Bluvstein2024} has paved the way for groundbreaking advancements in various fields, with quantum machine learning (QML) emerging as a promising next-generation artificial intelligence technology. A central goal in QML is learning unknown quantum states, the quantum analog of learning probability distributions in classical machine learning~\cite{Kearns1994}. Many practical problems can be abstracted into state learning such as ground state preparation of many-body systems in condensed matter physics and chemistry~\cite{Zheng2017, Qin2020, Xu2024, Kandala2017}, quantum tomography~\cite{Huang2020, Huang2024a, Huang2024}, and combinatorial optimization~\cite{Farhi2014, Zhou2020}. Quantum neural networks (QNN) have emerged as a natural and powerful approach for such problems, generally referring to parametrized quantum circuits (PQC) equipped with classical optimizers. Training QNNs to minimize task-dependent cost functions has been a paradigm as one of the most promising paths to practical quantum advantages on near-term quantum computers~\cite{Bharti2022, Cerezo2021a}.

However, this hybrid quantum-classical paradigm still faces great challenges because of its trainability and the tradeoff with expressibility. The most studied trainability issue is the so-called barren plateau phenomenon~\cite{McClean2018, Cerezo2021, Uvarov2021a, Pesah2021, Arrasmith2021, Holmes2021, Larocca2022, Liu2024, Zhang2024a, Cerezo2023, Larocca2024, Zhang2024, Liu2023a_z, Ragone2024}, which means that the cost gradient vanishes exponentially with the system size for deep QNNs under certain conditions. Another notorious issue is the bad local minimum problem~\cite{Anschuetz2022, Bittel2021, Zhang2023}, which exists extensively in shallow circuits like quantum convolutional neural networks~\cite{Cong2019} and brickwall circuits of moderate constant depth. These issues impose severe constraints on the expressive power of trainable QNNs. On the other hand, the learning will also fail if the expressive power is too weak to contain the learning target. The learnability, defined by the overall capability of a QNN to learn a given target and quantified by the final error to the ground truth, is adequate only when both expressibility and trainability are good enough to contain the target and navigate smoothly toward it. Resolving this conflict necessitates devising problem-specific QNNs. For example, for states with entanglement area law, finite local-depth circuits can be used to avoid barren plateaus~\cite{Zhang2024}. One can also use special gates instead of universal gates to construct ansatzes, e.g., Hamiltonian variational ansatzes (HVA)~\cite{Wiersema2020, Larocca2023}. Devising problem-specific QNNs systematically urgently needs a method to diagnose which QNNs are good at learning which targets, i.e., to predict the learnability of a QNN regarding a given target. The direct way is to run an actual training procedure and see the error, but this is inefficient and depends on the details of classical optimizers, and in practice it can be hard to get the exact ground truth for benchmarks. This motivates us to seek a general and efficient metric to predict learnability.

\begin{figure*}
    \centering
    \includegraphics[width=0.99\linewidth]{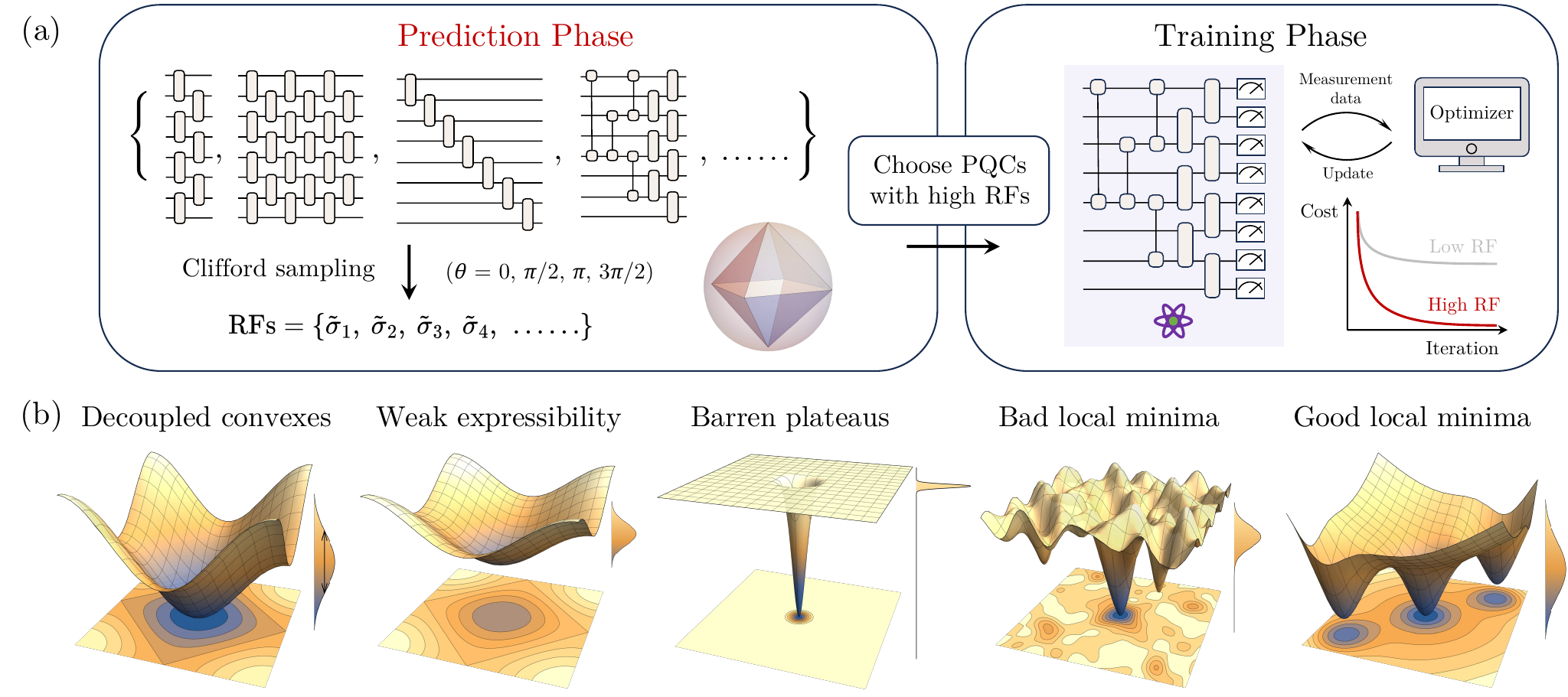}
    \caption{(a) Schematic workflow for variational quantum algorithms with learnability prediction. Prediction phase: For each architecture candidate, the relative fluctuation (RF) $\tilde{\sigma}$ of the cost function can be efficiently estimated on classical computers by Clifford sampling typically. High RF heralds better learnability and smaller training errors. Training phase: Select the architectures with high RFs and conduct training on real quantum devices. (b) Cartoons for typical training landscapes. The density of states in cost value is sketched on the right side. The spread indicates the landscape fluctuation $\sigma$. The parameter space is very high-dimensional in practice but here is only 2D for convenience of visualization.}
    \label{fig:workflow_landscape}
\end{figure*}

In this work, we demonstrate a simple and efficient metric for the learnability of QNNs termed ``relative fluctuation'' (RF). This metric emerges from comparing the fluctuations of the given training landscape and standard learnable landscapes, which is a set of decoupled convex functions. We conduct extensive numerical experiments to validate its effectiveness in predicting learnability and find this metric can detect all pivotal natures of QNNs, such as expressibility (reachability), barren plateaus, bad local minima, and overparametrization. Hence RF offers a unified picture for us to understand various crucial concepts in variational quantum algorithms. We also prove a lower bound for RF in circuits composed of local $2$-designs as analytical guidance. Notably, RF is highly efficient as it does not require actual training and can be estimated classically via Clifford sampling. Therefore, RF enables fast selections of suitable problem-specific QNN architectures, as in Fig.\,\ref{fig:workflow_landscape}(a), which can further serve as an indicator or even an objective function for automatic circuit design and improve the efficiency of future QML and variational quantum computing.

\section{Basic setup}
We will introduce the basic setup by focusing on the variational quantum eigensolver (VQE) for the ground state preparation problem. Other problems with different backgrounds can be described under a similar framework. VQE aims to prepare the ground state of an $N$-qubit system with the Hamiltonian $H$ given by $H=\sum_j\lambda_jh_j$ where $h_j$ is a Pauli string and $\lambda_j$ is a real coefficient. VQE uses a PQC $\mathbf{U}(\bm{\theta}) = \prod_{\mu} U_\mu(\theta_\mu)$ to construct the ansatz $\ket{\psi(\bm{\theta})}=\mathbf{U}(\bm{\theta})\ket{\bm{0}}$ involving $M$ tunable parameters with $\ket{\bm{0}}=\ket{0}^{\otimes N}$. $U_\mu(\theta_\mu)$ represents an elementary quantum gate depending on parameter $\theta_\mu$. A standard workflow of VQE involves running the PQC on a quantum device, measuring the cost function $C(\bm{\theta})=\braoprket{\psi(\bm{\theta})}{H}{\psi(\bm{\theta})}$, and updating the tunable parameters iteratively to minimize the cost.

The expressibility discussed below refers to not simply the expressive space size, as measured by frame potential~\cite{Holmes2021}. Instead, it depends on specific Hamiltonians, referring to the closeness of the lowest energy state within the expressive space to the true ground state as quantified by the reachability deficit~\cite{Bharti2022}. The difference compared to learnability is that the former does not care about whether the minimization process can be efficiently achieved by training. Thus, insufficient expressibility, barren plateaus, and bad local minima will all lead to poor learnability. Here bad local minima mean the cost values of most local minima concentrate far from the global minimum~\cite{Anschuetz2022}. Good local minima mean they are close, which may correspond to physical metastable states, such as the competing ordered states in high-temperature superconductors~\cite{Zheng2017, Qin2020, Xu2024}. By contrast, bad local minima should be seen as a severe trainability issue that needs to be fixed. The detailed setup is elaborated in Appendix~\ref{appendix:setup}.

\section{Relative fluctuation}
The basic information we used is the landscape fluctuation, i.e., the standard deviation of the (approximately) normalized cost function
\begin{equation}
    \sigma = \frac{\sqrt{\var_{\Theta} [C]} }{ \|\bm{\lambda}\|_1 },
\end{equation}
where $\Theta$ represents the uniform ensemble of parameter points. $\|\bm{\lambda}\|_1$ is the $1$-norm of the coefficient vector $\bm{\lambda}$ of the Pauli decomposition of the Hamiltonian. Physically, $\sigma$ can be seen as the energy fluctuation in the state ensemble induced by $\Theta$. As shown in Fig.\,\ref{fig:workflow_landscape}(b), characterizing learnability using $\sigma$ is natural and intuitive as the landscapes with either barren plateaus, insufficient expressibility, or bad local minima all tend to have their cost values concentrated around the averages, resulting in small $\sigma$, albeit to varying degrees. Especially, the exponential cost concentration is equivalent to the exponentially small variance of cost derivatives $\var_{\Theta} [\partial_\mu C]$~\cite{Arrasmith2021, Miao2024, Perez-Salinas2024}, which defines barren plateaus. However, as we will see below, besides exponential scaling implying barren plateaus, the specific values of $\sigma$ also encode unexpected useful information about expressibility, bad local minima, and overparametrization.

The landscape fluctuation $\sigma$ alone is not enough to characterize learnability, not simply as in Fig.\,\ref{fig:workflow_landscape}(b). Consider an extreme case: for a fully parametrized global unitary operator, the learnability should be strong while $\sigma$ is exponentially small. This contradiction arises because $\sigma$ does not correctly address the role of the expressive space dimension, resulting in an unfair comparison of ansatzes with different amounts of parameters. Therefore, we define the so-called relative fluctuation (RF)
\begin{equation}
    \tilde{\sigma} = \frac{\sigma}{\sigma_0},
\end{equation}
where $\sigma_0$ is the landscape fluctuation of some ``standard learnable'' function $C_0(\bm{\theta})$. The form of $\sigma_0$ can be determined heuristically through several fundamental assumptions. (i) To ensure good trainability, $C_0(\bm{\theta})$ should be convex. (ii) To ensure sufficient expressibility, there should exist a point $\bm{\theta}^*$ such that $C_0(\bm{\theta}^*)=-1$ (after normalization). (iii) $C_0(\bm{\theta})$ should have a similar form in arbitrary dimensions. A simple method is to define them recursively. Combined with (i), this recursive operation can be chosen as the addition since it preserves the convexity for functions of independent variables. This leads to  ``decoupled convex functions''
\begin{equation}
    C_0(\bm{\theta})=\sum_\mu f_\mu(\theta_\mu),
\end{equation}
where $f_\mu$ is a single-variable convex function.  According to the central limit theorem, the distribution of the cost value of $C_0(\bm{\theta})$ will converge to Gaussian asymptotically with $\sigma_0\sim 1/\sqrt{M}$ given similar forms of $f_\mu$. (iv) To conform with the parameter-shift rule~\cite{Cerezo2021a}, $f_\mu$ should be sine-shaped which is not strictly convex but is equally trainable. Combined with (ii), the scaling coefficient can be further determined to $1/\sqrt{2}$. To incorporate overparametrization, the parameter count that should go into $\sigma_0$ is the effective dimension $M_{\text{eff}}$ of the expressive space (manifold) instead of the plain dimension $M$, i.e.,
\begin{equation}
    \sigma_0 = \frac{1}{\sqrt{2M_{\text{eff}}}},
\end{equation}
where $M_{\text{eff}}$ is defined by the maximum rank of the quantum Fisher information (QFI) matrix~\cite{Haug2021a, Larocca2022, Larocca2023}
\begin{equation}
\begin{aligned}
    \mathcal{F}_{\mu\nu} &= 4\opr{Re}[\braket{\partial_\mu \psi}{\partial_\nu \psi} - \braket{\partial_\mu \psi}{\psi} \braket{\psi}{\partial_\nu \psi} ],
\end{aligned}
\end{equation}
over the parameter space. $M_{\text{eff}}$ can be efficiently estimated by random sampling. Importantly, it holds that $M_{\text{eff}}\leq M$ and $M_{\text{eff}}$ can usually be replaced by $M$, especially in underparametrization cases. One is referred to Appendix~\ref{appendix:details_RF} for elaboration. 

The numerical experiments below will show RF is indeed an effective metric with $\tilde{\sigma}=1$ signifying the computational phase transition or crossover between strong and weak learnability. Before that, we first give a lower bound of compact form as analytical guidance.




\begin{theorem}\label{thm:1}
For a PQC composed of blocks forming independent local $2$-designs and an $r$-local Hamiltonian, the landscape fluctuation is lower bounded by
\begin{equation}
    \sigma \geq 2^{-r\chi\beta} \frac{\|\bm{\lambda}\|_2}{\|\bm{\lambda}\|_1},
\end{equation}
where $\chi$ is the maximum local depth and $\beta$ is the maximum block size of the circuit.
\end{theorem}
The proof is left to Appendix~\ref{appendix:proof} where a tighter bound is proven with the picture of information scrambling through paths across the circuit. Here a ``block'' represents a grouped continuous series of elementary gates as a unit constructing a PQC, such as the $\mathrm{SU}(4)$ Cartan decomposition. The local depth refers to the number of blocks acting on a qubit and $\chi$ refers to its maximum value over all qubits~\cite{Zhang2024}. $\chi$ will reduce to the conventional depth $D$ for common circuits with repeated layers. $\beta$ is the maximum support size of all blocks. For circuits composed of two-qubit blocks, we just have $\beta=2$. By use of Theorem~\ref{thm:1}, one can roughly estimate landscape fluctuations and RFs without numerical simulations for locally scrambled circuits of arbitrary shapes and geometries. For example, for constant local-depth circuits and extensive Hamiltonians, Theorem~\ref{thm:1} gives $\sigma\in\Omega(1/\sqrt{N})$ consistent with the $1/\sqrt{M}$ scaling of $\sigma_0$ as $M$ scales linearly with $N$.

RF is classically simulable via Clifford sampling for standard PQCs with stabilizer initial states since RF only involves second-order moments. Instead of $[0,2\pi)$, sampling the Pauli rotation angles from the discrete set $\{0,\pi/2,\pi,3\pi/2\}$ is enough, which is efficient by the stabilizer formalism. The consistency of the two sampling methods is verified in Appendix~\ref{appendix:additional}. In the following numerical simulations, we focus on uniform distributions of tunable parameters but we remark that the equivalence of the two sampling methods can still hold even for non-uniform distributions~\cite{Heyraud2023}. Thus, RF is also applicable in scenarios with non-uniformly random initialization.

\begin{figure}
    \centering
    \includegraphics[width=0.99\linewidth]{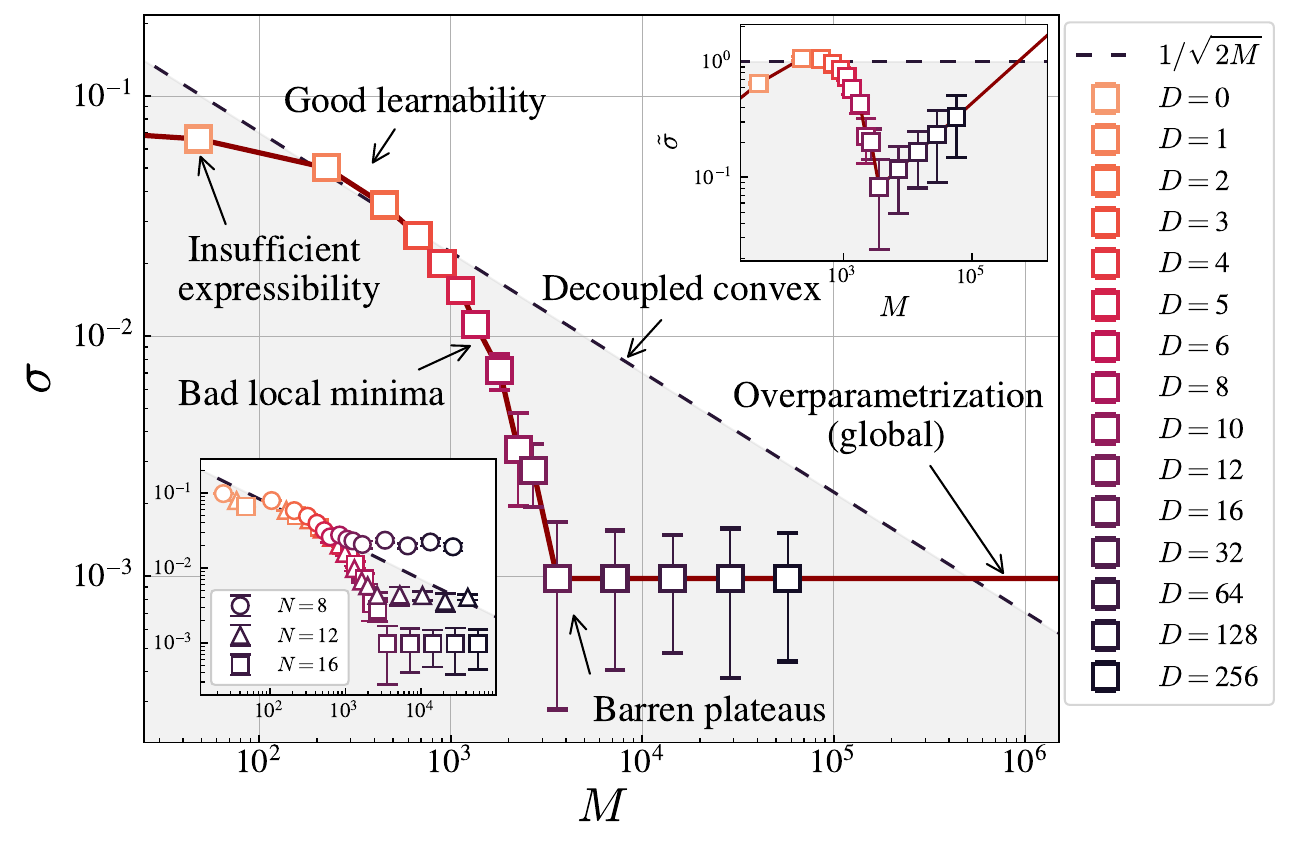}
    \caption{The landscape fluctuation $\sigma$ vs the number of tunable parameters $M$ in 1D brickwall circuits of depth $D$ and $N=16$ for the 1D cluster model, where $D=0$ stands for a layer of single-qubit rotation gates. Each block is a two-qubit random Clifford gate. The dashed line represents a reference boundary partitioning learnable and unlearnable phases. The upper inset shows the same data but in terms of the relative fluctuation $\tilde{\sigma}$. The lower inset compares the data for $N=8,12,16$.}
    \label{fig:depth_evolving}
\end{figure}

\section{Numerical experiments}
\subsection{Learnability evolution}
A typical behavior of $\sigma$ evolving with circuit depth $D$ is shown in Fig.\,\ref{fig:depth_evolving}. We use 1D brickwall circuits like the first two circuits in Fig.\,\ref{fig:workflow_landscape}(a) representing $D=1,3$. The Hamiltonian is chosen as the 1D transverse-field cluster model
\begin{equation}
    H_{ZXZ} = -\sum_j Z_{j-1} X_{j} Z_{j+1} - h \sum_j X_j,
\end{equation}
whose ground state at $|h|<1$ is a 1D symmetry-protected topological (SPT) state~\cite{Chen2011} that can be prepared with a layer of Hadamard gates followed by a layer of nearest-neighbor CZ gates. $\sigma$ decreases exponentially with $D$, consistent with Theorem~\ref{thm:1}, until the circuit forms a global $2$-design at the linear scale and converges to $\|\bm{\lambda}\|_2/\sqrt{2^N+1}\|\bm{\lambda}\|_1$.

Nontrivial information emerges from the comparison with $\sigma_0$ (cf. the inset of Fig.\,\ref{fig:depth_evolving}). At first, the circuit is too weak to include the ground state ($\tilde{\sigma}<1$). Then the expressive power becomes sufficient ($\tilde{\sigma}\geq 1$) for $D=1,2$ without causing bad local minima at least for $N = 16$. Next, the landscapes start to have extensive bad local minima and even become almost flat with barren plateaus at $D\in \Omega(N)$ ($\tilde{\sigma}\ll 1$). Finally, RF turns to increase and again enters the learnable phase ($\tilde{\sigma}\geq 1$) because the circuit is overparametrized by exponential parameters, resulting in the existence of a finite-gradient direction corresponding to the unitary rotation towards the ground state. This example shows RF can provide a unified description of the typical learnability evolution with circuit depth. 


\begin{figure}[t]
    \centering
    \includegraphics[width=0.98\linewidth]{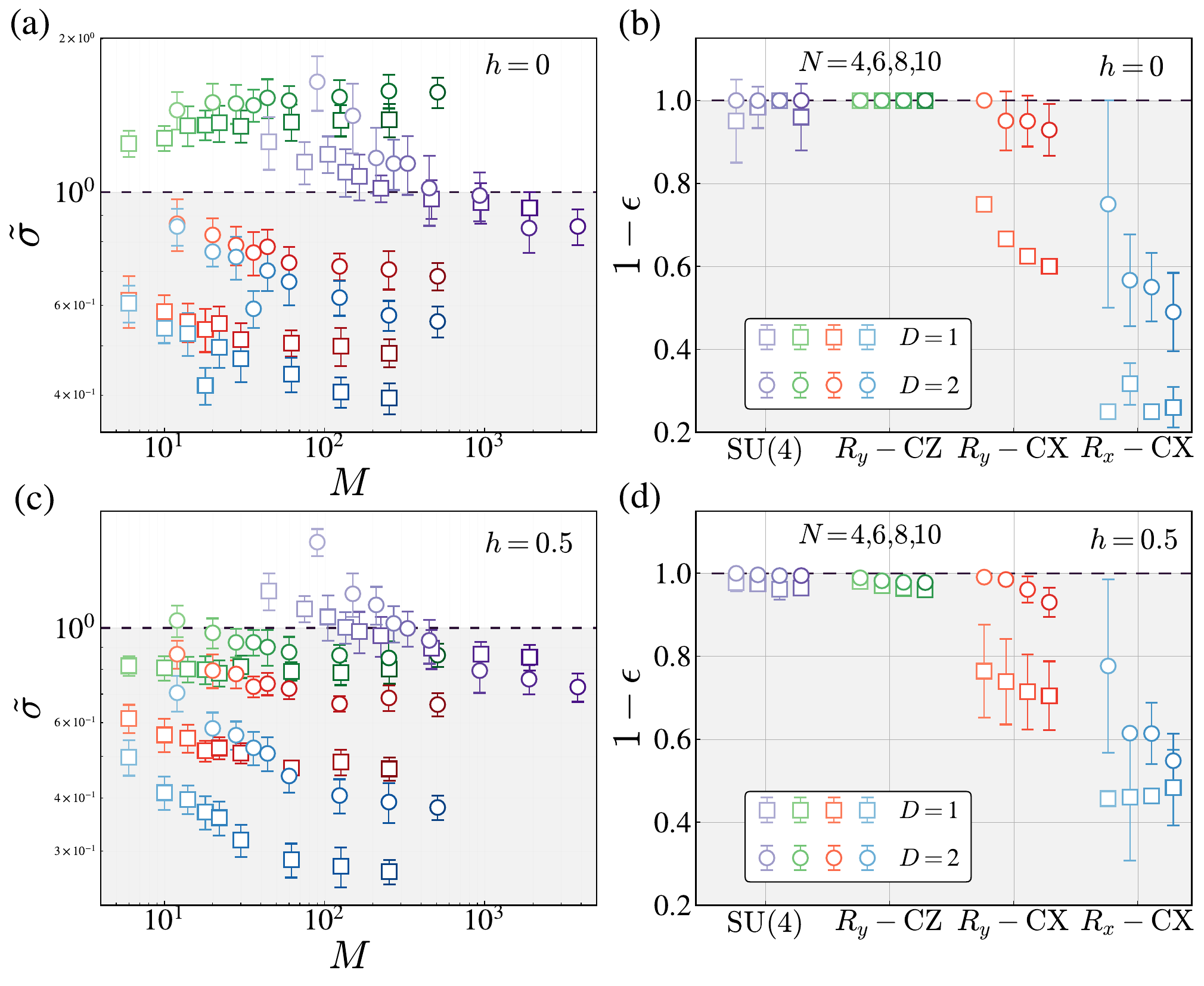}
    \caption{Detecting insufficient expressibility. (a) The RF $\tilde{\sigma}$ vs parameter count $M$ for 1D brickwall circuits with different depths $D$ and elementary gate templates regarding $H_{ZXZ}$. The colors with increasing intensity correspond to $N = 4, 6, 8, 10, 12, 16, 32, 64, 128$. (b) The actual training performance. $\epsilon$ is the training error after convergence. ``$R_y$-$\opr{CZ}$'' represents each two-qubit block is composed of two $R_y$ rotation gates followed by a $\opr{CZ}$ gate. Others are defined similarly. $\opr{SU}(4)$ represents the full parametrization of the two-qubit block. (c) and (d) The results obtained similarly to (a) and (b) but with a transverse field $h=0.5$.}
    \label{fig:express}
\end{figure}

\subsection{Gates and expressibility}
RF can also reflect which kinds of elementary gates are suitable for the given Hamiltonian. As shown in Fig.\,\ref{fig:express}, we choose different templates for the two-qubit blocks in brickwall circuits, compute their RFs, and test their real training performances regarding $H_{ZXZ}$. We find there is an explicit correlation between RFs and actual training errors. Higher RFs herald smaller errors. At $h=0$, $\opr{SU}(4)$ and $R_y$-$\opr{CZ}$ with $\tilde{\sigma}\geq 1$ indeed contain the exact solution and lead to almost zero errors, while $R_y$-$\opr{CX}$ and $R_x$-$\opr{CX}$ with $\tilde{\sigma}< 1$ have insufficient expressibility and relatively large errors. If an external field term is introduced, we get $\tilde{\sigma}<1$ for $R_y$-$\opr{CZ}$ as it no longer necessarily contains the exact ground state, as shown in Figs.~\ref{fig:express}(c) and (d). Further supporting data are left in Appendix~\ref{appendix:additional}. It is worth noticing that RF of $\opr{SU}(4)$ converges below $1$ while RF of $R_y$-$\opr{CZ}$ stays above in Figs.~\ref{fig:express}(a) and (b). This may be attributed to the bad local minima in locally underparametrized circuits as discussed below.

\begin{figure}[t]
    \centering
    \includegraphics[width=0.99\linewidth]{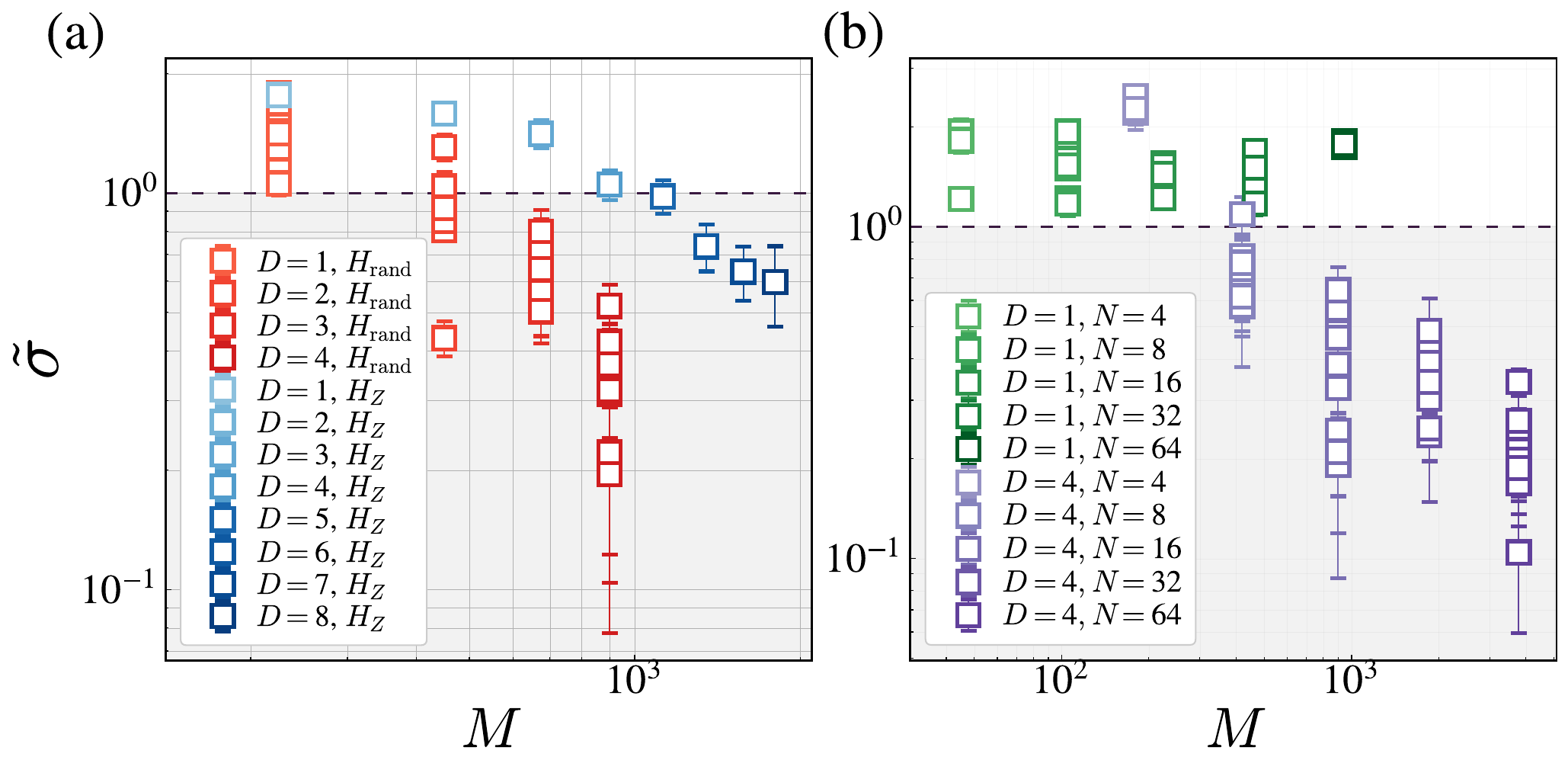}
    \caption{Detecting bad local minima. (a) The RF $\tilde{\sigma}$ vs parameter count $M$ for random Hamiltonians $H_{\text{rand}}$ with different brickwall circuit depth $D$ at $N=16$. (b) The same setting as in (a) but varying the number of qubits. Different points with the same markers correspond to different samples of $H_{\text{rand}}$. The results for $H_Z$ are also depicted in (a) for comparison.}
    \label{fig:localmin}
\end{figure}

\subsection{Transition to bad local minima}
The existence of bad local minima in circuits composed of locally scrambled blocks is determined by the local overparametrization ratio~\cite{Anschuetz2022}
\begin{equation}
    \gamma = \frac{M_{\text{cone}}}{ 2d_{\text{cone}}},
\end{equation}
where $M_{\text{cone}}$ is the number of tunable parameters within the backward causal cone of an observable and $d_{\text{cone}}$ is the Hilbert space dimension of the qubits within the cone. $\gamma\ll 1$ signifies local underparametrization and extensive bad local minima as $N$ increases. $\gamma$ decays exponentially in $D$, suggesting a typical depth $D_c$ separating phases of good and bad local minima. This can be witnessed by RFs. We choose the same setting as in Ref.~\cite{Anschuetz2022}. The Hamiltonians are generated from random back evolution
\begin{equation}
    H_{\text{rand}} = \mathbf{V}^\dagger H_{Z} \mathbf{V},
\end{equation}
where $H_{Z}=-\sum_j Z_j$ and $\mathbf{V}$ is a random 1D (reversed) brickwall circuit of depth $D$. Such a setting ensures sufficient expressibility for brickwall circuits of the same depth as ansatzes. As shown in Fig.\,\ref{fig:localmin}, for circuits of depth $D=1$ ($\gamma\gtrsim 1$), RF stays above $1$ even for large $N$, while for deeper circuits such as $D=4$ ($\gamma\ll 1$) RF gradually falls below $1$, in agreement with the results in Ref.~\cite{Anschuetz2022}. On the other hand, in Appendix~\ref{appendix:additional}, we also show that RF is insensitive to good local minima such as the symmetry-breaking states compared to the GHZ-type true ground states in the ferromagnetic phase of the transverse-field Ising (TFI) model.

\begin{figure}[b]
    \centering
    \includegraphics[width=0.98\linewidth]{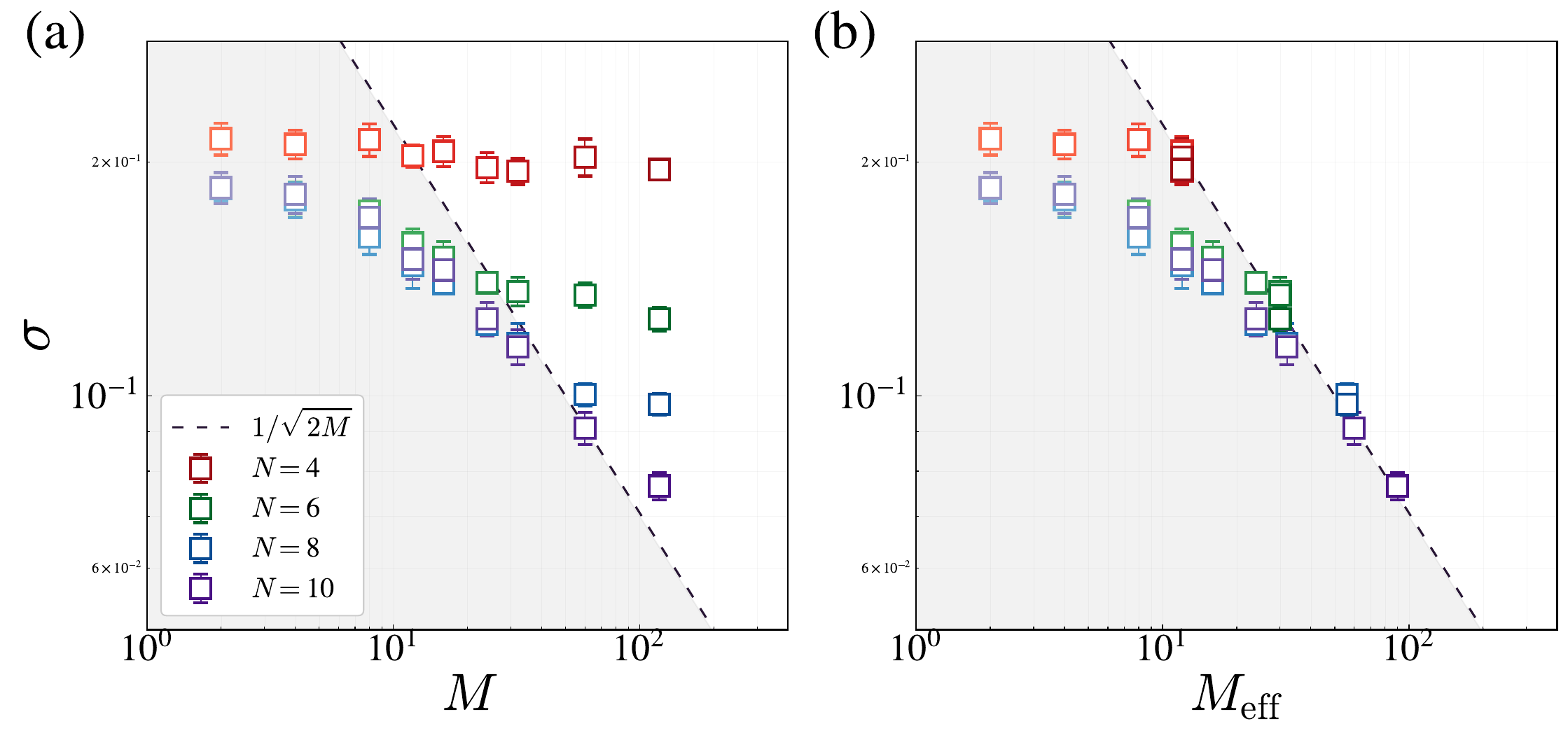}
    \caption{Detecting overparametrization with polynomial parameters. (a) The $\sigma$-$M$ diagram for the 1D HVA with only Pauli $X$ and nearest-neighbor $ZZ$ rotations and the 1D TFI model at $h=0.5$. The colors with increasing intensity correspond to depth $D= 1,2,4,6,8,12,16,30,60$. (b) The same data as in (a) but in terms of the effective dimension $M_{\text{eff}}$.}
    \label{fig:qaoa}
\end{figure}

\subsection{Polynomial overparametrization}
It is known that certain PQCs can be overparametrized with only a polynomial number of parameters in the sense that the effective dimension of the expressive space $M_{\text{eff}}$ saturates when increasing the number of parameters $M$, e.g., those PQCs with dynamical Lie algebras (DLA) of polynomial dimensions~\cite{Larocca2023}. The training errors are shown to be almost zero after reaching overparametrization, signifying high learnability. This can also be captured by RFs. We use the 1D TFI model 
\begin{equation}
    H_{\text{TFI}} = -\sum_{j} Z_jZ_{j+1} - h\sum_{j} X_j,
\end{equation}
at $h=0.5$ with open boundary conditions and a 1D HVA circuit. Each layer consists of a layer of $R_x(\theta)$ and a layer of nearest-neighbor $R_{zz}(\theta)$ gates. The initial state is $\ket{+}^{\otimes N}$. See Appendix~\ref{appendix:setup} for detailed information. Fig.\,\ref{fig:qaoa} shows that as $D$ increases, overparametrization ($M_{\text{eff}}$ saturation) occurs once the $(M, \sigma)$ points touch the line $\sigma=1/\sqrt{2M}$. This again indicates that this line can serve as a boundary between learnable and unlearnable phases since polynomial overparametrization results in perfect learnability. This is also consistent with the conjecture and results in Refs.~\cite{Larocca2022} and \cite{Ragone2024} since here the cost variance $\sigma^2$ and the inverse of the saturated $M_{\text{eff}}$ are identified on the phase boundary.

\section{Discussion}
In this work, we introduce the relative fluctuation as a metric of the problem-specific learnability of QNNs, which compares the degrees of cost concentration of given landscapes with standard learnable landscapes. Our results demonstrate that RF can effectively detect both insufficient expressibility and poor trainability, where the ability to detect bad local minima was rarely addressed before. Conceptually, RF serves as a bridge connecting key concepts in QML and variational quantum algorithms. Practically, RF is a highly efficient training-free metric independent of optimizer details such as learning rate and can be classically computed via Clifford sampling. This enables fast selection of suitable QNN architectures and possible advanced automatic search~\cite{Ostaszewski2019, Zhang2020b_z, Du2020a_z, Zhang2021_z}.

Here we directly focus on the second-order moment because the first moment usually vanishes. For example, given a traceless Hamiltonian, one layer of single-qubit rotations is enough to make the cost average become zero. The consequences of nonzero averages and higher-order moments are further discussed in Appendix~\ref{appendix:details_RF}. In addition, here we mainly focus on physical scenarios where the learning target is encoded in Hamiltonians. If the target is encoded in initial states of finite samples like in some QML algorithms, learnability should also incorporate overfitting and generalization issues~\cite{Du2021}.

\begin{acknowledgements}
We acknowledge the stimulating discussions with Shuo Liu and Shi-Xin Zhang. The numerical simulations are based on the TensorCircuit~\cite{Zhang2022_z} and PyClifford~\cite{Hu2023} packages.
This work was partially supported by the National Key R\&D Program of China (Grant No.~2024YFE0102500), the Guangdong Provincial Quantum Science Strategic Initiative (Grant No.~GDZX2303007), the Guangdong Provincial Key Lab of Integrated Communication, Sensing and Computation for Ubiquitous Internet of Things (Grant No.~2023B1212010007),  the Start-up Fund (Grant No.~G0101000151) from HKUST (Guangzhou), the Quantum Science Center of Guangdong-Hong Kong-Macao Greater Bay Area, and the Education Bureau of Guangzhou Municipality.
\end{acknowledgements}

\bibliographystyle{quantum}

\clearpage
\newpage
\onecolumn
\appendix

\renewcommand{\theproposition}{S\arabic{proposition}}
\setcounter{proposition}{0}
\renewcommand{\thedefinition}{S\arabic{definition}}
\setcounter{definition}{0}

\renewcommand{\thefigure}{S\arabic{figure}}
\setcounter{figure}{0}
\numberwithin{equation}{section}


\section{Preliminaries and basic setup}\label{appendix:setup}

In this section, we elaborate on the basic setup of our work. For clarity, we first introduce some preliminaries and background. Throughout this paper, we consider a $N$-qubit system $\mathcal{H}^{\otimes N}$ with local Hilbert space being the two-dimensional complex linear space $\mathcal{H}=\mathbb{C}^2=\opr{span}\{\ket{0},\ket{1}\}$. The results in this paper can be easily generalized to qudit systems with local Hilbert space dimension $d$ by replacing Pauli strings with other operator strings as an orthogonal basis. Here a Pauli string refers to a tensor product of single-qubit Pauli operators 
\begin{equation}
    X = \begin{pmatrix}
        0 & 1 \\ 1 & 0
    \end{pmatrix},\quad\quad
    Y = \begin{pmatrix}
        0 & -i \\ i & 0 \\
    \end{pmatrix},\quad\quad
    Z = \begin{pmatrix}
        1 & 0 \\ 0 & -1 \\
    \end{pmatrix}.
\end{equation}
We register the $N$ qubits as $\{q_1,q_2,\cdots,q_N\}$ and use $s(O)$ to denote the support of a linear operator $O$ on $\mathcal{H}^{\otimes N}$. Here the support refers to the subset of qubits that are acted by $O$ non-trivially. For example, $s(I\otimes X\otimes Y\otimes Z)=s(X_2 Y_3 Z_4)=\{q_2,q_3,q_4\}$. We use $|s|$ to denote the number of qubits in a subset of qubits $s$.

\textit{Ground state problem.---}For brevity, we will illustrate the setup using variational quantum eigensolver (VQE) and quantum physics problems as examples. Other variational quantum algorithms (VQA) can be described in a similar manner with some differences depending on the background of related tasks, such as quantum state learning, combinatorial optimization, quantum machine learning (QML), etc. VQE aims to prepare the ground state $\ket{G}$ of a given Hamiltonian $H$. The Hamiltonian is a Hermitian operator on $\mathcal{H}^{\otimes N}$ representing the total energy of the system. The Hamiltonian is usually given in the form of its Pauli string decomposition 
\begin{equation}
    H=\sum_j \lambda_j h_j,
\end{equation}
where $h_j$ denotes a Pauli string and $\lambda_j$ denotes the corresponding real coefficient. The ground state $\ket{G}$ is the eigenstate of $H$ associated with the lowest eigenvalue (ground state energy). Of course, there may exist multiple degenerate ground states for certain Hamiltonians. Here we just focus on the non-degenerate case and the degenerate case is similar. An important feature of a Hamiltonian is its locality. We say a Hamiltonian is $r$-local if the maximum size of the supports of its subterms is $r$, i.e., 
\begin{equation}
    r=\max_j|s(h_j)|.    
\end{equation}
Most realistic physical models only contain local interactions. That is to say, $r$ does not scale with the system size $N$. Note that here the term ``local'' actually means ``few-body'' by convention, distinguished from the spatial or geometrical locality which is more strict by requiring the qubits in the support of each subterm to be adjacent on the given lattice. Another common assumption on realistic physical models is that the Hamiltonian is extensive because energy is supposed to be an extensive quantity of physical systems. Here an extensive Hamiltonian refers to the case that the support of the local Hamiltonian covers mainly all the qubits, e.g., $s(H)=N$, which will result in the number of subterms in $H$ growing at least linearly in the system size.

As the entire Hilbert space dimension is exponentially large with the system size $N$, it will take exponential time and memory space to get the ground state if one chooses to diagonalize the Hamiltonian directly on classical computers (known as exact diagonalization, ED). For example, on a common modern classical computer, the ground state of a small-size system can be easily solved by ED but it will be hard when the system sizes become moderately large, such as around $20$ qubits. However, this does not mean these are all classical algorithms can do. Many successful classical algorithms have been developed in past decades based on insights into specific quantum systems~\cite{Ceperley1986, White1992, Verstraete2006, Vidal2007, Schollwock2011}. A large class of them belongs to the so-called \textbf{variational methods}, e.g., variational Monte Carlo (VMC), the density matrix renormalization group (DMRG) also known as the variational matrix product state (MPS) algorithm, and other tensor network (TN) algorithms. These methods limit themselves to a polynomially large sub-manifold of the entire Hilbert space and minimize the energy variationally as low as possible to approximate the ground state. Here the sub-manifold actually refers to a certain form of parametrized quantum states $\ket{\psi(\bm{\theta})}$ (also known as ansatz) that depends on polynomially many tunable parameters. If the ansatz can be efficiently manipulated by classical computers and the ansatz can contain the true ground state approximately by leveraging some prior information about the given system, these variational classical methods would have a fairly high probability of success in finding the ground state.

For example, most systems in quantum many-body physics admit low-entangled ground states, i.e., obeying the entanglement area law with at most some logarithmic corrections. Specifically, the area law, also known as the boundary law, refers to that the entanglement entropy $S_A$ of a subsystem $A$ only scales with the size of the boundary of the subsystem $|\partial A|$, i.e., 
\begin{equation}
    S_A\propto |\partial A|, 
\end{equation}
in stark contrast to typical Haar-random states whose entanglement entropy scales with the volume of the subsystem, i.e., obeying the volume law $S_A\propto |A|$. It is proved that the ground states of local gapped systems indeed satisfy the area law in one spatial dimension~\cite{Hastings2007} and it is empirically believed that the same also holds true in two and higher dimensions based on extensive instances. Therefore, leveraging this feature, one can design the so-called tensor network states such as projected entangled paired states (PEPS) and multiscale entanglement renormalization ansatz (MERA), which follows the entanglement area law with at most logarithmic corrections and hence can provide good approximations for the ground states of quantum many-body systems. For one-dimensional systems, the corresponding TN states, such as MPS (also known as 1D PEPS), will lead to small treewidth and can be contracted efficiently on classical computers~\cite{Markov2008}. Hence, it is known that 1D quantum systems are classically tractable. Nevertheless, for two and higher-dimensional systems, the corresponding TNs are in general hard to contract though there exist many heuristic contraction strategies without controllable errors~\cite{Schuch2007, Orus2019}. As a consequence, many two-dimensional systems of interest remain difficult to draw definitive conclusions, such as the 2D Fermi-Hubbard model and $t$-$J$ model that are relevant to the research of high-temperature superconductivity.

Of course, there also exist many successful non-variational classical methods for specific quantum systems, but these methods also face their respective difficulties when dealing with general systems of interest. For example, the quantum Monte Carlo (QMC) algorithms, which are designed to directly compute the expectation values of observables of interest without explicitly representing the wavefunction via certain mappings to classical statistical models, are very efficient for many systems such as the square-lattice Heisenberg model. However, QMC will encounter the so-called sign (phase) problem and become inefficient when dealing with frustrated quantum systems.

\textit{Variational quantum eigensolver.---}Recent rapid developments of noisy intermediate-scale quantum (NISQ) devices offer new hope for this ground-state preparation problem of quantum nature. Based on the above discussion of classical variational methods, a natural extension is to \textbf{use quantum circuits as variational ansatzes}. This is exactly the core idea of VQE. On quantum devices, one can run quantum circuits and read out the measurement results efficiently, which is believed to be exponentially hard on classical computers. Because of the limitations of current NISQ devices, textbook quantum algorithms that require very large-scale quantum computers with very high precision cannot be implemented in the near term. Thus, VQE with its variants in VQAs is expected to be one of the most promising routes toward practical quantum advantage in the NISQ era.

To be specific, VQE uses a parametrized quantum circuit (PQC)
\begin{equation}
    \mathbf{U}(\bm{\theta}) = \prod_{\mu}^{\leftarrow} U_\mu(\theta_\mu),
\end{equation}
to construct the variational ansatz $\ket{\psi(\bm{\theta})}=\mathbf{U}(\bm{\theta})\ket{\bm{0}}$ by acting on the initial state $\ket{\bm{0}}=\ket{0}^{\otimes N}$. The left arrow means that the index $\mu$ follows the increasing order from right to left in the product. $U_\mu(\theta_\mu)$ represents a single elementary quantum gate depending on the parameter $\theta_\mu$. Note that the quantum gates without parameter dependence such as CNOT are also included by considering them as gates with fixed parameters. The common form of parametrized gates is $U_\mu(\theta_\mu)=e^{-i\Omega_\mu\theta_\mu/2}$ where $\Omega_\mu$ is a Hermitian generator such as a Pauli string. In the context of quantum state learning or quantum machine learning, a PQC can be seen as a quantum neural network (QNN) when equipped with a classical optimizer. A typical PQC template is the so-called 1-dimensional brickwall circuit, e.g.,
\begin{equation}
\begin{mytikz2}
\draw [line width=0.75]    (50,141.93) -- (224.93,141.93) ;
\draw [line width=0.75]    (50,175.47) -- (224.93,175.46) ;
\draw [line width=0.75]    (50,108.39) -- (224.93,108.39) ;
\draw  [fill={rgb, 255:red, 244; green, 240; blue, 235 }  ,fill opacity=1 ][line width=0.75]  (66.72,105.02) .. controls (66.72,102.25) and (68.96,100) .. (71.73,100) -- (78.42,100) .. controls (81.19,100) and (83.44,102.25) .. (83.44,105.02) -- (83.44,145.29) .. controls (83.44,148.06) and (81.19,150.31) .. (78.42,150.31) -- (71.73,150.31) .. controls (68.96,150.31) and (66.72,148.06) .. (66.72,145.29) -- cycle ;
\draw [line width=0.75]    (50,242.54) -- (224.93,242.54) ;
\draw [line width=0.75]    (50,276.08) -- (224.93,276.07) ;
\draw [line width=0.75]    (50,209) -- (224.93,209) ;
\draw  [fill={rgb, 255:red, 244; green, 240; blue, 235 }  ,fill opacity=1 ][line width=0.75]  (91.93,138.56) .. controls (91.93,135.79) and (94.18,133.54) .. (96.95,133.54) -- (103.64,133.54) .. controls (106.41,133.54) and (108.65,135.79) .. (108.65,138.56) -- (108.65,178.83) .. controls (108.65,181.6) and (106.41,183.85) .. (103.64,183.85) -- (96.95,183.85) .. controls (94.18,183.85) and (91.93,181.6) .. (91.93,178.83) -- cycle ;
\draw  [fill={rgb, 255:red, 244; green, 240; blue, 235 }  ,fill opacity=1 ][line width=0.75]  (117.09,172.09) .. controls (117.09,169.32) and (119.34,167.08) .. (122.11,167.08) -- (128.8,167.08) .. controls (131.57,167.08) and (133.81,169.32) .. (133.81,172.09) -- (133.81,212.37) .. controls (133.81,215.14) and (131.57,217.38) .. (128.8,217.38) -- (122.11,217.38) .. controls (119.34,217.38) and (117.09,215.14) .. (117.09,212.37) -- cycle ;
\draw  [fill={rgb, 255:red, 244; green, 240; blue, 235 }  ,fill opacity=1 ][line width=0.75]  (142.25,205.63) .. controls (142.25,202.86) and (144.5,200.61) .. (147.27,200.61) -- (153.96,200.61) .. controls (156.73,200.61) and (158.97,202.86) .. (158.97,205.63) -- (158.97,245.9) .. controls (158.97,248.67) and (156.73,250.92) .. (153.96,250.92) -- (147.27,250.92) .. controls (144.5,250.92) and (142.25,248.67) .. (142.25,245.9) -- cycle ;
\draw  [fill={rgb, 255:red, 244; green, 240; blue, 235 }  ,fill opacity=1 ][line width=0.75]  (167.41,239.17) .. controls (167.41,236.4) and (169.66,234.15) .. (172.43,234.15) -- (179.12,234.15) .. controls (181.89,234.15) and (184.13,236.4) .. (184.13,239.17) -- (184.13,279.44) .. controls (184.13,282.21) and (181.89,284.46) .. (179.12,284.46) -- (172.43,284.46) .. controls (169.66,284.46) and (167.41,282.21) .. (167.41,279.44) -- cycle ;
\draw  [fill={rgb, 255:red, 244; green, 240; blue, 235 }  ,fill opacity=1 ][line width=0.75]  (117.09,105.02) .. controls (117.09,102.25) and (119.34,100) .. (122.11,100) -- (128.8,100) .. controls (131.57,100) and (133.81,102.25) .. (133.81,105.02) -- (133.81,145.29) .. controls (133.81,148.06) and (131.57,150.31) .. (128.8,150.31) -- (122.11,150.31) .. controls (119.34,150.31) and (117.09,148.06) .. (117.09,145.29) -- cycle ;
\draw  [fill={rgb, 255:red, 244; green, 240; blue, 235 }  ,fill opacity=1 ][line width=0.75]  (142.25,138.56) .. controls (142.25,135.79) and (144.5,133.54) .. (147.27,133.54) -- (153.96,133.54) .. controls (156.73,133.54) and (158.97,135.79) .. (158.97,138.56) -- (158.97,178.83) .. controls (158.97,181.6) and (156.73,183.85) .. (153.96,183.85) -- (147.27,183.85) .. controls (144.5,183.85) and (142.25,181.6) .. (142.25,178.83) -- cycle ;
\draw  [fill={rgb, 255:red, 244; green, 240; blue, 235 }  ,fill opacity=1 ][line width=0.75]  (167.41,172.09) .. controls (167.41,169.32) and (169.66,167.08) .. (172.43,167.08) -- (179.12,167.08) .. controls (181.89,167.08) and (184.13,169.32) .. (184.13,172.09) -- (184.13,212.37) .. controls (184.13,215.14) and (181.89,217.38) .. (179.12,217.38) -- (172.43,217.38) .. controls (169.66,217.38) and (167.41,215.14) .. (167.41,212.37) -- cycle ;
\draw  [fill={rgb, 255:red, 244; green, 240; blue, 235 }  ,fill opacity=1 ][line width=0.75]  (192.57,205.63) .. controls (192.57,202.86) and (194.82,200.61) .. (197.59,200.61) -- (204.28,200.61) .. controls (207.05,200.61) and (209.29,202.86) .. (209.29,205.63) -- (209.29,245.9) .. controls (209.29,248.67) and (207.05,250.92) .. (204.28,250.92) -- (197.59,250.92) .. controls (194.82,250.92) and (192.57,248.67) .. (192.57,245.9) -- cycle ;
\draw  [fill={rgb, 255:red, 244; green, 240; blue, 235 }  ,fill opacity=1 ][line width=0.75]  (66.77,172.11) .. controls (66.77,169.34) and (69.02,167.1) .. (71.79,167.1) -- (78.48,167.1) .. controls (81.25,167.1) and (83.49,169.34) .. (83.49,172.11) -- (83.49,212.39) .. controls (83.49,215.16) and (81.25,217.4) .. (78.48,217.4) -- (71.79,217.4) .. controls (69.02,217.4) and (66.77,215.16) .. (66.77,212.39) -- cycle ;
\draw  [fill={rgb, 255:red, 244; green, 240; blue, 235 }  ,fill opacity=1 ][line width=0.75]  (91.93,205.65) .. controls (91.93,202.88) and (94.18,200.63) .. (96.95,200.63) -- (103.64,200.63) .. controls (106.41,200.63) and (108.65,202.88) .. (108.65,205.65) -- (108.65,245.92) .. controls (108.65,248.69) and (106.41,250.94) .. (103.64,250.94) -- (96.95,250.94) .. controls (94.18,250.94) and (91.93,248.69) .. (91.93,245.92) -- cycle ;
\draw  [fill={rgb, 255:red, 244; green, 240; blue, 235 }  ,fill opacity=1 ][line width=0.75]  (66.77,239.18) .. controls (66.77,236.41) and (69.02,234.16) .. (71.79,234.16) -- (78.48,234.16) .. controls (81.25,234.16) and (83.49,236.41) .. (83.49,239.18) -- (83.49,279.45) .. controls (83.49,282.22) and (81.25,284.47) .. (78.48,284.47) -- (71.79,284.47) .. controls (69.02,284.47) and (66.77,282.22) .. (66.77,279.45) -- cycle ;
\draw  [fill={rgb, 255:red, 244; green, 240; blue, 235 }  ,fill opacity=1 ][line width=0.75]  (117.09,239.18) .. controls (117.09,236.41) and (119.34,234.16) .. (122.11,234.16) -- (128.8,234.16) .. controls (131.57,234.16) and (133.81,236.41) .. (133.81,239.18) -- (133.81,279.45) .. controls (133.81,282.22) and (131.57,284.47) .. (128.8,284.47) -- (122.11,284.47) .. controls (119.34,284.47) and (117.09,282.22) .. (117.09,279.45) -- cycle ;
\draw  [fill={rgb, 255:red, 244; green, 240; blue, 235 }  ,fill opacity=1 ][line width=0.75]  (167.41,105.01) .. controls (167.41,102.24) and (169.66,99.99) .. (172.43,99.99) -- (179.12,99.99) .. controls (181.89,99.99) and (184.13,102.24) .. (184.13,105.01) -- (184.13,145.28) .. controls (184.13,148.05) and (181.89,150.3) .. (179.12,150.3) -- (172.43,150.3) .. controls (169.66,150.3) and (167.41,148.05) .. (167.41,145.28) -- cycle ;
\draw  [fill={rgb, 255:red, 244; green, 240; blue, 235 }  ,fill opacity=1 ][line width=0.75]  (192.57,138.55) .. controls (192.57,135.78) and (194.82,133.53) .. (197.59,133.53) -- (204.28,133.53) .. controls (207.05,133.53) and (209.29,135.78) .. (209.29,138.55) -- (209.29,178.82) .. controls (209.29,181.59) and (207.05,183.84) .. (204.28,183.84) -- (197.59,183.84) .. controls (194.82,183.84) and (192.57,181.59) .. (192.57,178.82) -- cycle ;
\end{mytikz2}\quad.
\end{equation}
The above circuit diagram shows a brickwall circuit of $3$ layers on $6$ qubits, where a single ``block'' represents a grouped continuous series of elementary gates, denoted by $B_k$, which can be seen as the elementary unit when constructing a PQC. Two typical templates of two-qubit blocks are depicted as follows, 
\begin{equation}
\begin{mytikz2}
\draw  [fill={rgb, 255:red, 245; green, 240; blue, 235 }  ,fill opacity=0.37 ][dash pattern={on 3pt off 1.5pt}][line width=0.75]  (64.41,118.52) .. controls (64.41,108.84) and (72.25,101) .. (81.93,101) -- (378.27,101) .. controls (387.95,101) and (395.79,108.84) .. (395.79,118.52) -- (395.79,203.48) .. controls (395.79,213.16) and (387.95,221) .. (378.27,221) -- (81.93,221) .. controls (72.25,221) and (64.41,213.16) .. (64.41,203.48) -- cycle ;
\draw [line width=0.75]    (50,191) -- (410.2,191) ;
\draw [line width=0.75]    (50,131) -- (410.2,131) ;
\draw  [fill={rgb, 255:red, 245; green, 240; blue, 235 }  ,fill opacity=1 ][line width=0.75]  (90.5,111.4) .. controls (90.5,111.4) and (90.5,111.4) .. (90.5,111.4) -- (130.5,111.4) .. controls (130.5,111.4) and (130.5,111.4) .. (130.5,111.4) -- (130.5,151.4) .. controls (130.5,151.4) and (130.5,151.4) .. (130.5,151.4) -- (90.5,151.4) .. controls (90.5,151.4) and (90.5,151.4) .. (90.5,151.4) -- cycle ;
\draw  [fill={rgb, 255:red, 245; green, 240; blue, 235 }  ,fill opacity=1 ][line width=0.75]  (90.5,171.4) .. controls (90.5,171.4) and (90.5,171.4) .. (90.5,171.4) -- (130.5,171.4) .. controls (130.5,171.4) and (130.5,171.4) .. (130.5,171.4) -- (130.5,211.4) .. controls (130.5,211.4) and (130.5,211.4) .. (130.5,211.4) -- (90.5,211.4) .. controls (90.5,211.4) and (90.5,211.4) .. (90.5,211.4) -- cycle ;
\draw  [fill={rgb, 255:red, 245; green, 240; blue, 235 }  ,fill opacity=1 ][line width=0.75]  (150.5,111) .. controls (150.5,111) and (150.5,111) .. (150.5,111) -- (190.5,111) .. controls (190.5,111) and (190.5,111) .. (190.5,111) -- (190.5,211) .. controls (190.5,211) and (190.5,211) .. (190.5,211) -- (150.5,211) .. controls (150.5,211) and (150.5,211) .. (150.5,211) -- cycle ;
\draw  [fill={rgb, 255:red, 245; green, 240; blue, 235 }  ,fill opacity=1 ][line width=0.75]  (330,111) .. controls (330,111) and (330,111) .. (330,111) -- (370,111) .. controls (370,111) and (370,111) .. (370,111) -- (370,151) .. controls (370,151) and (370,151) .. (370,151) -- (330,151) .. controls (330,151) and (330,151) .. (330,151) -- cycle ;
\draw  [fill={rgb, 255:red, 245; green, 240; blue, 235 }  ,fill opacity=1 ][line width=0.75]  (330,171) .. controls (330,171) and (330,171) .. (330,171) -- (370,171) .. controls (370,171) and (370,171) .. (370,171) -- (370,211) .. controls (370,211) and (370,211) .. (370,211) -- (330,211) .. controls (330,211) and (330,211) .. (330,211) -- cycle ;
\draw  [fill={rgb, 255:red, 245; green, 240; blue, 235 }  ,fill opacity=1 ][line width=0.75]  (210.5,111) .. controls (210.5,111) and (210.5,111) .. (210.5,111) -- (250.5,111) .. controls (250.5,111) and (250.5,111) .. (250.5,111) -- (250.5,211) .. controls (250.5,211) and (250.5,211) .. (250.5,211) -- (210.5,211) .. controls (210.5,211) and (210.5,211) .. (210.5,211) -- cycle ;
\draw  [fill={rgb, 255:red, 245; green, 240; blue, 235 }  ,fill opacity=1 ][line width=0.75]  (270,111) .. controls (270,111) and (270,111) .. (270,111) -- (310,111) .. controls (310,111) and (310,111) .. (310,111) -- (310,211) .. controls (310,211) and (310,211) .. (310,211) -- (270,211) .. controls (270,211) and (270,211) .. (270,211) -- cycle ;

\draw (94,122) node [anchor=north west][inner sep=0.75pt]  [font=\normalsize] [align=left] {$\displaystyle R_{3}$};
\draw (94,182) node [anchor=north west][inner sep=0.75pt]  [font=\normalsize] [align=left] {$\displaystyle R_{3}$};
\draw (147,150) node [anchor=north west][inner sep=0.75pt]  [font=\normalsize] [align=left] {$\displaystyle R_{xx}$};
\draw (334,122) node [anchor=north west][inner sep=0.75pt]  [font=\normalsize] [align=left] {$\displaystyle R_{3}$};
\draw (334,182) node [anchor=north west][inner sep=0.75pt]  [font=\normalsize] [align=left] {$\displaystyle R_{3}$};
\draw (207,151.2) node [anchor=north west][inner sep=0.75pt]  [font=\normalsize] [align=left] {$\displaystyle R_{yy}$};
\draw (267,151.2) node [anchor=north west][inner sep=0.75pt]  [font=\normalsize] [align=left] {$\displaystyle R_{zz}$};
\end{mytikz2}\quad,\quad
\begin{mytikz2}
\draw  [fill={rgb, 255:red, 245; green, 240; blue, 235 }  ,fill opacity=0.37 ][dash pattern={on 3pt off 1.5pt}][line width=0.75]  (64.41,118.08) .. controls (64.41,108.4) and (72.25,100.56) .. (81.93,100.56) -- (182.61,100.56) .. controls (192.29,100.56) and (200.13,108.4) .. (200.13,118.08) -- (200.13,203.04) .. controls (200.13,212.71) and (192.29,220.56) .. (182.61,220.56) -- (81.93,220.56) .. controls (72.25,220.56) and (64.41,212.71) .. (64.41,203.04) -- cycle ;
\draw [line width=0.75]    (50,190.33) -- (216.13,190.33) ;
\draw [line width=0.75]    (50,130.33) -- (216.13,130.33) ;
\draw  [fill={rgb, 255:red, 245; green, 240; blue, 235 }  ,fill opacity=1 ][line width=0.75]  (90.5,110.73) .. controls (90.5,110.73) and (90.5,110.73) .. (90.5,110.73) -- (130.5,110.73) .. controls (130.5,110.73) and (130.5,110.73) .. (130.5,110.73) -- (130.5,150.73) .. controls (130.5,150.73) and (130.5,150.73) .. (130.5,150.73) -- (90.5,150.73) .. controls (90.5,150.73) and (90.5,150.73) .. (90.5,150.73) -- cycle ;
\draw  [fill={rgb, 255:red, 245; green, 240; blue, 235 }  ,fill opacity=1 ][line width=0.75]  (90.5,170.73) .. controls (90.5,170.73) and (90.5,170.73) .. (90.5,170.73) -- (130.5,170.73) .. controls (130.5,170.73) and (130.5,170.73) .. (130.5,170.73) -- (130.5,210.73) .. controls (130.5,210.73) and (130.5,210.73) .. (130.5,210.73) -- (90.5,210.73) .. controls (90.5,210.73) and (90.5,210.73) .. (90.5,210.73) -- cycle ;
\draw [line width=0.8]    (160.17,130.34) -- (160.17,199.7) ;
\draw  [line width=0.8]  (150.61,190.14) .. controls (150.61,184.85) and (154.89,180.57) .. (160.17,180.57) .. controls (165.45,180.57) and (169.73,184.85) .. (169.73,190.14) .. controls (169.73,195.42) and (165.45,199.7) .. (160.17,199.7) .. controls (154.89,199.7) and (150.61,195.42) .. (150.61,190.14) -- cycle ;
\draw  [fill={rgb, 255:red, 0; green, 0; blue, 0 }  ,fill opacity=1 ] (157.07,130.34) .. controls (157.07,128.63) and (158.46,127.24) .. (160.17,127.24) .. controls (161.88,127.24) and (163.27,128.63) .. (163.27,130.34) .. controls (163.27,132.05) and (161.88,133.44) .. (160.17,133.44) .. controls (158.46,133.44) and (157.07,132.05) .. (157.07,130.34) -- cycle ;
\draw [line width=0.8]    (150.61,190.14) -- (169.73,190.14) ;

\draw (94,120) node [anchor=north west][inner sep=0.75pt]  [font=\normalsize] [align=left] {$\displaystyle R_{y}$};
\draw (94,180) node [anchor=north west][inner sep=0.75pt]  [font=\normalsize] [align=left] {$\displaystyle R_{y}$};
\end{mytikz2}\quad,
\label{eq:ansatz_template}
\end{equation}
which is named ``$\mathrm{SU}(4)$ Cartan decomposition'' and ``$R_y$-$\opr{CX}$'', respectively. The former is a fully parametrized two-qubit block that can express any two-qubit unitary. Here $R_{xx}$, $R_{yy}$ and $R_{zz}$ represent the two-qubit Pauli rotation gates with generators $X\otimes X$, $Y\otimes Y$ and $Z\otimes Z$, respectively. $R_3$ represents a universal single-qubit gate, e.g., $R_zR_yR_z$ where $R_{y}$ and $R_{z}$ represent the single-qubit Pauli rotation gates with generators $Y$ and $Z$. Other block templates can be defined and named similarly. Here we clarify that the spatial dimension of a quantum circuit actually refers to the dimension of the underlying qubit connectivity lattice, where the existence of gate locality is assumed.

A standard workflow of VQE involves running the PQC on a quantum device, measuring the cost, and updating the tunable parameters iteratively to minimize the cost function using classical optimizers such as gradient descent algorithms. The cost function in VQE is just the energy expectation value with respect to the output state
\begin{equation}
    C(\bm{\theta})=\braoprket{\psi(\bm{\theta})}{H}{\psi(\bm{\theta})}.
\end{equation}
The core idea of VQE is very clear and natural because a quantum circuit is a quite natural way to characterize and implement a quantum state. However, VQE and other VQAs still face great challenges on the way to practical quantum advantages because of their trainability, and the tradeoff between their trainability and expressibility. To be specific, it is known that a variational algorithm can learn the true ground state successfully if two conditions are met: the expressibility of the ansatz is strong enough to contain the ground state, and the trainability is good enough to navigate smoothly toward the ground state. However, in the practice of VQAs, these two conditions seem not easy to be met simultaneously when the system size is large as compared to classical neural networks. 

\subsection{Barren plateaus}
To date, the most extensively studied trainability issue in VQAs is the so-called barren plateau (BP) phenomenon~\cite{McClean2018, Cerezo2021}, which means that the cost gradient would vanish exponentially with the system size for deep PQCs under certain conditions, i.e.,
\begin{equation}
    \mathbb{E}_{\Theta}[\partial_\mu C]=0,\quad \var_{\Theta}[\partial_\mu C] \in O(b^{-N}),
\end{equation}
where $\partial_\mu C=\frac{\partial C}{\partial \theta_\mu}$ and $b>1$. $\mathbb{E}_{\Theta}[\cdot]$ and $\var_{\Theta}[\cdot]$ represent the average and variance taken over the ensemble $\Theta$ that is defined in the tunable parameter space. This is akin to the vanishing gradient issue in classical neural networks. The mechanism for the BP phenomenon can be explained as follows. At the beginning of the VQE workflow, the tunable parameters are often randomly initialized so that the whole parameter space can be explored uniformly in the probabilistic sense to search for the optimum. This random initialization procedure gives rise to a probability measure that makes the PQC $\mathbf{U}(\bm{\theta})$ become a random quantum circuit (RQC). An RQC in a generic shape like brickwall will form an approximate unitary $2$-design if the RQC is sufficiently deep to achieve linear depth such as $2\times N$~\cite{Brandao2016}. Here a $t$-design of a certain group, such as the unitary group $\mathrm{U}(d)$, refers to an ensemble of group elements that can lead to the same $t$-degree moment as that of the Haar (uniform) distribution over the group. Thus one can obtain the results for $2$-design integrals by using the Weingarten formulas for Haar integrals, such as 
\begin{equation}\label{eq:twirling_1st_order}
    \int d\mu(U)~ U^{\dagger} (\cdot) U = \frac{1}{d}\tr(\cdot) I, 
\end{equation}
\begin{equation}\label{eq:twirling_2st_order}
\begin{aligned}
    \int d\mu(U)~ U^{\dagger\otimes 2} (\cdot) U^{\otimes 2} =& \frac{1}{d^{2}-1}\left[\tr(\cdot) I + \tr(S \cdot) S \right]
    - \frac{1}{d(d^{2}-1)}\left[\tr(\cdot) S + \tr(S \cdot) I \right],
\end{aligned}
\end{equation}
where $\mu(U)$ is the Haar measure and ``$\cdot$'' is the placeholder for an arbitrary linear operator. $I$ is the identity, and $S$ is the swap operator on the two replicas. If the RQC forms a global $2$-design, i.e., a $2$-design over the entire Hilbert space of the whole system, the factor $d^2$ in the denominator will become $4^{N}$ and result in an exponential vanishing variance of cost derivatives. This is the reason for the original BP phenomenon, i.e., a general linear depth RQC is enough to explore the exponential large Hilbert space in the sense of the second moment, resulting in the exponential concentration of the expectation value of an observable. This indicates that a general linear depth circuit (GLDC) is too powerful to be used as a variational ansatz running on quantum devices. Fortunately, based on the above discussion regarding the entanglement area law, GLDC seems overkill for the ground-state preparation task as GLDC typically generates volume-law states. Utilizing this area-law feature, one can design a new circuit class called finite local-depth circuit (FLDC), which is proved to be free from BP for local Hamiltonians~\cite{Zhang2024}, and at the same time is generally hard to simulate classically by tensor network methods in two and higher dimensions. This creates new possibilities for practical quantum advantages in VQE.

\subsection{Bad local minima}
However, BP is not the only trainability issue in VQAs. Another notorious issue is the local minimum problem, which is thought to exist extensively in moderately shallow circuits~\cite{Anschuetz2022}. It is worth noticing that the scenario that one really needs to worry about is the landscape with so-called bad local minima, where the cost values of local minima concentrate far from that of the global minimum. By contrast, the landscape with good local minima refers to the case where the cost values of local minima are close to that of the global minimum. The existence of good local minima does not need to be overly concerning, as in physical systems, these good local minima might correspond to some physically meaningful metastable low-energy states compared to the true ground state, e.g., the competing ordered states in strongly correlated systems such as high-temperature superconductors~\cite{Zheng2017, Qin2020, Xu2024}. On the other hand, the existence of bad local minima is a severe trainability issue that needs to be fixed. Under certain conditions, it is proved~\cite{Anschuetz2022} that there is a super-polynomially small fraction of local minima within any constant additive error of the ground state energy when the so-called local overparameterization ratio $\gamma$ is much smaller than $1$. Here the local overparameterization ratio is defined by 
\begin{equation}
    \gamma=\frac{M_{\text{cone}}}{2d_{\text{cone}}},    
\end{equation}
where $M_{\text{cone}}$ is the number of tunable parameters in the circuit within the backward causal cone of a certain observable and $d_{\text{cone}}$ is the Hilbert space dimension of the set of qubits within the backward causal cone. The backward causal cone refers to the set of gates in a circuit that has an influence on the expectation value of the observable, e.g.,
\begin{equation}
\begin{mytikz2}
\draw [line width=0.75]    (50.5,84.32) -- (178.24,84.53) ;
\draw [line width=0.75]    (50.48,117.86) -- (178.22,118.07) ;
\draw [line width=0.75]    (50.52,50.78) -- (178.27,50.99) ;
\draw  [fill={rgb, 255:red, 245; green, 240; blue, 235 }  ,fill opacity=1 ][line width=0.75]  (67.31,47.43) .. controls (67.31,44.66) and (69.56,42.42) .. (72.33,42.42) -- (79.02,42.43) .. controls (81.79,42.43) and (84.03,44.68) .. (84.03,47.45) -- (83.98,87.72) .. controls (83.98,90.49) and (81.73,92.74) .. (78.96,92.73) -- (72.28,92.73) .. controls (69.51,92.72) and (67.26,90.48) .. (67.27,87.71) -- cycle ;
\draw [line width=0.75]    (50.43,185.33) -- (178.18,185.54) ;
\draw [line width=0.75]    (50.41,218.87) -- (178.15,219.08) ;
\draw [line width=0.75]    (50.46,151.39) -- (178.2,151.6) ;
\draw  [fill={rgb, 255:red, 245; green, 240; blue, 235 }  ,fill opacity=1 ][line width=0.75]  (92.49,80.99) .. controls (92.49,78.22) and (94.74,75.98) .. (97.51,75.98) -- (104.2,75.99) .. controls (106.97,75.99) and (109.21,78.24) .. (109.21,81.01) -- (109.17,121.29) .. controls (109.16,124.06) and (106.91,126.3) .. (104.15,126.3) -- (97.46,126.29) .. controls (94.69,126.29) and (92.45,124.04) .. (92.45,121.27) -- cycle ;
\draw  [fill={rgb, 255:red, 245; green, 240; blue, 235 }  ,fill opacity=1 ][line width=0.75]  (67.29,114.53) .. controls (67.3,111.76) and (69.55,109.51) .. (72.32,109.52) -- (79,109.52) .. controls (81.77,109.53) and (84.02,111.77) .. (84.01,114.54) -- (83.97,154.82) .. controls (83.97,157.59) and (81.72,159.83) .. (78.95,159.83) -- (72.26,159.82) .. controls (69.49,159.82) and (67.25,157.57) .. (67.25,154.8) -- cycle ;
\draw  [fill={rgb, 255:red, 245; green, 240; blue, 235 }  ,fill opacity=1 ][line width=0.75]  (92.42,147.59) .. controls (92.42,144.82) and (94.67,142.58) .. (97.44,142.58) -- (104.13,142.59) .. controls (106.9,142.59) and (109.14,144.84) .. (109.14,147.61) -- (109.1,187.88) .. controls (109.09,190.65) and (106.84,192.89) .. (104.07,192.89) -- (97.39,192.88) .. controls (94.62,192.88) and (92.37,190.63) .. (92.38,187.86) -- cycle ;
\draw  [fill={rgb, 255:red, 245; green, 240; blue, 235 }  ,fill opacity=1 ][line width=0.75]  (67.22,181.09) .. controls (67.23,178.32) and (69.48,176.08) .. (72.25,176.08) -- (78.93,176.09) .. controls (81.7,176.09) and (83.95,178.34) .. (83.94,181.11) -- (83.9,221.38) .. controls (83.9,224.15) and (81.65,226.39) .. (78.88,226.39) -- (72.19,226.38) .. controls (69.42,226.38) and (67.18,224.13) .. (67.18,221.36) -- cycle ;
\draw [line width=0.75]    (50.5,252.76) -- (178.24,252.97) ;
\draw [line width=0.75]    (50.47,286.3) -- (178.22,286.51) ;
\draw  [fill={rgb, 255:red, 245; green, 240; blue, 235 }  ,fill opacity=1 ][line width=0.75]  (92.52,215.02) .. controls (92.52,212.25) and (94.77,210) .. (97.54,210.01) -- (104.22,210.01) .. controls (106.99,210.02) and (109.24,212.26) .. (109.23,215.03) -- (109.19,255.31) .. controls (109.19,258.08) and (106.94,260.32) .. (104.17,260.32) -- (97.48,260.31) .. controls (94.71,260.31) and (92.47,258.06) .. (92.47,255.29) -- cycle ;
\draw  [fill={rgb, 255:red, 245; green, 240; blue, 235 }  ,fill opacity=1 ][line width=0.75]  (67.32,249.02) .. controls (67.32,246.25) and (69.57,244) .. (72.34,244.01) -- (79.03,244.01) .. controls (81.8,244.02) and (84.04,246.27) .. (84.04,249.04) -- (84,289.31) .. controls (83.99,292.08) and (81.74,294.32) .. (78.97,294.32) -- (72.29,294.31) .. controls (69.52,294.31) and (67.27,292.06) .. (67.28,289.29) -- cycle ;
\draw  [fill={rgb, 255:red, 245; green, 166; blue, 35 }  ,fill opacity=1 ][line width=0.75]  (176.52,210.84) .. controls (176.52,210.84) and (176.52,210.84) .. (176.52,210.84) -- (193.24,210.85) .. controls (193.24,210.85) and (193.24,210.85) .. (193.24,210.85) -- (193.22,228.27) .. controls (193.22,228.27) and (193.22,228.27) .. (193.22,228.27) -- (176.5,228.26) .. controls (176.5,228.26) and (176.5,228.26) .. (176.5,228.26) -- cycle ;
\draw [line width=0.75]    (50.48,319.36) -- (178.22,319.57) ;
\draw [line width=0.75]    (50.46,352.89) -- (178.2,353.1) ;
\draw  [fill={rgb, 255:red, 245; green, 240; blue, 235 }  ,fill opacity=1 ][line width=0.75]  (92.49,282.49) .. controls (92.49,279.72) and (94.74,277.48) .. (97.51,277.48) -- (104.2,277.49) .. controls (106.97,277.49) and (109.21,279.74) .. (109.21,282.51) -- (109.17,322.79) .. controls (109.16,325.56) and (106.91,327.8) .. (104.15,327.8) -- (97.46,327.79) .. controls (94.69,327.79) and (92.45,325.54) .. (92.45,322.77) -- cycle ;
\draw  [fill={rgb, 255:red, 245; green, 240; blue, 235 }  ,fill opacity=1 ][line width=0.75]  (67.29,316.03) .. controls (67.3,313.26) and (69.55,311.01) .. (72.32,311.02) -- (79,311.02) .. controls (81.77,311.03) and (84.02,313.27) .. (84.01,316.04) -- (83.97,356.32) .. controls (83.97,359.09) and (81.72,361.33) .. (78.95,361.33) -- (72.26,361.32) .. controls (69.49,361.32) and (67.25,359.07) .. (67.25,356.3) -- cycle ;
\draw  [fill={rgb, 255:red, 245; green, 240; blue, 235 }  ,fill opacity=1 ][line width=0.75]  (118.81,47.43) .. controls (118.81,44.66) and (121.06,42.42) .. (123.83,42.42) -- (130.52,42.43) .. controls (133.29,42.43) and (135.53,44.68) .. (135.53,47.45) -- (135.48,87.72) .. controls (135.48,90.49) and (133.23,92.74) .. (130.46,92.73) -- (123.78,92.73) .. controls (121.01,92.72) and (118.76,90.48) .. (118.77,87.71) -- cycle ;
\draw  [fill={rgb, 255:red, 245; green, 240; blue, 235 }  ,fill opacity=1 ][line width=0.75]  (143.99,80.99) .. controls (143.99,78.22) and (146.24,75.98) .. (149.01,75.98) -- (155.7,75.99) .. controls (158.47,75.99) and (160.71,78.24) .. (160.71,81.01) -- (160.67,121.29) .. controls (160.66,124.06) and (158.41,126.3) .. (155.65,126.3) -- (148.96,126.29) .. controls (146.19,126.29) and (143.95,124.04) .. (143.95,121.27) -- cycle ;
\draw  [fill={rgb, 255:red, 245; green, 240; blue, 235 }  ,fill opacity=1 ][line width=0.75]  (118.79,114.53) .. controls (118.8,111.76) and (121.05,109.51) .. (123.82,109.52) -- (130.5,109.52) .. controls (133.27,109.53) and (135.52,111.77) .. (135.51,114.54) -- (135.47,154.82) .. controls (135.47,157.59) and (133.22,159.83) .. (130.45,159.83) -- (123.76,159.82) .. controls (120.99,159.82) and (118.75,157.57) .. (118.75,154.8) -- cycle ;
\draw  [fill={rgb, 255:red, 245; green, 240; blue, 235 }  ,fill opacity=1 ][line width=0.75]  (143.92,147.59) .. controls (143.92,144.82) and (146.17,142.58) .. (148.94,142.58) -- (155.63,142.59) .. controls (158.4,142.59) and (160.64,144.84) .. (160.64,147.61) -- (160.6,187.88) .. controls (160.59,190.65) and (158.34,192.89) .. (155.57,192.89) -- (148.89,192.88) .. controls (146.12,192.88) and (143.87,190.63) .. (143.88,187.86) -- cycle ;
\draw  [fill={rgb, 255:red, 245; green, 240; blue, 235 }  ,fill opacity=1 ][line width=0.75]  (118.72,181.09) .. controls (118.73,178.32) and (120.98,176.08) .. (123.75,176.08) -- (130.43,176.09) .. controls (133.2,176.09) and (135.45,178.34) .. (135.44,181.11) -- (135.4,221.38) .. controls (135.4,224.15) and (133.15,226.39) .. (130.38,226.39) -- (123.69,226.38) .. controls (120.92,226.38) and (118.68,224.13) .. (118.68,221.36) -- cycle ;
\draw  [fill={rgb, 255:red, 245; green, 240; blue, 235 }  ,fill opacity=1 ][line width=0.75]  (144.02,215.02) .. controls (144.02,212.25) and (146.27,210) .. (149.04,210.01) -- (155.72,210.01) .. controls (158.49,210.02) and (160.74,212.26) .. (160.73,215.03) -- (160.69,255.31) .. controls (160.69,258.08) and (158.44,260.32) .. (155.67,260.32) -- (148.98,260.31) .. controls (146.21,260.31) and (143.97,258.06) .. (143.97,255.29) -- cycle ;
\draw  [fill={rgb, 255:red, 245; green, 240; blue, 235 }  ,fill opacity=1 ][line width=0.75]  (118.82,249.02) .. controls (118.82,246.25) and (121.07,244) .. (123.84,244.01) -- (130.53,244.01) .. controls (133.3,244.02) and (135.54,246.27) .. (135.54,249.04) -- (135.5,289.31) .. controls (135.49,292.08) and (133.24,294.32) .. (130.47,294.32) -- (123.79,294.31) .. controls (121.02,294.31) and (118.77,292.06) .. (118.78,289.29) -- cycle ;
\draw  [fill={rgb, 255:red, 245; green, 240; blue, 235 }  ,fill opacity=1 ][line width=0.75]  (143.99,282.49) .. controls (143.99,279.72) and (146.24,277.48) .. (149.01,277.48) -- (155.7,277.49) .. controls (158.47,277.49) and (160.71,279.74) .. (160.71,282.51) -- (160.67,322.79) .. controls (160.66,325.56) and (158.41,327.8) .. (155.65,327.8) -- (148.96,327.79) .. controls (146.19,327.79) and (143.95,325.54) .. (143.95,322.77) -- cycle ;
\draw  [fill={rgb, 255:red, 245; green, 240; blue, 235 }  ,fill opacity=1 ][line width=0.75]  (118.79,316.03) .. controls (118.8,313.26) and (121.05,311.01) .. (123.82,311.02) -- (130.5,311.02) .. controls (133.27,311.03) and (135.52,313.27) .. (135.51,316.04) -- (135.47,356.32) .. controls (135.47,359.09) and (133.22,361.33) .. (130.45,361.33) -- (123.76,361.32) .. controls (120.99,361.32) and (118.75,359.07) .. (118.75,356.3) -- cycle ;
\draw  [color={rgb, 255:red, 0; green, 0; blue, 0 }  ,draw opacity=0.16 ][fill={rgb, 255:red, 245; green, 166; blue, 35 }  ,fill opacity=0.2 ][dash pattern={on 4.5pt off 4.5pt}] (61.42,85.86) -- (168.99,219.27) -- (168.85,251.56) -- (60.12,384.02) -- cycle ;
\end{mytikz2}\quad,
\end{equation}
where the orange square represents the location of the observable and the shaded area represents the causal cone. For a 1-dimensional brickwall circuit of depth $D$ and a single-qubit observable, the ratio can be roughly estimated as
\begin{equation}
    \gamma = \frac{2D(2D+1)/2\cdot (d^{2\beta}-1)}{2d^{2D\beta}}\approx D(D+\frac{1}{2})d^{2\beta(1-D)},
\end{equation}
where $\beta$ denotes the block size, i.e., the number of qubits acted by each block, and $d$ denotes the local Hilbert space dimension with $d=2$ for qubit systems. For very shallow circuits such as a 1-dimensional brickwall circuit of depth $D=1$, it holds that $\gamma \geq 1$, which means that the circuit is overparametrized locally and hence bad local minima are rare. On the other hand, for constant depth circuits of moderately larger depth $D$, the ratio decays rapidly to $\gamma\ll 1$, resulting in local underparametrization, and hence bad local minima become extensive for large-size systems. In other words, there exists a critical depth $D_c$ such that when increasing the depth of constant depth circuits to $D> D_c$, the training landscape undergoes a computational phase transition (or crossover) from those with good local minima to those with bad local minima, which would continue to be more severe for logarithmic depth circuits such as quantum convolutional neural network (QCNN) and MERA. Nevertheless, we point out that if the support of each subterm in the Hamiltonian is just one qubit, there is even no need to consider a 1D brickwall circuit of depth $D=1$, i.e., this is a trivial polarization Hamiltonian so that one layer of single-qubit rotations is enough to solve the ground state. In other words, as long as we consider a non-trivial Hamiltonian including extensive subterms with support larger than two qubits, even a 1D brickwall circuit of depth $D=1$ can be locally underparametrized which will lead to bad local minima for large-size systems.

Here, we briefly explain the reason why overparametrization can eliminate local minima. An intuitive explanation is that when the number of tunable parameters is equal to or larger than the entire Hilbert space dimension, then for a certain parameter point away from the global minimum, there will always exist a direction that allows the cost value to decrease further, which corresponds to the rotations from the current state to the ground state. Consequently, the local minima in the underparametrization case will become saddle points as the parameter space enlarges, leading to the elimination of local minima. Rigorously, for PQCs consisting of repeated layers, overparametrization can be defined by the saturate rank of the quantum Fisher information (QFI) matrix~\cite{Larocca2023}, though the overparametrization there is defined with respect to the maximally reachable sub-manifold instead of the entire Hilbert space. In general, an exponential number of parameters is required to achieve overparametrization. For some PQCs with special gates forming a dynamical Lie algebra (DLA) of polynomially large dimensions~\cite{Larocca2022}, it can be shown that only a polynomial number of parameters is required to achieve overparametrization~\cite{Larocca2023}, such as the PQC consisting of only $e^{-iX_j\theta/2}$ and $e^{-iZ_{j} Z_{j+1}\theta/2}$ gates.

\subsection{Insufficient expressibility}
Based on the above discussion, one can see that BP and bad local minima impose strong constraints on the expressive power of trainable PQCs. On the other hand, if the expressive power of the PQC is too weak to include the true ground state of the given Hamiltonian, the PQC will also fail to learn the ground state successfully. It is worth emphasizing that the expressibility mentioned here is not simply the size of the expressive space itself, as measured by metrics like frame potential~\cite{Holmes2021}. Instead, it depends on specific Hamiltonians, referring to the closeness of the lowest energy state within its expressive space to the true ground state as quantified by the so-called reachability deficit~\cite{Bharti2022}
\begin{equation}\label{eq:deficit}
    \Delta_{\text{r}} = \min_{\bm{\theta}}C(\bm{\theta}) - E_{\text{G}},
\end{equation}
where $E_{\text{G}}$ denotes the true ground state energy of the given Hamiltonian. By definition, it ensures that $\Delta_{\text{r}} \geq 0$ and the equality holds if and only if the expressive space contains the true ground state. Notably, $\Delta_{\text{r}} =0$ only indicates that the ground state is within the expressive space, while it does not convey any information about trainability, i.e., it does not guarantee an efficient and smooth optimization path from the initial state to the ground state. For example, suppose the expressive power of a PQC is strong enough such as a general polynomial-depth circuit. In that case, one can almost always get $\Delta=0$ but the PQC still cannot learn the ground state successfully due to trainability issues.

However, even with such severe constraints from BP, bad local minima, and insufficient expressibility, there is still a possibility for using quantum circuits as variational ansatzes to solve problems. The possibility comes from utilizing the more detailed prior information on the given task to devise the so-called Hamiltonian informed ansatzes or called problem-inspired ansatzes or task-specific ansatzes. For example, based on the prior information that ground states of locally interacting quantum many-body systems usually satisfy the entanglement area law with at most logarithmic corrections, one can use FLDC or logarithmic local-depth circuits to avoid barren plateaus~\cite{Zhang2024}, which also satisfy the entanglement area law under the condition of local circuits and at the same time cannot be efficiently simulated by classical tensor network methods in two and higher dimensions. On the other hand, one can increase the expressibility by introducing ancilla qubits and increasing the block size (like increasing the bond dimension in MPS) instead of directly appending repeated layers that may promote local underparametrization with bad local minima. As another example, one can use special block templates that depend on the given Hamiltonian instead of universal blocks (such as $\mathrm{SU}(4)$ Cartan decomposition) to construct the PQC, e.g., Hamiltonian variational ansatzes (HVA)~\cite{Wiersema2020} like in the quantum approximate optimization algorithm (QAOA). This is because the theoretical results mentioned above on BP and bad local minima are based on the assumption of local scrambled gates. More broadly, various heuristic algorithms based on the prior knowledge of a specific task, e.g., a clever initialization, always have the potential to bypass the above theoretical limitations.
\vspace{10pt}

In summary, the above discussion inspires us to design specific PQCs (e.g., gate types, circuit depth, qubit connectivity, circuit architectures, etc.) tailored for specific problems instead of attempting to use a universal PQC to solve all problems, so that the expressive space can just contain the target despite it is small. This necessitates an efficient method to diagnose which PQCs are suitable for which problems, i.e., to identify the learnability of a PQC regarding a given target.

Here we clarify the term ``\textbf{learnability}'' used throughout the paper. We define learnability as the ability of a PQC to learn a given target state, which can be quantified by the training error or test error. For example, in the context of VQE which is also the focus of this paper, the learnability can be quantified by the training error, i.e., the gap between the converged energy $E_{\text{conv}}$ in training and the true ground state energy $E_{\text{G}}$, i.e.,
\begin{equation}\label{eq:error}
    \epsilon = E_{\text{conv}} - E_{\text{G}},
\end{equation}
or the fidelity between the PQC output state after convergence and the true ground state $\left|\braket{\psi_{\text{conv}}}{G}\right|$. Note that the difference between Eq.\,\eqref{eq:error} and Eq.\,\eqref{eq:deficit} is that Eq.\,\eqref{eq:deficit} does not care about whether the minimization process can be efficiently achieved by training. In these contexts, learnability can be simply regarded as the combination of trainability and expressibility. That is to say, the learnability of a given PQC and a Hamiltonian is high if and only if the trainability and expressibility are both high. However, in other contexts such as quantum machine learning (QML) where the information of the learning target is encoded in the initial state instead of the Hamiltonian like in VQE, learnability also depends on another factor besides trainability and expressibility, i.e., the number of samples. That is to say, even with trainable and expressive PQCs, if the sample size of the given dataset is so small that overfitting occurs and results in a large generalization error~\cite{Du2021}, the learnability is still low. Also, here we only consider ideal quantum devices where the influences of gate noise and measurement error on learnability are not considered, which we leave to future research.

A direct method for identifying learnability is to run an actual training procedure and use the training error to assess learnability. Nevertheless, this method is inefficient, heavily depends on the choice of classical optimizers and sometimes it is hard to get the exact ground state for benchmarks. (Although the learnability quantified by Eq.\,\eqref{eq:error} also depends on the choice of classical optimizers, here we only care about its dependence on the design of PQCs, or say the precise characterization of learnability here is the minimum of average case of the error over all possible optimizers.) This motivates us to seek an efficient criterion to assess the learnability of a PQC regarding a given target, leading to the study of the relative fluctuation (RF) in this work.

\section{Theorem and proof}\label{appendix:proof}
In this section, we provide a detailed statement and proof of the theorem shown in the main text. The basic idea is similar to Ref.~\cite{Zhang2024} of one of the authors, i.e., tracking the back evolution of the local Pauli strings along different paths across the circuit. One can refer to Ref.~\cite{Zhang2024} for more pedagogical explanations.

We first clarify some of our notations in advance. We use $M$ to denote the total number of tunable parameters and $M'$ to denote the total number of blocks. We suppose that the initial state for the PQC is always the zero state $\rho_0=\ketbrasame{\bm{0}}$. Our results below can be easily generalized to any other product state. We use the Schatten $p$-norm for linear operators and the vector $l$-norm for vectors. We use $h\vert_s$ to denote the sub-string of a Pauli string $h$ on the qubit subset $s$, e.g., $(X\otimes I\otimes Y)\vert_{\{q_1,q_2\}}=X\otimes I$. We denote the $t$-degree twirling channel on support $s$ as
\begin{equation}\label{eq:twirling_def}
    \mathcal{T}_{s}^{(t)}(\cdot) = \int d\mu_s(U)~ U^{\dagger\otimes t} (\cdot) U^{\otimes t},
\end{equation}
where $\mu_s$ is the Haar measure over the unitary group $\mathrm{U}(2^{|s|})$ acting on $s$. We use $\mathcal{P}\vert_s$ to denote the set of all Pauli strings on $s$, and $\mathcal{P}'\vert_s$ to denote the set of all non-trivial Pauli strings on $s$, i.e., $\mathcal{P}'\vert_s=\mathcal{P}\vert_s- \{I\vert_s\}$.

The main assumption here is that the blocks in the PQC are locally scrambled independently. That is to say, the ensemble of randomly initialized parameters within each block of the PQC can induce a unitary ensemble forming a $2$-design on its local support so that one can calculate the average and variance by the Weingarten calculus. This can be approximately realized by uniform-randomly initialized blocks with full parametrization, e.g., the $\mathrm{SU}(4)$ Cartan decomposition. We denote this circuit ensemble composed of independent $2$-design blocks as $\mathbb{U}$. Using blocks as the elementary unit, the PQC can be rewritten as
\begin{equation}\label{eq:u_bbb}
    \mathbf{U} = \prod_{k}^{\leftarrow} B_k,
\end{equation}
where the dependence on $\bm{\theta}$ is omitted for simplicity. Hence the average and variance of the cost function can be expressed as
\begin{align}
    \E_\mathbb{U}\left[ C\right] &= \tr\left[\rho_0 \mathcal{T}^{(1)}_{s(B_1)}\circ\mathcal{T}^{(1)}_{s(B_2)} \circ\cdots\circ \mathcal{T}^{(1)}_{s(B_{M'})}(H) \right], \label{eq:exp_expression_by_twirling}\\
    \var_\mathbb{U}\left[ C\right] &= \tr\left[\rho_0^{\otimes 2} \mathcal{T}^{(2)}_{s(B_1)}\circ\mathcal{T}^{(2)}_{s(B_2)} \circ\cdots\circ \mathcal{T}^{(2)}_{s(B_{M'})}(H^{\otimes 2}) \right] - \left(\E_\mathbb{U}\left[ C\right]\right)^2.
    \label{eq:var_expression_by_twirling}
\end{align}
This can be seen as a back-evolution process for the Hamiltonian by a sequence of twirling channels, like a ``stochastic Heisenberg picture''. We assume the Hamiltonian is always traceless $\tr(H)=0$ since the identity subterm only contributes to a constant to the energy expectation. Moreover, we assume $s(\mathbf{U}) \supseteq s(H)$ without loss of generality. Under these assumptions, Eq.\,\eqref{eq:twirling_1st_order} directly leads to $\E_\mathbb{U}\left[ C\right]=0$. Our main goal in the following is to estimate the first term in Eq.\,\eqref{eq:var_expression_by_twirling}. Our main tool is the following lemma for each single twirling channel and each Pauli string subterm.

\begin{lemma}\label{lemma:pauli_string_copies_one_step_evolution}
Suppose that $h_1$ and $h_2$ are two Pauli strings that are distinct on and only on $s'$. If $s'\cap s=\varnothing$ and $h\vert_s$ is non-trivial, then it holds that
\begin{equation}\label{eq:twirling_h2}
    \mathcal{T}^{(2)}_{s} (h_1\otimes h_2) = \frac{1}{2^{|2s|}-1} \sum_{\sigma\in\mathcal{P}'\vert_s} \sigma^{\otimes 2} \otimes ( h_1\vert_{\bar{s}}\otimes h_2\vert_{\bar{s}} ),
\end{equation}
where $\bar{s}$ is the complement of $s$. If $s'\cap s=\varnothing$ and $h\vert_s$ is trivial, then $\mathcal{T}^{(2)}_{s} (h_1\otimes h_2)= h_1\otimes h_2$. If $s'\cap s\neq\varnothing$, then $\mathcal{T}^{(2)}_s(h_1\otimes h_2)=0$.
\end{lemma}
\begin{proof}
This can be simply proved by applying Eq.\,\eqref{eq:twirling_2st_order} and substituting the Pauli decomposition of the SWAP operator.
\end{proof}

This lemma means that a single step of back evolution of a doubled Pauli string on $s$ will map the sub-string inside $s$ to a uniformly weighted sum of all non-trivial Pauli strings on $s$ while keeping the part outside $s$ unchanged. This lemma gives us the key point of how to deal with the sequence of twirling channels, i.e., tracking the back evolution of doubled Pauli strings. Note that the cross-terms in $H^{\otimes2} = \sum_{ij} \lambda_i \lambda_j h_i\otimes h_j$ will all be eliminated, i.e.,
\begin{equation}\label{eq:var_no_cross_terms}
    \var_\mathbb{U}[C] = \sum_j \lambda_j^2 \var_\mathbb{U}[ \avg{h_j}],
\end{equation}
where $\avg{\cdot}=\braoprket{\bm{0}}{\mathbf{U}^\dagger(\cdot)\mathbf{U}}{\bm{0}}$. As the coefficients $\lambda_j^2$ are all non-negative, a lower bound can be obtained by bounding the variance corresponding to each subterm respectively. 

In order to deal with circuits in various shapes instead of only the brickwall circuit, we introduce the notion of ``path'' in quantum circuits. A path $p$ in a circuit $\mathbf{U}$ is defined as a time-ordered sequence of blocks $p=\{B_{k_1}, B_{k_2},  B_{k_3}, \cdots\}$ with $k_1<k_2<k_3<\cdots$ where adjacent blocks are connected in the circuit diagram. Every path should traverse the entire circuit and connect to the given observable, as illustrated below
\begin{equation}
\begin{mytikz2}
\draw [line width=0.75]    (50.67,92.46) -- (178.41,92.67) ;
\draw [line width=0.75]    (50.65,126.01) -- (178.39,126.21) ;
\draw [line width=0.75]    (50.69,58.93) -- (178.43,59.14) ;
\draw  [fill={rgb, 255:red, 245; green, 240; blue, 235 }  ,fill opacity=1 ][line width=0.75]  (67.48,55.58) .. controls (67.48,52.81) and (69.73,50.56) .. (72.5,50.57) -- (79.18,50.57) .. controls (81.95,50.58) and (84.2,52.82) .. (84.19,55.59) -- (84.15,95.87) .. controls (84.15,98.64) and (81.9,100.88) .. (79.13,100.88) -- (72.44,100.87) .. controls (69.67,100.87) and (67.43,98.62) .. (67.43,95.85) -- cycle ;
\draw [line width=0.75]    (50.6,193.48) -- (178.34,193.69) ;
\draw [line width=0.75]    (50.58,227.02) -- (178.32,227.22) ;
\draw [line width=0.75]    (50.62,159.54) -- (178.37,159.75) ;
\draw  [fill={rgb, 255:red, 255; green, 234; blue, 195 }  ,fill opacity=1 ][line width=0.75]  (92.66,89.14) .. controls (92.66,86.37) and (94.91,84.13) .. (97.68,84.13) -- (104.36,84.14) .. controls (107.13,84.14) and (109.38,86.39) .. (109.37,89.16) -- (109.33,129.43) .. controls (109.33,132.2) and (107.08,134.44) .. (104.31,134.44) -- (97.62,134.43) .. controls (94.85,134.43) and (92.61,132.18) .. (92.61,129.41) -- cycle ;
\draw  [fill={rgb, 255:red, 255; green, 234; blue, 195 }  ,fill opacity=1 ][line width=0.75]  (67.46,122.67) .. controls (67.46,119.9) and (69.71,117.66) .. (72.48,117.66) -- (79.17,117.67) .. controls (81.94,117.67) and (84.18,119.92) .. (84.18,122.69) -- (84.14,162.96) .. controls (84.13,165.73) and (81.89,167.98) .. (79.12,167.97) -- (72.43,167.97) .. controls (69.66,167.96) and (67.42,165.72) .. (67.42,162.95) -- cycle ;
\draw  [fill={rgb, 255:red, 245; green, 240; blue, 235 }  ,fill opacity=1 ][line width=0.75]  (92.59,155.73) .. controls (92.59,152.96) and (94.84,150.72) .. (97.61,150.72) -- (104.29,150.73) .. controls (107.06,150.73) and (109.31,152.98) .. (109.3,155.75) -- (109.26,196.03) .. controls (109.26,198.8) and (107.01,201.04) .. (104.24,201.04) -- (97.55,201.03) .. controls (94.78,201.03) and (92.54,198.78) .. (92.54,196.01) -- cycle ;
\draw  [fill={rgb, 255:red, 245; green, 240; blue, 235 }  ,fill opacity=1 ][line width=0.75]  (67.39,189.23) .. controls (67.39,186.46) and (69.64,184.22) .. (72.41,184.22) -- (79.1,184.23) .. controls (81.87,184.23) and (84.11,186.48) .. (84.11,189.25) -- (84.07,229.53) .. controls (84.06,232.3) and (81.82,234.54) .. (79.05,234.54) -- (72.36,234.53) .. controls (69.59,234.53) and (67.35,232.28) .. (67.35,229.51) -- cycle ;
\draw [line width=0.75]    (50.66,260.9) -- (178.4,261.11) ;
\draw [line width=0.75]    (50.64,294.45) -- (178.38,294.65) ;
\draw  [fill={rgb, 255:red, 245; green, 240; blue, 235 }  ,fill opacity=1 ][line width=0.75]  (92.68,223.16) .. controls (92.68,220.39) and (94.93,218.15) .. (97.7,218.15) -- (104.39,218.16) .. controls (107.16,218.16) and (109.4,220.41) .. (109.4,223.18) -- (109.36,263.45) .. controls (109.35,266.22) and (107.11,268.47) .. (104.34,268.46) -- (97.65,268.46) .. controls (94.88,268.45) and (92.64,266.21) .. (92.64,263.44) -- cycle ;
\draw  [fill={rgb, 255:red, 245; green, 240; blue, 235 }  ,fill opacity=1 ][line width=0.75]  (67.49,257.16) .. controls (67.49,254.39) and (69.74,252.15) .. (72.51,252.15) -- (79.19,252.16) .. controls (81.96,252.16) and (84.21,254.41) .. (84.2,257.18) -- (84.16,297.45) .. controls (84.16,300.22) and (81.91,302.47) .. (79.14,302.46) -- (72.45,302.46) .. controls (69.68,302.45) and (67.44,300.21) .. (67.44,297.44) -- cycle ;
\draw  [fill={rgb, 255:red, 245; green, 166; blue, 35 }  ,fill opacity=1 ][line width=0.75]  (176.72,185.66) .. controls (176.72,185.66) and (176.72,185.66) .. (176.72,185.66) -- (193.44,185.68) .. controls (193.44,185.68) and (193.44,185.68) .. (193.44,185.68) -- (193.39,235.42) .. controls (193.39,235.42) and (193.39,235.42) .. (193.39,235.42) -- (176.67,235.4) .. controls (176.67,235.4) and (176.67,235.4) .. (176.67,235.4) -- cycle ;
\draw  [fill={rgb, 255:red, 245; green, 240; blue, 235 }  ,fill opacity=1 ][line width=0.75]  (118.98,55.58) .. controls (118.98,52.81) and (121.23,50.56) .. (124,50.57) -- (130.68,50.57) .. controls (133.45,50.58) and (135.7,52.82) .. (135.69,55.59) -- (135.65,95.87) .. controls (135.65,98.64) and (133.4,100.88) .. (130.63,100.88) -- (123.94,100.87) .. controls (121.17,100.87) and (118.93,98.62) .. (118.93,95.85) -- cycle ;
\draw  [fill={rgb, 255:red, 245; green, 240; blue, 235 }  ,fill opacity=1 ][line width=0.75]  (144.16,89.14) .. controls (144.16,86.37) and (146.41,84.13) .. (149.18,84.13) -- (155.86,84.14) .. controls (158.63,84.14) and (160.88,86.39) .. (160.87,89.16) -- (160.83,129.43) .. controls (160.83,132.2) and (158.58,134.44) .. (155.81,134.44) -- (149.12,134.43) .. controls (146.35,134.43) and (144.11,132.18) .. (144.11,129.41) -- cycle ;
\draw  [fill={rgb, 255:red, 255; green, 234; blue, 195 }  ,fill opacity=1 ][line width=0.75]  (118.96,123.07) .. controls (118.96,120.3) and (121.21,118.06) .. (123.98,118.06) -- (130.67,118.07) .. controls (133.44,118.07) and (135.68,120.32) .. (135.68,123.09) -- (135.64,163.36) .. controls (135.63,166.13) and (133.39,168.38) .. (130.62,168.37) -- (123.93,168.37) .. controls (121.16,168.36) and (118.92,166.12) .. (118.92,163.35) -- cycle ;
\draw  [fill={rgb, 255:red, 255; green, 234; blue, 195 }  ,fill opacity=1 ][line width=0.75]  (144.09,156.13) .. controls (144.09,153.36) and (146.34,151.12) .. (149.11,151.12) -- (155.79,151.13) .. controls (158.56,151.13) and (160.81,153.38) .. (160.8,156.15) -- (160.76,196.43) .. controls (160.76,199.2) and (158.51,201.44) .. (155.74,201.44) -- (149.05,201.43) .. controls (146.28,201.43) and (144.04,199.18) .. (144.04,196.41) -- cycle ;
\draw  [fill={rgb, 255:red, 245; green, 240; blue, 235 }  ,fill opacity=1 ][line width=0.75]  (118.89,189.23) .. controls (118.89,186.46) and (121.14,184.22) .. (123.91,184.22) -- (130.6,184.23) .. controls (133.37,184.23) and (135.61,186.48) .. (135.61,189.25) -- (135.57,229.53) .. controls (135.56,232.3) and (133.32,234.54) .. (130.55,234.54) -- (123.86,234.53) .. controls (121.09,234.53) and (118.85,232.28) .. (118.85,229.51) -- cycle ;
\draw  [fill={rgb, 255:red, 245; green, 240; blue, 235 }  ,fill opacity=1 ][line width=0.75]  (144.18,223.16) .. controls (144.18,220.39) and (146.43,218.15) .. (149.2,218.15) -- (155.89,218.16) .. controls (158.66,218.16) and (160.9,220.41) .. (160.9,223.18) -- (160.86,263.45) .. controls (160.85,266.22) and (158.61,268.47) .. (155.84,268.46) -- (149.15,268.46) .. controls (146.38,268.45) and (144.14,266.21) .. (144.14,263.44) -- cycle ;
\draw  [fill={rgb, 255:red, 245; green, 240; blue, 235 }  ,fill opacity=1 ][line width=0.75]  (118.99,257.16) .. controls (118.99,254.39) and (121.24,252.15) .. (124.01,252.15) -- (130.69,252.16) .. controls (133.46,252.16) and (135.71,254.41) .. (135.7,257.18) -- (135.66,297.45) .. controls (135.66,300.22) and (133.41,302.47) .. (130.64,302.46) -- (123.95,302.46) .. controls (121.18,302.45) and (118.94,300.21) .. (118.94,297.44) -- cycle ;
\draw [color={rgb, 255:red, 245; green, 166; blue, 35 }  ,draw opacity=1 ][line width=1.5]  [dash pattern={on 4.5pt off 2.25pt}]  (40.4,143.03) .. controls (49.73,143.03) and (67.87,142.93) .. (75.8,142.82) .. controls (83.73,142.7) and (94.09,109.26) .. (100.99,109.29) .. controls (107.9,109.31) and (143.27,175.7) .. (152.42,175.88) .. controls (161.58,176.06) and (165.93,176.03) .. (180.6,176.03) ;
\end{mytikz2}\quad.
\end{equation}
We define the ``length'' of a path $p$ by the sum of the lengths of all edges, where an edge refers to a pair of adjacent blocks $(B_{k}, B_{k'})$ in path $p$ with $k<k'$. The length of the edge is defined by
\begin{equation}\label{eq:def_edge_length}
    l(B_{k},B_{k'}) = \log_4 \left[\frac{4^{|s(B_{k'})|}-1}{4^{|s_c(B_{k}, B_{k'})|}-1}\right],
\end{equation}
where $s_c(B_{k},B_{k'})$ denotes the connecting support of $B_{k}$ and $B_{k'}$ with $k\leq k'$, which is defined by $s_c(B_{k},B_{k'})=s(B_{k})\cap s(B_{k'}) - \bigcup_{k<k''<k'}s(B_{k''})$. Specifically, we regard $(\rho_0, B_{k_1})$ also as an edge in the path.

However, a single path alone is not enough for the proof below because the back-evolved Pauli strings can be non-trivial on multiple regions simultaneously so we need to introduce the notion of ``path set''. A path set $P$ on a PQC with respect to an observable is defined by a collection of paths $P=\{p_1,p_2,\cdots\}$, where each path connects to the observable and the right ends of these paths cover the support of the observable. Along the backward direction, different paths in $P$ are allowed to merge but not allowed to split after overlapping. The length of the path set $P$ is defined by the sum of the length of all edges in the paths of $P$, denoted by $\opr{Edge}(P)$. Below are some instances of legal path sets
\begin{equation}
\begin{mytikz2}
\draw [line width=0.75]    (50.67,92.46) -- (178.41,92.67) ;
\draw [line width=0.75]    (50.65,126.01) -- (178.39,126.21) ;
\draw [line width=0.75]    (50.69,58.93) -- (178.43,59.14) ;
\draw  [fill={rgb, 255:red, 245; green, 240; blue, 235 }  ,fill opacity=1 ][line width=0.75]  (67.48,55.58) .. controls (67.48,52.81) and (69.73,50.56) .. (72.5,50.57) -- (79.18,50.57) .. controls (81.95,50.58) and (84.2,52.82) .. (84.19,55.59) -- (84.15,95.87) .. controls (84.15,98.64) and (81.9,100.88) .. (79.13,100.88) -- (72.44,100.87) .. controls (69.67,100.87) and (67.43,98.62) .. (67.43,95.85) -- cycle ;
\draw [line width=0.75]    (50.6,193.48) -- (178.34,193.69) ;
\draw [line width=0.75]    (50.58,227.02) -- (178.32,227.22) ;
\draw [line width=0.75]    (50.62,159.54) -- (178.37,159.75) ;
\draw  [fill={rgb, 255:red, 255; green, 234; blue, 195 }  ,fill opacity=1 ][line width=0.75]  (92.66,89.14) .. controls (92.66,86.37) and (94.91,84.13) .. (97.68,84.13) -- (104.36,84.14) .. controls (107.13,84.14) and (109.38,86.39) .. (109.37,89.16) -- (109.33,129.43) .. controls (109.33,132.2) and (107.08,134.44) .. (104.31,134.44) -- (97.62,134.43) .. controls (94.85,134.43) and (92.61,132.18) .. (92.61,129.41) -- cycle ;
\draw  [fill={rgb, 255:red, 255; green, 234; blue, 195 }  ,fill opacity=1 ][line width=0.75]  (67.46,122.67) .. controls (67.46,119.9) and (69.71,117.66) .. (72.48,117.66) -- (79.17,117.67) .. controls (81.94,117.67) and (84.18,119.92) .. (84.18,122.69) -- (84.14,162.96) .. controls (84.13,165.73) and (81.89,167.98) .. (79.12,167.97) -- (72.43,167.97) .. controls (69.66,167.96) and (67.42,165.72) .. (67.42,162.95) -- cycle ;
\draw  [fill={rgb, 255:red, 245; green, 240; blue, 235 }  ,fill opacity=1 ][line width=0.75]  (92.59,155.73) .. controls (92.59,152.96) and (94.84,150.72) .. (97.61,150.72) -- (104.29,150.73) .. controls (107.06,150.73) and (109.31,152.98) .. (109.3,155.75) -- (109.26,196.03) .. controls (109.26,198.8) and (107.01,201.04) .. (104.24,201.04) -- (97.55,201.03) .. controls (94.78,201.03) and (92.54,198.78) .. (92.54,196.01) -- cycle ;
\draw  [fill={rgb, 255:red, 255; green, 234; blue, 195 }  ,fill opacity=1 ][line width=0.75]  (67.39,189.23) .. controls (67.39,186.46) and (69.64,184.22) .. (72.41,184.22) -- (79.1,184.23) .. controls (81.87,184.23) and (84.11,186.48) .. (84.11,189.25) -- (84.07,229.53) .. controls (84.06,232.3) and (81.82,234.54) .. (79.05,234.54) -- (72.36,234.53) .. controls (69.59,234.53) and (67.35,232.28) .. (67.35,229.51) -- cycle ;
\draw [line width=0.75]    (50.66,260.9) -- (178.4,261.11) ;
\draw [line width=0.75]    (50.64,294.45) -- (178.38,294.65) ;
\draw  [fill={rgb, 255:red, 255; green, 234; blue, 195 }  ,fill opacity=1 ][line width=0.75]  (92.68,223.16) .. controls (92.68,220.39) and (94.93,218.15) .. (97.7,218.15) -- (104.39,218.16) .. controls (107.16,218.16) and (109.4,220.41) .. (109.4,223.18) -- (109.36,263.45) .. controls (109.35,266.22) and (107.11,268.47) .. (104.34,268.46) -- (97.65,268.46) .. controls (94.88,268.45) and (92.64,266.21) .. (92.64,263.44) -- cycle ;
\draw  [fill={rgb, 255:red, 245; green, 240; blue, 235 }  ,fill opacity=1 ][line width=0.75]  (67.49,257.16) .. controls (67.49,254.39) and (69.74,252.15) .. (72.51,252.15) -- (79.19,252.16) .. controls (81.96,252.16) and (84.21,254.41) .. (84.2,257.18) -- (84.16,297.45) .. controls (84.16,300.22) and (81.91,302.47) .. (79.14,302.46) -- (72.45,302.46) .. controls (69.68,302.45) and (67.44,300.21) .. (67.44,297.44) -- cycle ;
\draw  [fill={rgb, 255:red, 245; green, 166; blue, 35 }  ,fill opacity=1 ][line width=0.75]  (176.72,185.66) .. controls (176.72,185.66) and (176.72,185.66) .. (176.72,185.66) -- (193.44,185.68) .. controls (193.44,185.68) and (193.44,185.68) .. (193.44,185.68) -- (193.39,235.42) .. controls (193.39,235.42) and (193.39,235.42) .. (193.39,235.42) -- (176.67,235.4) .. controls (176.67,235.4) and (176.67,235.4) .. (176.67,235.4) -- cycle ;
\draw  [fill={rgb, 255:red, 245; green, 240; blue, 235 }  ,fill opacity=1 ][line width=0.75]  (118.98,55.58) .. controls (118.98,52.81) and (121.23,50.56) .. (124,50.57) -- (130.68,50.57) .. controls (133.45,50.58) and (135.7,52.82) .. (135.69,55.59) -- (135.65,95.87) .. controls (135.65,98.64) and (133.4,100.88) .. (130.63,100.88) -- (123.94,100.87) .. controls (121.17,100.87) and (118.93,98.62) .. (118.93,95.85) -- cycle ;
\draw  [fill={rgb, 255:red, 245; green, 240; blue, 235 }  ,fill opacity=1 ][line width=0.75]  (144.16,89.14) .. controls (144.16,86.37) and (146.41,84.13) .. (149.18,84.13) -- (155.86,84.14) .. controls (158.63,84.14) and (160.88,86.39) .. (160.87,89.16) -- (160.83,129.43) .. controls (160.83,132.2) and (158.58,134.44) .. (155.81,134.44) -- (149.12,134.43) .. controls (146.35,134.43) and (144.11,132.18) .. (144.11,129.41) -- cycle ;
\draw  [fill={rgb, 255:red, 255; green, 234; blue, 195 }  ,fill opacity=1 ][line width=0.75]  (118.96,123.07) .. controls (118.96,120.3) and (121.21,118.06) .. (123.98,118.06) -- (130.67,118.07) .. controls (133.44,118.07) and (135.68,120.32) .. (135.68,123.09) -- (135.64,163.36) .. controls (135.63,166.13) and (133.39,168.38) .. (130.62,168.37) -- (123.93,168.37) .. controls (121.16,168.36) and (118.92,166.12) .. (118.92,163.35) -- cycle ;
\draw  [fill={rgb, 255:red, 255; green, 234; blue, 195 }  ,fill opacity=1 ][line width=0.75]  (144.09,156.13) .. controls (144.09,153.36) and (146.34,151.12) .. (149.11,151.12) -- (155.79,151.13) .. controls (158.56,151.13) and (160.81,153.38) .. (160.8,156.15) -- (160.76,196.43) .. controls (160.76,199.2) and (158.51,201.44) .. (155.74,201.44) -- (149.05,201.43) .. controls (146.28,201.43) and (144.04,199.18) .. (144.04,196.41) -- cycle ;
\draw  [fill={rgb, 255:red, 255; green, 234; blue, 195 }  ,fill opacity=1 ][line width=0.75]  (118.89,189.23) .. controls (118.89,186.46) and (121.14,184.22) .. (123.91,184.22) -- (130.6,184.23) .. controls (133.37,184.23) and (135.61,186.48) .. (135.61,189.25) -- (135.57,229.53) .. controls (135.56,232.3) and (133.32,234.54) .. (130.55,234.54) -- (123.86,234.53) .. controls (121.09,234.53) and (118.85,232.28) .. (118.85,229.51) -- cycle ;
\draw  [fill={rgb, 255:red, 255; green, 234; blue, 195 }  ,fill opacity=1 ][line width=0.75]  (144.18,223.16) .. controls (144.18,220.39) and (146.43,218.15) .. (149.2,218.15) -- (155.89,218.16) .. controls (158.66,218.16) and (160.9,220.41) .. (160.9,223.18) -- (160.86,263.45) .. controls (160.85,266.22) and (158.61,268.47) .. (155.84,268.46) -- (149.15,268.46) .. controls (146.38,268.45) and (144.14,266.21) .. (144.14,263.44) -- cycle ;
\draw  [fill={rgb, 255:red, 245; green, 240; blue, 235 }  ,fill opacity=1 ][line width=0.75]  (118.99,257.16) .. controls (118.99,254.39) and (121.24,252.15) .. (124.01,252.15) -- (130.69,252.16) .. controls (133.46,252.16) and (135.71,254.41) .. (135.7,257.18) -- (135.66,297.45) .. controls (135.66,300.22) and (133.41,302.47) .. (130.64,302.46) -- (123.95,302.46) .. controls (121.18,302.45) and (118.94,300.21) .. (118.94,297.44) -- cycle ;
\draw [color={rgb, 255:red, 245; green, 166; blue, 35 }  ,draw opacity=1 ][line width=1.5]  [dash pattern={on 4.5pt off 2.25pt}]  (40.4,143.03) .. controls (49.73,143.03) and (67.87,142.93) .. (75.8,142.82) .. controls (83.73,142.7) and (94.09,109.26) .. (100.99,109.29) .. controls (107.9,109.31) and (143.27,175.7) .. (152.42,175.88) .. controls (161.58,176.06) and (165.93,176.03) .. (180.6,176.03) ;
\draw [color={rgb, 255:red, 245; green, 166; blue, 35 }  ,draw opacity=1 ][line width=1.5]  [dash pattern={on 4.5pt off 2.25pt}]  (40.33,209.6) .. controls (49.66,209.6) and (67.79,209.5) .. (75.73,209.38) .. controls (83.66,209.26) and (94.12,243.28) .. (101.02,243.31) .. controls (107.92,243.33) and (120.06,209.24) .. (127.23,209.38) .. controls (134.4,209.52) and (147.78,242.35) .. (152.35,242.44) .. controls (156.93,242.53) and (165.86,242.6) .. (180.53,242.6) ;
\end{mytikz2}\quad,\quad
\begin{mytikz2}
\draw [line width=0.75]    (50.67,92.46) -- (178.41,92.67) ;
\draw [line width=0.75]    (50.65,126.01) -- (178.39,126.21) ;
\draw [line width=0.75]    (50.69,58.93) -- (178.43,59.14) ;
\draw  [fill={rgb, 255:red, 245; green, 240; blue, 235 }  ,fill opacity=1 ][line width=0.75]  (67.48,55.58) .. controls (67.48,52.81) and (69.73,50.56) .. (72.5,50.57) -- (79.18,50.57) .. controls (81.95,50.58) and (84.2,52.82) .. (84.19,55.59) -- (84.15,95.87) .. controls (84.15,98.64) and (81.9,100.88) .. (79.13,100.88) -- (72.44,100.87) .. controls (69.67,100.87) and (67.43,98.62) .. (67.43,95.85) -- cycle ;
\draw [line width=0.75]    (50.6,193.48) -- (178.34,193.69) ;
\draw [line width=0.75]    (50.58,227.02) -- (178.32,227.22) ;
\draw [line width=0.75]    (50.62,159.54) -- (178.37,159.75) ;
\draw  [fill={rgb, 255:red, 255; green, 234; blue, 195 }  ,fill opacity=1 ][line width=0.75]  (92.66,89.14) .. controls (92.66,86.37) and (94.91,84.13) .. (97.68,84.13) -- (104.36,84.14) .. controls (107.13,84.14) and (109.38,86.39) .. (109.37,89.16) -- (109.33,129.43) .. controls (109.33,132.2) and (107.08,134.44) .. (104.31,134.44) -- (97.62,134.43) .. controls (94.85,134.43) and (92.61,132.18) .. (92.61,129.41) -- cycle ;
\draw  [fill={rgb, 255:red, 255; green, 234; blue, 195 }  ,fill opacity=1 ][line width=0.75]  (67.46,122.67) .. controls (67.46,119.9) and (69.71,117.66) .. (72.48,117.66) -- (79.17,117.67) .. controls (81.94,117.67) and (84.18,119.92) .. (84.18,122.69) -- (84.14,162.96) .. controls (84.13,165.73) and (81.89,167.98) .. (79.12,167.97) -- (72.43,167.97) .. controls (69.66,167.96) and (67.42,165.72) .. (67.42,162.95) -- cycle ;
\draw  [fill={rgb, 255:red, 255; green, 234; blue, 195 }  ,fill opacity=1 ][line width=0.75]  (92.59,155.73) .. controls (92.59,152.96) and (94.84,150.72) .. (97.61,150.72) -- (104.29,150.73) .. controls (107.06,150.73) and (109.31,152.98) .. (109.3,155.75) -- (109.26,196.03) .. controls (109.26,198.8) and (107.01,201.04) .. (104.24,201.04) -- (97.55,201.03) .. controls (94.78,201.03) and (92.54,198.78) .. (92.54,196.01) -- cycle ;
\draw  [fill={rgb, 255:red, 245; green, 240; blue, 235 }  ,fill opacity=1 ][line width=0.75]  (67.39,189.23) .. controls (67.39,186.46) and (69.64,184.22) .. (72.41,184.22) -- (79.1,184.23) .. controls (81.87,184.23) and (84.11,186.48) .. (84.11,189.25) -- (84.07,229.53) .. controls (84.06,232.3) and (81.82,234.54) .. (79.05,234.54) -- (72.36,234.53) .. controls (69.59,234.53) and (67.35,232.28) .. (67.35,229.51) -- cycle ;
\draw [line width=0.75]    (50.66,260.9) -- (178.4,261.11) ;
\draw [line width=0.75]    (50.64,294.45) -- (178.38,294.65) ;
\draw  [fill={rgb, 255:red, 245; green, 240; blue, 235 }  ,fill opacity=1 ][line width=0.75]  (92.68,223.16) .. controls (92.68,220.39) and (94.93,218.15) .. (97.7,218.15) -- (104.39,218.16) .. controls (107.16,218.16) and (109.4,220.41) .. (109.4,223.18) -- (109.36,263.45) .. controls (109.35,266.22) and (107.11,268.47) .. (104.34,268.46) -- (97.65,268.46) .. controls (94.88,268.45) and (92.64,266.21) .. (92.64,263.44) -- cycle ;
\draw  [fill={rgb, 255:red, 245; green, 240; blue, 235 }  ,fill opacity=1 ][line width=0.75]  (67.49,257.16) .. controls (67.49,254.39) and (69.74,252.15) .. (72.51,252.15) -- (79.19,252.16) .. controls (81.96,252.16) and (84.21,254.41) .. (84.2,257.18) -- (84.16,297.45) .. controls (84.16,300.22) and (81.91,302.47) .. (79.14,302.46) -- (72.45,302.46) .. controls (69.68,302.45) and (67.44,300.21) .. (67.44,297.44) -- cycle ;
\draw  [fill={rgb, 255:red, 245; green, 166; blue, 35 }  ,fill opacity=1 ][line width=0.75]  (176.72,185.66) .. controls (176.72,185.66) and (176.72,185.66) .. (176.72,185.66) -- (193.44,185.68) .. controls (193.44,185.68) and (193.44,185.68) .. (193.44,185.68) -- (193.39,235.42) .. controls (193.39,235.42) and (193.39,235.42) .. (193.39,235.42) -- (176.67,235.4) .. controls (176.67,235.4) and (176.67,235.4) .. (176.67,235.4) -- cycle ;
\draw  [fill={rgb, 255:red, 245; green, 240; blue, 235 }  ,fill opacity=1 ][line width=0.75]  (118.98,55.58) .. controls (118.98,52.81) and (121.23,50.56) .. (124,50.57) -- (130.68,50.57) .. controls (133.45,50.58) and (135.7,52.82) .. (135.69,55.59) -- (135.65,95.87) .. controls (135.65,98.64) and (133.4,100.88) .. (130.63,100.88) -- (123.94,100.87) .. controls (121.17,100.87) and (118.93,98.62) .. (118.93,95.85) -- cycle ;
\draw  [fill={rgb, 255:red, 245; green, 240; blue, 235 }  ,fill opacity=1 ][line width=0.75]  (144.16,89.14) .. controls (144.16,86.37) and (146.41,84.13) .. (149.18,84.13) -- (155.86,84.14) .. controls (158.63,84.14) and (160.88,86.39) .. (160.87,89.16) -- (160.83,129.43) .. controls (160.83,132.2) and (158.58,134.44) .. (155.81,134.44) -- (149.12,134.43) .. controls (146.35,134.43) and (144.11,132.18) .. (144.11,129.41) -- cycle ;
\draw  [fill={rgb, 255:red, 255; green, 234; blue, 195 }  ,fill opacity=1 ][line width=0.75]  (118.96,123.07) .. controls (118.96,120.3) and (121.21,118.06) .. (123.98,118.06) -- (130.67,118.07) .. controls (133.44,118.07) and (135.68,120.32) .. (135.68,123.09) -- (135.64,163.36) .. controls (135.63,166.13) and (133.39,168.38) .. (130.62,168.37) -- (123.93,168.37) .. controls (121.16,168.36) and (118.92,166.12) .. (118.92,163.35) -- cycle ;
\draw  [fill={rgb, 255:red, 255; green, 234; blue, 195 }  ,fill opacity=1 ][line width=0.75]  (144.09,156.13) .. controls (144.09,153.36) and (146.34,151.12) .. (149.11,151.12) -- (155.79,151.13) .. controls (158.56,151.13) and (160.81,153.38) .. (160.8,156.15) -- (160.76,196.43) .. controls (160.76,199.2) and (158.51,201.44) .. (155.74,201.44) -- (149.05,201.43) .. controls (146.28,201.43) and (144.04,199.18) .. (144.04,196.41) -- cycle ;
\draw  [fill={rgb, 255:red, 255; green, 234; blue, 195 }  ,fill opacity=1 ][line width=0.75]  (118.89,189.23) .. controls (118.89,186.46) and (121.14,184.22) .. (123.91,184.22) -- (130.6,184.23) .. controls (133.37,184.23) and (135.61,186.48) .. (135.61,189.25) -- (135.57,229.53) .. controls (135.56,232.3) and (133.32,234.54) .. (130.55,234.54) -- (123.86,234.53) .. controls (121.09,234.53) and (118.85,232.28) .. (118.85,229.51) -- cycle ;
\draw  [fill={rgb, 255:red, 255; green, 234; blue, 195 }  ,fill opacity=1 ][line width=0.75]  (144.18,223.16) .. controls (144.18,220.39) and (146.43,218.15) .. (149.2,218.15) -- (155.89,218.16) .. controls (158.66,218.16) and (160.9,220.41) .. (160.9,223.18) -- (160.86,263.45) .. controls (160.85,266.22) and (158.61,268.47) .. (155.84,268.46) -- (149.15,268.46) .. controls (146.38,268.45) and (144.14,266.21) .. (144.14,263.44) -- cycle ;
\draw  [fill={rgb, 255:red, 245; green, 240; blue, 235 }  ,fill opacity=1 ][line width=0.75]  (118.99,257.16) .. controls (118.99,254.39) and (121.24,252.15) .. (124.01,252.15) -- (130.69,252.16) .. controls (133.46,252.16) and (135.71,254.41) .. (135.7,257.18) -- (135.66,297.45) .. controls (135.66,300.22) and (133.41,302.47) .. (130.64,302.46) -- (123.95,302.46) .. controls (121.18,302.45) and (118.94,300.21) .. (118.94,297.44) -- cycle ;
\draw [color={rgb, 255:red, 245; green, 166; blue, 35 }  ,draw opacity=1 ][line width=1.5]  [dash pattern={on 4.5pt off 2.25pt}]  (40.4,143.03) .. controls (49.73,143.03) and (67.87,142.93) .. (75.8,142.82) .. controls (83.73,142.7) and (94.09,109.26) .. (100.99,109.29) .. controls (107.9,109.31) and (143.27,175.7) .. (152.42,175.88) .. controls (161.58,176.06) and (165.93,176.03) .. (180.6,176.03) ;
\draw [color={rgb, 255:red, 245; green, 166; blue, 35 }  ,draw opacity=1 ][line width=1.5]  [dash pattern={on 4.5pt off 2.25pt}]  (40.4,150.55) .. controls (49.73,150.55) and (68.23,150.5) .. (76.17,150.38) .. controls (84.1,150.27) and (145.62,243.28) .. (152.52,243.31) .. controls (159.42,243.33) and (166.03,243.46) .. (180.7,243.46) ;
\end{mytikz2}\quad,\quad
\begin{mytikz2}
\draw [line width=0.75]    (50.67,92.46) -- (178.41,92.67) ;
\draw [line width=0.75]    (50.65,126.01) -- (178.39,126.21) ;
\draw [line width=0.75]    (50.69,58.93) -- (178.43,59.14) ;
\draw  [fill={rgb, 255:red, 255; green, 234; blue, 195 }  ,fill opacity=1 ][line width=0.75]  (67.48,55.58) .. controls (67.48,52.81) and (69.73,50.56) .. (72.5,50.57) -- (79.18,50.57) .. controls (81.95,50.58) and (84.2,52.82) .. (84.19,55.59) -- (84.15,95.87) .. controls (84.15,98.64) and (81.9,100.88) .. (79.13,100.88) -- (72.44,100.87) .. controls (69.67,100.87) and (67.43,98.62) .. (67.43,95.85) -- cycle ;
\draw [line width=0.75]    (50.6,193.48) -- (178.34,193.69) ;
\draw [line width=0.75]    (50.58,227.02) -- (178.32,227.22) ;
\draw [line width=0.75]    (50.62,159.54) -- (178.37,159.75) ;
\draw  [fill={rgb, 255:red, 255; green, 234; blue, 195 }  ,fill opacity=1 ][line width=0.75]  (92.66,89.14) .. controls (92.66,86.37) and (94.91,84.13) .. (97.68,84.13) -- (104.36,84.14) .. controls (107.13,84.14) and (109.38,86.39) .. (109.37,89.16) -- (109.33,129.43) .. controls (109.33,132.2) and (107.08,134.44) .. (104.31,134.44) -- (97.62,134.43) .. controls (94.85,134.43) and (92.61,132.18) .. (92.61,129.41) -- cycle ;
\draw  [fill={rgb, 255:red, 245; green, 240; blue, 235 }  ,fill opacity=1 ][line width=0.75]  (67.46,122.67) .. controls (67.46,119.9) and (69.71,117.66) .. (72.48,117.66) -- (79.17,117.67) .. controls (81.94,117.67) and (84.18,119.92) .. (84.18,122.69) -- (84.14,162.96) .. controls (84.13,165.73) and (81.89,167.98) .. (79.12,167.97) -- (72.43,167.97) .. controls (69.66,167.96) and (67.42,165.72) .. (67.42,162.95) -- cycle ;
\draw  [fill={rgb, 255:red, 245; green, 240; blue, 235 }  ,fill opacity=1 ][line width=0.75]  (92.59,155.73) .. controls (92.59,152.96) and (94.84,150.72) .. (97.61,150.72) -- (104.29,150.73) .. controls (107.06,150.73) and (109.31,152.98) .. (109.3,155.75) -- (109.26,196.03) .. controls (109.26,198.8) and (107.01,201.04) .. (104.24,201.04) -- (97.55,201.03) .. controls (94.78,201.03) and (92.54,198.78) .. (92.54,196.01) -- cycle ;
\draw  [fill={rgb, 255:red, 245; green, 240; blue, 235 }  ,fill opacity=1 ][line width=0.75]  (67.39,189.23) .. controls (67.39,186.46) and (69.64,184.22) .. (72.41,184.22) -- (79.1,184.23) .. controls (81.87,184.23) and (84.11,186.48) .. (84.11,189.25) -- (84.07,229.53) .. controls (84.06,232.3) and (81.82,234.54) .. (79.05,234.54) -- (72.36,234.53) .. controls (69.59,234.53) and (67.35,232.28) .. (67.35,229.51) -- cycle ;
\draw [line width=0.75]    (50.66,260.9) -- (178.4,261.11) ;
\draw [line width=0.75]    (50.64,294.45) -- (178.38,294.65) ;
\draw  [fill={rgb, 255:red, 245; green, 240; blue, 235 }  ,fill opacity=1 ][line width=0.75]  (92.68,223.16) .. controls (92.68,220.39) and (94.93,218.15) .. (97.7,218.15) -- (104.39,218.16) .. controls (107.16,218.16) and (109.4,220.41) .. (109.4,223.18) -- (109.36,263.45) .. controls (109.35,266.22) and (107.11,268.47) .. (104.34,268.46) -- (97.65,268.46) .. controls (94.88,268.45) and (92.64,266.21) .. (92.64,263.44) -- cycle ;
\draw  [fill={rgb, 255:red, 255; green, 234; blue, 195 }  ,fill opacity=1 ][line width=0.75]  (67.49,257.16) .. controls (67.49,254.39) and (69.74,252.15) .. (72.51,252.15) -- (79.19,252.16) .. controls (81.96,252.16) and (84.21,254.41) .. (84.2,257.18) -- (84.16,297.45) .. controls (84.16,300.22) and (81.91,302.47) .. (79.14,302.46) -- (72.45,302.46) .. controls (69.68,302.45) and (67.44,300.21) .. (67.44,297.44) -- cycle ;
\draw  [fill={rgb, 255:red, 245; green, 166; blue, 35 }  ,fill opacity=1 ][line width=0.75]  (176.72,185.66) .. controls (176.72,185.66) and (176.72,185.66) .. (176.72,185.66) -- (193.44,185.68) .. controls (193.44,185.68) and (193.44,185.68) .. (193.44,185.68) -- (193.39,235.42) .. controls (193.39,235.42) and (193.39,235.42) .. (193.39,235.42) -- (176.67,235.4) .. controls (176.67,235.4) and (176.67,235.4) .. (176.67,235.4) -- cycle ;
\draw  [fill={rgb, 255:red, 245; green, 240; blue, 235 }  ,fill opacity=1 ][line width=0.75]  (118.98,55.58) .. controls (118.98,52.81) and (121.23,50.56) .. (124,50.57) -- (130.68,50.57) .. controls (133.45,50.58) and (135.7,52.82) .. (135.69,55.59) -- (135.65,95.87) .. controls (135.65,98.64) and (133.4,100.88) .. (130.63,100.88) -- (123.94,100.87) .. controls (121.17,100.87) and (118.93,98.62) .. (118.93,95.85) -- cycle ;
\draw  [fill={rgb, 255:red, 245; green, 240; blue, 235 }  ,fill opacity=1 ][line width=0.75]  (144.16,89.14) .. controls (144.16,86.37) and (146.41,84.13) .. (149.18,84.13) -- (155.86,84.14) .. controls (158.63,84.14) and (160.88,86.39) .. (160.87,89.16) -- (160.83,129.43) .. controls (160.83,132.2) and (158.58,134.44) .. (155.81,134.44) -- (149.12,134.43) .. controls (146.35,134.43) and (144.11,132.18) .. (144.11,129.41) -- cycle ;
\draw  [fill={rgb, 255:red, 255; green, 234; blue, 195 }  ,fill opacity=1 ][line width=0.75]  (118.96,123.07) .. controls (118.96,120.3) and (121.21,118.06) .. (123.98,118.06) -- (130.67,118.07) .. controls (133.44,118.07) and (135.68,120.32) .. (135.68,123.09) -- (135.64,163.36) .. controls (135.63,166.13) and (133.39,168.38) .. (130.62,168.37) -- (123.93,168.37) .. controls (121.16,168.36) and (118.92,166.12) .. (118.92,163.35) -- cycle ;
\draw  [fill={rgb, 255:red, 255; green, 234; blue, 195 }  ,fill opacity=1 ][line width=0.75]  (144.09,156.13) .. controls (144.09,153.36) and (146.34,151.12) .. (149.11,151.12) -- (155.79,151.13) .. controls (158.56,151.13) and (160.81,153.38) .. (160.8,156.15) -- (160.76,196.43) .. controls (160.76,199.2) and (158.51,201.44) .. (155.74,201.44) -- (149.05,201.43) .. controls (146.28,201.43) and (144.04,199.18) .. (144.04,196.41) -- cycle ;
\draw  [fill={rgb, 255:red, 245; green, 240; blue, 235 }  ,fill opacity=1 ][line width=0.75]  (118.89,189.23) .. controls (118.89,186.46) and (121.14,184.22) .. (123.91,184.22) -- (130.6,184.23) .. controls (133.37,184.23) and (135.61,186.48) .. (135.61,189.25) -- (135.57,229.53) .. controls (135.56,232.3) and (133.32,234.54) .. (130.55,234.54) -- (123.86,234.53) .. controls (121.09,234.53) and (118.85,232.28) .. (118.85,229.51) -- cycle ;
\draw  [fill={rgb, 255:red, 255; green, 234; blue, 195 }  ,fill opacity=1 ][line width=0.75]  (144.18,223.16) .. controls (144.18,220.39) and (146.43,218.15) .. (149.2,218.15) -- (155.89,218.16) .. controls (158.66,218.16) and (160.9,220.41) .. (160.9,223.18) -- (160.86,263.45) .. controls (160.85,266.22) and (158.61,268.47) .. (155.84,268.46) -- (149.15,268.46) .. controls (146.38,268.45) and (144.14,266.21) .. (144.14,263.44) -- cycle ;
\draw  [fill={rgb, 255:red, 255; green, 234; blue, 195 }  ,fill opacity=1 ][line width=0.75]  (118.99,257.16) .. controls (118.99,254.39) and (121.24,252.15) .. (124.01,252.15) -- (130.69,252.16) .. controls (133.46,252.16) and (135.71,254.41) .. (135.7,257.18) -- (135.66,297.45) .. controls (135.66,300.22) and (133.41,302.47) .. (130.64,302.46) -- (123.95,302.46) .. controls (121.18,302.45) and (118.94,300.21) .. (118.94,297.44) -- cycle ;
\draw [color={rgb, 255:red, 245; green, 166; blue, 35 }  ,draw opacity=1 ][line width=1.5]  [dash pattern={on 4.5pt off 2.25pt}]  (44.24,75.71) .. controls (53.57,75.71) and (67.88,75.84) .. (75.81,75.72) .. controls (83.75,75.61) and (145.52,175.86) .. (152.42,175.88) .. controls (159.33,175.9) and (165.93,176.03) .. (180.6,176.03) ;
\draw [color={rgb, 255:red, 245; green, 166; blue, 35 }  ,draw opacity=1 ][line width=1.5]  [dash pattern={on 4.5pt off 2.25pt}]  (48.24,284.51) .. controls (57.57,284.51) and (119.51,284.62) .. (127.44,284.51) .. controls (135.37,284.39) and (145.62,243.28) .. (152.52,243.31) .. controls (159.42,243.33) and (166.03,243.46) .. (180.7,243.46) ;
\end{mytikz2}\quad, \quad \cdots\quad.
\end{equation}
Besides length, we also need to define a ``width'' quantity to describe the non-locality of the path set. We define the ``head width'' of a path set $P$ by the sum of the widths of all the head blocks in $P$. Here the width of a block $B_k$ is defined by
\begin{equation}\label{eq:def_forward_width}
    w(B_k) = \log_2\left[\frac{4^{|s_f(B_k)|}-1}{2^{|s_f(B_k)|}-1}\right],
\end{equation}
where $s_f(B_k)$ denotes the so-called forward residual support of $B_k$, which is defined by $s_f(B_k) = s(B_k) - \bigcup_{k'<k} s(B_{k'})$. The head blocks in $P$ refer to the blocks that connect to $\rho_0$ directly, denoted by $\opr{Head}(P)$. With the concepts of the path set, we are prepared to prove the theorem in the main text.

\begin{theorem}\label{theorem:var_lower_bound_path_set}
Suppose that $H$ is a Hamiltonian with the Pauli decomposition $H=\sum_j\lambda_j h_j$. $\mathbf{U}(\bm{\theta})$ is a PQC in arbitrary shapes composed of blocks forming independent local $2$-designs with $s(\mathbf{U})\supseteq s(H)$. The cost function is the energy expectation $C(\bm{\theta})=\braoprket{\bm{0}}{\mathbf{U}^\dagger(\bm{\theta}) H \mathbf{U}(\bm{\theta})}{\bm{0}}$. Then, the variance of the cost value is lower bounded by
\begin{equation}\label{eq:var_bound_path_length_width}
    \var_\mathbb{U} [C] \geq \sum_j \sum_{P_j}  \lambda_j^2 \cdot 2^{- 2 l(P_j) - w(P_j)},
\end{equation}
where $j$ runs over the indices of $h_j$ and $P_j$ runs over all the path sets covering $h_j$. $l(P_j)$ and $w(P_j)$ denote the length and the head width of the path set $P_j$, respectively.
\end{theorem}
\begin{proof}
Based on the analysis above, the variance of the cost value can be expressed as
\begin{equation}\label{eq:var_hj}
    \var_\mathbb{U}\left[C\right] = \sum_j \lambda_j^2 \tr\left[\rho_0^{\otimes 2} \mathcal{T}^{(2)}_{s(B_1)}\circ\mathcal{T}^{(2)}_{s(B_2)} \circ\cdots\circ \mathcal{T}^{(2)}_{s(B_{M'})}(h_j^{\otimes 2}) \right].
\end{equation}
We first give a big picture of the things that occur in each summed term in the expression above. A twirling channel just maps a doubled Pauli string into a sum of doubled Pauli strings, with the weights positive and summed to one. The subsequent twirling channels just map them to other sums of doubled Pauli strings. In the end, when encountering the initial state $\ket{\bm{0}}$, in the final summation, those containing only $Z$ and $I$ operators contribute to the variance with a base factor $1$. So the variance equals the sum of the weights corresponding to all surviving doubled Pauli strings that contain only $Z$ and $I$ in the output of the final channel.

It is hard to track the exact backward evolution of all possible Pauli strings. Instead, noticing that the weights are always non-negative, it is easy to track a subset of these Pauli strings to give a lower bound. We choose to track those local Pauli strings, i.e., those non-trivial only on a few qubits. To be specific, the variance can be relaxed by manually inserting an identity channel $\mathcal{I}_{s(B_{k})}^{(2)}$ just after each twirling channel $\mathcal{T}_{s(B_{k})}^{(2)}$, decomposing the identity channel into several selection channels
\begin{equation}\label{eq:decompose_identity}
    \mathcal{I}_{s(B_{k})}^{(2)} = \sum_{k'} \mathcal{S}_{s(B_{k'}), s(B_{k}) }^{(2)} + \cdots,
\end{equation}
and just retaining a part of the selection channels. Here $k'$ runs over the blocks connected to $B_{k}$ with $k'<k$. $\mathcal{S}_{s(B_{k'}), s(B_{k})}^{(2)}$ is the ``local'' selection channel defined by
\begin{equation}\label{eq:selection_channel}
    \mathcal{S}_{s(B_{k})}^{(2)} (\sigma^{\otimes 2}) = \begin{cases}
        \sigma^{\otimes 2} & ~\text{if}~ s(\sigma\vert_{s(B_{k})})\subseteq s_c(B_{k'}, B_{k})\cup \bar{s}(B_{k}) \\
        0 & ~\text{else}~
    \end{cases},
\end{equation}
where $\sigma$ is an $N$-qubit Pauli string. Namely, if $\sigma$ is trivial on $s(B_k)$, then the selection channel does not change anything. If $\sigma$ is non-trivial on $s(B_k)$ and $\sigma\vert_{s(B_{k})}$ is only non-trivial on $s_c(B_{k'},B_{k})$, the selection channel still do nothing. But if $\sigma$ is non-trivial on $s(B_k)$ and $\sigma\vert_{s(B_{k})}$ is non-trivial outside $s_c(B_{k'},B_{k})$, the selection channel will project out the Pauli string. In other words, we discard the Pauli strings that are non-trivial simultaneously on the supports of multiple predecessors, represented by the residual terms ``$\cdots$'' in Eq.\,\eqref{eq:decompose_identity}. After expanding the product of the summations of the local selection channels after each twirling channel, one obtains a summation of contributions from all the legal path sets, i.e., the variance regarding $h_j$ can be lower bounded by
\begin{equation}\label{eq:var_bound_selection_channel}
    \var_\mathbb{U}\left[\avg{h_j}\right] \geq \sum_{P_j} \tr\left[\rho_0^{\otimes 2} \prod_{(B_{k'},B_{k})\in \opr{Edge}(P_j)}^{\rightarrow} \left(\mathcal{S}_{s(B_{k'}),s(B_{k})}^{(2)} \circ \mathcal{T}_{s(B_{k})}^{(2)} \right)(h_j^{\otimes 2}) \right].
\end{equation}
To estimate each of the contributions, consider that for any two adjacent blocks $(B_{k'}, B_{k})$ in $P_j$, only a certain proportion of strings in the output of $\mathcal{T}_{s(B_{k})}^{(2)}$ satisfies the selection rule. This proportion is just the number of elements in $\mathcal{P}'\vert_{s_c(B_{k}, B_{k'})}$ over that in $\mathcal{P}'\vert_{s(B_{k'})}$. The proportion factors are accumulated recursively until meeting the initial state, which contributes a total factor of
\begin{equation}
    \prod_{(B_{k'},B_{k})\in \opr{Edge}(P_j)}\frac{4^{|s_c(B_{k'},B_{k})|} - 1}{4^{|s(B_{k})|} - 1} = 4^{-l(P_j)}.
\end{equation}
Finally, only those that contain only $Z$ and $I$ will contribute by a base number of $1$, whose proportion is
\begin{equation}
    \prod_{B_{k}\in \opr{Head}(P_j)}\frac{2^{|s_f(B_{k})|}-1}{4^{|s_f(B_{k})|}-1} = 2^{-w(P_j)},
\end{equation}
Thus, the final contribution of path set $P_j$ is just $4^{-l(P_j)} \times 2^{-w(P_j)}$.
\end{proof}

One can find that this proof is simpler than that for the variance of the cost derivative~\cite{Zhang2024} where we need special treatments for the differential block and the path set has to go through the differential block. Here the path set gets much more freedom to go over the circuit.

According to the proof, we can find that for locally scrambled circuits, the information of the Hamiltonian received by $\var_\mathbb{U}[C]$ is partially erased, leaving only the information of the support of each subterm. That is to say, a term like $X_1Y_2Z_3$ should have the same variance as the term $Z_1Z_2Z_3$. This can be partially resolved by introducing specialized block templates such as $R_y$-$\opr{CZ}$ as mentioned above.

The lower bound in Theorem~\ref{theorem:var_lower_bound_path_set} is quite general but abstract for practical use. Another looser but more compact lower bound can be derived from it, which involves the concept of ``local depth''~\cite{Zhang2024}. The local depth of a PQC at qubit $q_i$ is defined by the number of gates or blocks acting on $q_i$. The maximum value of the local depths of all qubits determines the (maximum) local depth of a PQC, denoted by $\chi$. Note that $\chi$ reduces to the conventional circuit depth $D$ for common repeated-layer circuits such as the brickwall circuits, but can be different from $D$ in certain cases. Thus, we can derive Theorem~\ref{theorem:var_lower_bound_path_set} mentioned in the main text.

\renewcommand\theproposition{1}
\setcounter{proposition}{\arabic{proposition}-1}
\begin{theorem}
For a PQC composed of blocks forming independent local $2$-designs and an $r$-local Hamiltonian, the landscape fluctuation is lower bounded by
\begin{equation}\label{eq:sigma_lower_bound}
    \sigma\geq 2^{-r\chi\beta} \frac{\|\bm{\lambda}\|_2}{\|\bm{\lambda}\|_1},
\end{equation}
where $\chi$ is the maximum local depth and $\beta$ is the maximum block size of the circuit.
\end{theorem}
\renewcommand{\theproposition}{S\arabic{proposition}}
\begin{proof}
According to Theorem~\ref{theorem:var_lower_bound_path_set}, for each subterm $h_j$, one can choose a single path set that goes along the straight wires of $h_j$ and discard the contributions from all other path sets. An instance is depicted as follows 
\begin{equation}
\begin{mytikz2}
\draw [line width=0.75]    (50.67,92.46) -- (178.41,92.67) ;
\draw [line width=0.75]    (50.65,126.01) -- (178.39,126.21) ;
\draw [line width=0.75]    (50.69,58.93) -- (178.43,59.14) ;
\draw  [fill={rgb, 255:red, 245; green, 240; blue, 235 }  ,fill opacity=1 ][line width=0.75]  (67.48,55.58) .. controls (67.48,52.81) and (69.73,50.56) .. (72.5,50.57) -- (79.18,50.57) .. controls (81.95,50.58) and (84.2,52.82) .. (84.19,55.59) -- (84.15,95.87) .. controls (84.15,98.64) and (81.9,100.88) .. (79.13,100.88) -- (72.44,100.87) .. controls (69.67,100.87) and (67.43,98.62) .. (67.43,95.85) -- cycle ;
\draw [line width=0.75]    (50.6,193.48) -- (178.34,193.69) ;
\draw [line width=0.75]    (50.58,227.02) -- (178.32,227.22) ;
\draw [line width=0.75]    (50.62,159.54) -- (178.37,159.75) ;
\draw  [fill={rgb, 255:red, 245; green, 240; blue, 235 }  ,fill opacity=1 ][line width=0.75]  (92.66,89.14) .. controls (92.66,86.37) and (94.91,84.13) .. (97.68,84.13) -- (104.36,84.14) .. controls (107.13,84.14) and (109.38,86.39) .. (109.37,89.16) -- (109.33,129.43) .. controls (109.33,132.2) and (107.08,134.44) .. (104.31,134.44) -- (97.62,134.43) .. controls (94.85,134.43) and (92.61,132.18) .. (92.61,129.41) -- cycle ;
\draw  [fill={rgb, 255:red, 245; green, 240; blue, 235 }  ,fill opacity=1 ][line width=0.75]  (67.46,122.67) .. controls (67.46,119.9) and (69.71,117.66) .. (72.48,117.66) -- (79.17,117.67) .. controls (81.94,117.67) and (84.18,119.92) .. (84.18,122.69) -- (84.14,162.96) .. controls (84.13,165.73) and (81.89,167.98) .. (79.12,167.97) -- (72.43,167.97) .. controls (69.66,167.96) and (67.42,165.72) .. (67.42,162.95) -- cycle ;
\draw  [fill={rgb, 255:red, 255; green, 234; blue, 195 }  ,fill opacity=1 ][line width=0.75]  (92.59,155.73) .. controls (92.59,152.96) and (94.84,150.72) .. (97.61,150.72) -- (104.29,150.73) .. controls (107.06,150.73) and (109.31,152.98) .. (109.3,155.75) -- (109.26,196.03) .. controls (109.26,198.8) and (107.01,201.04) .. (104.24,201.04) -- (97.55,201.03) .. controls (94.78,201.03) and (92.54,198.78) .. (92.54,196.01) -- cycle ;
\draw  [fill={rgb, 255:red, 255; green, 234; blue, 195 }  ,fill opacity=1 ][line width=0.75]  (67.39,189.23) .. controls (67.39,186.46) and (69.64,184.22) .. (72.41,184.22) -- (79.1,184.23) .. controls (81.87,184.23) and (84.11,186.48) .. (84.11,189.25) -- (84.07,229.53) .. controls (84.06,232.3) and (81.82,234.54) .. (79.05,234.54) -- (72.36,234.53) .. controls (69.59,234.53) and (67.35,232.28) .. (67.35,229.51) -- cycle ;
\draw [line width=0.75]    (50.66,260.9) -- (178.4,261.11) ;
\draw [line width=0.75]    (50.64,294.45) -- (178.38,294.65) ;
\draw  [fill={rgb, 255:red, 255; green, 234; blue, 195 }  ,fill opacity=1 ][line width=0.75]  (92.68,223.16) .. controls (92.68,220.39) and (94.93,218.15) .. (97.7,218.15) -- (104.39,218.16) .. controls (107.16,218.16) and (109.4,220.41) .. (109.4,223.18) -- (109.36,263.45) .. controls (109.35,266.22) and (107.11,268.47) .. (104.34,268.46) -- (97.65,268.46) .. controls (94.88,268.45) and (92.64,266.21) .. (92.64,263.44) -- cycle ;
\draw  [fill={rgb, 255:red, 245; green, 240; blue, 235 }  ,fill opacity=1 ][line width=0.75]  (67.49,257.16) .. controls (67.49,254.39) and (69.74,252.15) .. (72.51,252.15) -- (79.19,252.16) .. controls (81.96,252.16) and (84.21,254.41) .. (84.2,257.18) -- (84.16,297.45) .. controls (84.16,300.22) and (81.91,302.47) .. (79.14,302.46) -- (72.45,302.46) .. controls (69.68,302.45) and (67.44,300.21) .. (67.44,297.44) -- cycle ;
\draw  [fill={rgb, 255:red, 245; green, 166; blue, 35 }  ,fill opacity=1 ][line width=0.75]  (176.72,185.66) .. controls (176.72,185.66) and (176.72,185.66) .. (176.72,185.66) -- (193.44,185.68) .. controls (193.44,185.68) and (193.44,185.68) .. (193.44,185.68) -- (193.39,235.42) .. controls (193.39,235.42) and (193.39,235.42) .. (193.39,235.42) -- (176.67,235.4) .. controls (176.67,235.4) and (176.67,235.4) .. (176.67,235.4) -- cycle ;
\draw  [fill={rgb, 255:red, 245; green, 240; blue, 235 }  ,fill opacity=1 ][line width=0.75]  (118.98,55.58) .. controls (118.98,52.81) and (121.23,50.56) .. (124,50.57) -- (130.68,50.57) .. controls (133.45,50.58) and (135.7,52.82) .. (135.69,55.59) -- (135.65,95.87) .. controls (135.65,98.64) and (133.4,100.88) .. (130.63,100.88) -- (123.94,100.87) .. controls (121.17,100.87) and (118.93,98.62) .. (118.93,95.85) -- cycle ;
\draw  [fill={rgb, 255:red, 245; green, 240; blue, 235 }  ,fill opacity=1 ][line width=0.75]  (144.16,89.14) .. controls (144.16,86.37) and (146.41,84.13) .. (149.18,84.13) -- (155.86,84.14) .. controls (158.63,84.14) and (160.88,86.39) .. (160.87,89.16) -- (160.83,129.43) .. controls (160.83,132.2) and (158.58,134.44) .. (155.81,134.44) -- (149.12,134.43) .. controls (146.35,134.43) and (144.11,132.18) .. (144.11,129.41) -- cycle ;
\draw  [fill={rgb, 255:red, 245; green, 240; blue, 235 }  ,fill opacity=1 ][line width=0.75]  (118.96,123.07) .. controls (118.96,120.3) and (121.21,118.06) .. (123.98,118.06) -- (130.67,118.07) .. controls (133.44,118.07) and (135.68,120.32) .. (135.68,123.09) -- (135.64,163.36) .. controls (135.63,166.13) and (133.39,168.38) .. (130.62,168.37) -- (123.93,168.37) .. controls (121.16,168.36) and (118.92,166.12) .. (118.92,163.35) -- cycle ;
\draw  [fill={rgb, 255:red, 255; green, 234; blue, 195 }  ,fill opacity=1 ][line width=0.75]  (144.09,156.13) .. controls (144.09,153.36) and (146.34,151.12) .. (149.11,151.12) -- (155.79,151.13) .. controls (158.56,151.13) and (160.81,153.38) .. (160.8,156.15) -- (160.76,196.43) .. controls (160.76,199.2) and (158.51,201.44) .. (155.74,201.44) -- (149.05,201.43) .. controls (146.28,201.43) and (144.04,199.18) .. (144.04,196.41) -- cycle ;
\draw  [fill={rgb, 255:red, 255; green, 234; blue, 195 }  ,fill opacity=1 ][line width=0.75]  (118.89,189.23) .. controls (118.89,186.46) and (121.14,184.22) .. (123.91,184.22) -- (130.6,184.23) .. controls (133.37,184.23) and (135.61,186.48) .. (135.61,189.25) -- (135.57,229.53) .. controls (135.56,232.3) and (133.32,234.54) .. (130.55,234.54) -- (123.86,234.53) .. controls (121.09,234.53) and (118.85,232.28) .. (118.85,229.51) -- cycle ;
\draw  [fill={rgb, 255:red, 255; green, 234; blue, 195 }  ,fill opacity=1 ][line width=0.75]  (144.18,223.16) .. controls (144.18,220.39) and (146.43,218.15) .. (149.2,218.15) -- (155.89,218.16) .. controls (158.66,218.16) and (160.9,220.41) .. (160.9,223.18) -- (160.86,263.45) .. controls (160.85,266.22) and (158.61,268.47) .. (155.84,268.46) -- (149.15,268.46) .. controls (146.38,268.45) and (144.14,266.21) .. (144.14,263.44) -- cycle ;
\draw  [fill={rgb, 255:red, 245; green, 240; blue, 235 }  ,fill opacity=1 ][line width=0.75]  (118.99,257.16) .. controls (118.99,254.39) and (121.24,252.15) .. (124.01,252.15) -- (130.69,252.16) .. controls (133.46,252.16) and (135.71,254.41) .. (135.7,257.18) -- (135.66,297.45) .. controls (135.66,300.22) and (133.41,302.47) .. (130.64,302.46) -- (123.95,302.46) .. controls (121.18,302.45) and (118.94,300.21) .. (118.94,297.44) -- cycle ;
\draw [color={rgb, 255:red, 245; green, 166; blue, 35 }  ,draw opacity=1 ][fill={rgb, 255:red, 245; green, 166; blue, 35 }  ,fill opacity=1 ][line width=1.5]  [dash pattern={on 4.5pt off 2.25pt}]  (50,189.97) -- (171.87,189.97) ;
\draw [color={rgb, 255:red, 206; green, 206; blue, 206 }  ,draw opacity=1 ][fill={rgb, 255:red, 255; green, 255; blue, 255 }  ,fill opacity=1 ][line width=1.5]  [dash pattern={on 4.5pt off 2.25pt}]  (50.19,223.57) -- (145.42,223.57) ;
\draw [color={rgb, 255:red, 245; green, 166; blue, 35 }  ,draw opacity=1 ][line width=1.5]  [dash pattern={on 4.5pt off 2.25pt}]  (50.53,196.84) .. controls (53.42,196.84) and (117.35,197.01) .. (127.6,197.1) .. controls (137.85,197.19) and (141.35,223.19) .. (152.93,223.43) .. controls (164.52,223.67) and (169.52,223.62) .. (172.05,223.57) ;
\end{mytikz2}\quad.
\end{equation}
Note that the legal paths (e.g., the orange lines) corresponding to the straight wires (e.g., the upper orange line and the grey line) are possible to merge at some block due to the requirement of the definition of the path set. However, we can always bound the length of the legal path set by the length of the straight wires as the latter is always longer. The length of the chosen single path set $P_j$ can be upper bounded by
\begin{equation}
    l(P_j) \leq r \cdot (\chi-1) \cdot \log_4\left(\frac{4^\beta - 1}{4^1-1}\right) + l_0 \leq r (\chi-1) \beta + l_0,
\end{equation}
where $l_0$ denotes the length of the edge between $\rho_0$ and the head blocks of $P_j$. On the other hand, the head width of $P_j$ plus $2l_0$ is upper bounded by
\begin{equation}
    w(P_j) + 2l_0 = \sum_{B_k\in \mathrm{Head}(P_j)} \log_2\left( \frac{4^{|s_f(B_k)|}-1}{2^{|s_f(B_k)|}-1}\right) + \log_2\left(\frac{4^{|s(B_k)|}-1}{4^{|s_f(B_k)|}-1}\right) \leq r\cdot \log_2\left(\frac{4^\beta-1}{2^1-1}\right) \leq r\cdot 2\beta.
\end{equation}
Combining the two inequalities and substituting into Theorem~\ref{theorem:var_lower_bound_path_set} leads to
\begin{equation}\label{eq:var_bound_local_depth}
    \var_\mathbb{U} [C] \geq 2^{-2r(\chi-1)\beta-r\cdot2\beta} \sum_j \lambda_j^2 = 4^{-r\chi\beta} \|\bm{\lambda}\|_2^2.
\end{equation}
Therefore, the landscape fluctuation, i.e., the standard deviation of the normalized cost function, can be lower bounded by Eq.\,\eqref{eq:sigma_lower_bound}.
\end{proof}

Here we remark that the cost function is normalized by $\|\bm{\lambda}\|_1$ only approximately since we usually do not know the exact ground state energy $E_0$ but have $|E_0|\leq \|H\|_\infty \leq\|\bm{\lambda}\|_1$ for any traceless Hamiltonian $H$. For those traceless Hamiltonians consisting of commuting subterms, this normalization is exact in the sense that the global minimum is $-1$. Of course, one can also choose to exactly normalize the cost function by the ground state energy, i.e., $C(\bm{\theta})/|E_0|$, to obtain more accurate landscape fluctuations $\sigma$. However, as $E_0$ is not easy to estimate in general (which is actually one of the purposes of VQE), we usually use the approximate normalization via $\|\bm{\lambda}\|_1$.

For constant depth circuits and extensive Hamiltonians, the standard deviation at least scales like $\sigma\sim 1/\sqrt{N}$ for large $N$, which is consistent with the prediction of the central limit theorem for decoupled convex functions. This indicates that in locally scrambled constant depth circuits, the scaling with $N$ is the same for large $N$ regardless of whether the expressibility is sufficient and whether there are bad local minima, unlike in the case of BP where one can determine the existence of BP by only calculating the scaling with $N$. However, the existence of insufficient expressibility and bad local minima can still manifest themselves in the specific values of the landscape fluctuation.

One can see that the landscape fluctuation decays exponentially with the local depth $\chi$. On the other hand, we know that it will not decay endlessly for a given $N$ because the circuit will become a $2$-design approximately when the depth becomes comparable with $N$. Namely, the landscape fluctuation will finally converge to an exact known value. The converged variance of the cost value can be calculated by the global Haar integral, i.e.,
\begin{equation}\label{eq:var_converge}
\begin{aligned}
    \var_{\text{Haar}} [C] &= \int d\mu(U)~ \tr(\rho_0^{\otimes 2} U^{\dagger\otimes 2} H^{\otimes 2} U^{\otimes 2}) \\
    &= \frac{1}{2^{2N}-1}\left[\tr(H^{\otimes 2}) \tr(\rho_0^{\otimes 2}) + \tr(S H^{\otimes 2}) \tr(S \rho_0^{\otimes 2}) \right]
    \\
    &\quad - \frac{1}{2^N(2^{2N}-1)}\left[\tr(H^{\otimes 2}) \tr(S \rho_0^{\otimes 2}) + \tr(S H^{\otimes 2}) \tr(\rho_0^{\otimes 2}) \right] \\
    &= \frac{\tr(H)^2 \tr(\rho_0)^2 + \tr(H^{2}) \tr( \rho_0^{2})}{2^{2N}-1} - \frac{\tr(H)^2 \tr(\rho_0^{2}) + \tr(H^{2}) \tr(\rho_0)^2 }{2^N(2^{2N}-1)} \\
    &= \frac{\tr(H^{2})}{2^{N}(2^{N}+1)} = \frac{\|H\|_2^2}{2^{N}(2^{N}+1)} = \frac{\|\bm{\lambda}\|_2^2}{2^{N}+1}, 
\end{aligned}
\end{equation}
where we have used the conditions that the Hamiltonian is traceless and the initial state is a pure state, i.e., $\tr(H)=0$ and $\rho_0=\ketbrasame{\bm{0}}$. Hence, the converged value of the landscape fluctuation is
\begin{equation}
    \sigma_{\text{Haar}} = \frac{1}{\sqrt{2^{N}+1}} \frac{\|\bm{\lambda}\|_2}{\|\bm{\lambda}\|_1},
\end{equation}
which is consistent with Eq.\,\eqref{eq:sigma_lower_bound} and the linear depth condition $\chi\in \Omega(N)$ of the circuits forming global $2$-designs.

\section{Details on relative fluctuation}\label{appendix:details_RF}

\subsection{Definition and heuristic derivation of relative fluctuation}
Here we elaborate more on the definition and properties of relative fluctuation. The relative fluctuation (RF), denoted by $\tilde{\sigma}$, is defined as the ratio of the fluctuation (i.e., the standard deviation of the value of the normalized cost function) of the given training landscape $\sigma$ and that of certain standard learnable landscapes $\sigma_0$, i.e.,
\begin{equation}
    \tilde{\sigma} = \frac{\sigma}{\sigma_0},\quad \sigma=\frac{\var_{\Theta}[C]}{\|\bm{\lambda}\|_1},
\end{equation}
so that one can identify the learnability by the magnitude of RF with respective to $1$. We remark that utilizing the information in the landscape fluctuation $\sigma$ is natural and intuitive as typical syndromes of low learnability, such as barren plateaus, insufficient expressibility, and bad local minima, all tend to have their cost values concentrated around their averages, albeit to varying degrees. That is to say, for landscapes with barren plateaus, $\sigma$ is exponentially small in the number of qubits while for landscapes with bad local minima or insufficient expressibility, $\sigma$ may have no exponentially small scaling but the specific value of $\sigma$ can be much smaller than some standard value. 

We remark that $\sigma$ might seem a bit cliche at first glance because the exponentially small variance of cost derivatives $\var_{\Theta} [\partial_\mu C]$ has been used to quantify barren plateaus~\cite{McClean2018} which is equivalent to the exponential cost concentration~\cite{Arrasmith2021, Miao2024, Perez-Salinas2024}. However, besides exponential scaling implying barren plateaus, the specific values of $\sigma$ may also encode the information of insufficient expressibility and bad local minima. In addition, $\sigma$ naturally encodes the global property of the whole landscape while $\var_{\Theta} [\partial_\mu C]$ describes more the effects of individual parameters.

Importantly, we point out that $\sigma$ alone is not enough to characterize learnability. Considering an extreme case: for a fully parametrized circuit of exponential depth, the learnability should be strong (though the parameter count is impractical and not scalable). However, $\sigma$ is exponentially small since the circuit forms a global unitary $2$-design. The underlying reason for this contradiction is that $\sigma$ does not correctly account for the contribution from the dimensionality of the expressive space, resulting in an unfair comparison of the learnability of ansatzes with different amounts of parameters. This is also why we need to introduce a set of standard learnable landscapes, which can take the effect of parameter counts into consideration.

Thus, the key point is how to define the fluctuation $\sigma_0$ of the ``standard learnable landscapes'' $C_0(\bm{\theta})$ properly. We will justify the form of $C_0(\bm{\theta})$ through several fundamental assumptions.
\begin{itemize}
    \item Firstly, to ensure good trainability, $C_0(\bm{\theta})$ should be convex. 
    \item Secondly, to ensure sufficient expressibility, there should exist a global minimum $\bm{\theta}^*$ such that $\ket{\psi(\bm{\theta}^*)}$ is the true ground state, which can be expressed as 
    \begin{equation}
        C_0(\bm{\theta}^*) = -1,
    \end{equation}
    given that the function $C_0(\bm{\theta})$ is normalized already.
    \item Thirdly, to define $C_0(\bm{\theta})$ for arbitrary dimensions of parameter space, $C_0(\bm{\theta})$ with different numbers of parameters should have a similar form. A simple method is to define them recursively, i.e., constructing $C_0(\bm{\theta})$ with $M$ parameters based on $C_0(\bm{\theta})$ with $(M-1)$ parameters. Combined with the assumption of convexity, this recursive operation can be chosen as the addition, i.e., 
    \begin{equation}
        C_0(\bm{\theta})=\sum_\mu f_\mu(\theta_\mu),
    \end{equation}
    where $f_\mu$ is a single-variable convex function. This is because the addition operation preserves the convexity for convex functions of independent variables. We call such type of cost functions $C_0(\bm{\theta})$ as ``decoupled convex functions'', since there is no coupling term like $g_{\mu\nu}(\theta_\mu, \theta_\nu)$.
    \item Finally, incorporating the common feature of the cost functions in VQAs, i.e., the parameter-shift rule which says the cross-sections of the cost function along the axial directions are always sine-shaped~\cite{Cerezo2021a}, the single-variable function $f_\mu(\theta_\mu)$ should take the form of
    \begin{equation}
        f_\mu(\theta_\mu)=a_\mu \sin(\theta_\mu+\phi_\mu)+b_\mu,
    \end{equation}
    where $a_\mu$, $\phi_\mu$ and $b_\mu$ are undetermined real numbers. We point out that such a sine function $f_\mu$ is in fact not strictly convex but it does only have a unique minimum up to the periodicity, so the trainability of $f_\mu$ is as good as a strict convex function. Combined with the assumption of sufficient expressibility, the coefficients should satisfy $\sum_\mu (-|a_\mu|+ b_\mu)=-1$.
\end{itemize}
Therefore, the standard learnable landscapes should take the form of decoupled sine functions whose global minimum is $-1$. To give a rough estimation for the landscape fluctuation of $C_0(\bm{\theta})$, we assume the coefficient $a_\mu$ is uniform. The phase shift $\phi_\mu$ and the constant $b_\mu$ are set to be zero since they do not affect the landscape fluctuation. Thus, we obtain a reduced form of the standard learnable normalized cost function
\begin{equation}\label{eq:c0_sin}
    C_0(\bm{\theta}) = \frac{1}{M} \sum_{\mu=1}^{M} \sin\theta_\mu.
\end{equation}
The variance of the function value of a single sine function is
\begin{equation}
    \var_{\theta_\mu\sim \text{uniform}[0,2\pi)}[\sin\theta_\mu] = \int_{0}^{2\pi} \frac{d\theta_\mu}{2\pi} \sin^2\theta_\mu = \frac{1}{2},
\end{equation}
so that the corresponding standard deviation is just
\begin{equation}
    \sigma_0^{(1)}=\frac{1}{\sqrt{2}}.
\end{equation}
If we regard $\{\sin\theta_\mu\}_{\mu=1}^{M}$ as a set of random variables that are independent and identically distributed, $C_0(\bm{\theta})$ in Eq.\,\eqref{eq:c0_sin} is just the average of these i.i.d. random variables. According to the central limit theorem, the distribution of the cost value of $C_0(\bm{\theta})$ converges to the Gaussian distribution whose standard deviation is
\begin{equation}\label{eq:sigma0_2M}
    \sigma_0 = \frac{\sigma_0^{(1)}}{\sqrt{M}} = \frac{1}{\sqrt{2M}}.
\end{equation}
This is exactly the reference line that is used repeatedly in the main text. This kind of cost function $C_0(\bm{\theta})$ can be actually realized by the case of product-state learning, e.g., using one layer of single-qubit rotations $R_y^{\otimes N}$ to learn the ground state of a trivial polarization Hamiltonian $H=-\sum_j X_j$. Therefore, combined with the definition of RF, we have $\tilde{\sigma}=\sqrt{2M}\sigma$.

However, the expression in Eq.\,\eqref{eq:sigma0_2M} is still not complete because the special form of Eq.\,\eqref{eq:c0_sin} hides the fact that the ``parameter count'' that should go into Eq.\,\eqref{eq:sigma0_2M} is the dimension of the expressive space of quantum states, instead of the dimension of the actual parameter space. Such a difference can be revealed by the following extreme example. Consider a PQC composed of repeated single-qubit rotation layers, whose expressive space is limited to all product states. If the true ground state is entangled, the learnability should be low. However, since one can arbitrarily increase the number of repeated layers, $\sqrt{2M}\sigma$ can be arbitrarily large. The underlying reason for this contradiction is that the dimension that is really helpful to increase expressibility is the dimension of the expressive space of quantum states instead of the plain dimension of the parameter space. Overparametrization may lead to good trainability but does not always lead to good learnability considering the possibility of insufficient expressibility.

To incorporate this fact into the definition of RF rigorously, we utilize the notion of the effective (quantum) dimension $M_{\text{eff}}$ of the expressive space (or say manifold) of a parametrized quantum state~\cite{Haug2021a}, which is defined by the rank of the quantum Fisher information (QFI) matrix $\mathcal{F}(\bm{\theta})$, i.e.,
\begin{equation}
    M_{\text{eff}} = \max_{\bm{\theta}} \{\opr{rank}[\mathcal{F}(\bm{\theta})] \},
\end{equation}
where the QFI matrix for the parametrized quantum state $\ket{\psi(\bm{\theta})}$ is defined by~\cite{Meyer2021}
\begin{equation}
    \mathcal{F}_{\mu\nu} = 4\opr{Re}[\braket{\partial_\mu \psi}{\partial_\nu \psi} - \braket{\partial_\mu \psi}{\psi} \braket{\psi}{\partial_\nu \psi} ].
\end{equation}
Here we use the maximization over all $\bm{\theta}$ to avoid those special parameter points of measure zero whose QFI rank is lower than most of the others, such as the poles of the Bloch sphere~\cite{Haug2021a}. In fact, one can easily get $M_{\text{eff}}$ by computing the QFI ranks at several random parameter points instead of maximization. Importantly, it holds that $M_{\text{eff}}\leq M$ and $M_{\text{eff}}$ will reduce to the plain dimension of the parameter space $M$ for underparametrization scenarios such as $C_0(\bm{\theta})$. Actually, overparametrization is defined by the saturation of $M_{\text{eff}}$ when increasing $M$~\cite{Larocca2023}. Therefore, the proper definition of the fluctuation of the standard learnable landscapes should be modified from Eq.\,\eqref{eq:sigma0_2M} into
\begin{equation}\label{eq:sigma0_2Meff}
    \sigma_0 = \frac{1}{\sqrt{2M_{\text{eff}}}}.
\end{equation}
Thus the rigorous definition of RF is
\begin{equation}
    \tilde{\sigma} = \frac{\sigma}{\sigma_0} = \sqrt{2M_{\text{eff}}}\sigma.
\end{equation}
One can see that generally speaking, RF is the landscape fluctuation multiplied by a correction factor originating from the dimensionality of the expressive space.

Here we make some clarification on the definitions of $\sigma_0$ and $\tilde{\sigma}$. Firstly, the use of the parameter-shift rule above is only a way of heuristic derivation of RF but not a necessary condition for RF to be effective. Our numerical experiments support that RF can still be an effective metric for learnability even in the case where the parameter-shift rule does not hold rigorously, such as PQCs with correlating parameters. Secondly, $\sigma = \sigma_0$, or equivalently $\tilde{\sigma}=1$, is not an absolute criterion separating learnable and unlearnable phases with theoretical guarantee, but a reference line of practical use to predict and compare the learnability of different PQCs. It is completely possible to use some advanced structure-specific optimization strategies to enhance the training performance even if $\tilde{\sigma}$ is a bit smaller than $1$. Thirdly, the $\sqrt{M}$ scaling can also consistently arise when the average $2$-norm of the cost gradient $\mathbb{E}_{\Theta}[\|\nabla C\|_2]$ is estimated by its upper bound from Jensen's inequality
\begin{equation}
    \mathbb{E}_{\Theta}[\|\nabla C\|_2] \leq \sqrt{\sum_{\mu=1}^M \mathbb{E}_{\Theta}[(\partial_\mu C)^2]},
\end{equation}
and $\mathbb{E}_{\Theta}[(\partial_\mu C)^2]$ is approximately the same constant for all tunable parameters.

In the main text, we claim that global overparametrization occurs around the intersection point of the reference line $\sigma=\sigma_0$ and the converged horizontal line $\sigma=\sigma_{\text{Haar}}$. This can be easily checked by calculating the number of parameters $M_{\text{cross}}$ at the intersection point and comparing it with the total number of parameters of a fully parametrized $N$-qubit quantum state $M_{\text{full}} = 2^{N+1} - 2$ (where the constraints from the normalization and global phase have been considered). Let $\sigma_0=\sigma_{\text{Haar}}$ and we have
\begin{equation}
    M_{\text{cross}} = \frac{1}{2\sigma_{\text{Haar}}^2} = \frac{\|\bm{\lambda}\|_1^2}{\|\bm{\lambda}\|_2^2} \left(2^{N-1}+\frac{1}{2}\right),
\end{equation}
which has almost the same scaling as $M_{\text{full}}$ except for a Hamiltonian-dependent correction $\|\bm{\lambda}\|_1^2/\|\bm{\lambda}\|_2^2$ which grows at most polynomially in $N$ for common physical Hamiltonians.

As RF only involves second-order moments of random circuits, it can be computed on classical computers efficiently and parallelly via random Clifford circuits under the following conditions.
\begin{itemize}
    \item Firstly, the PQC should be composed of independent Pauli rotation gates (such as $e^{-iX \theta/2}$ and $e^{-iX\otimes X\theta/2}$) and other Clifford gates involving no tunable parameters (such as $\opr{CX}$, $\opr{CZ}$ and $\opr{SWAP}$). Sampling the Pauli rotation angles from the discrete set $\{0,\pi/2,\pi,3\pi/2\}$ will be equivalent to sampling them from $[0,2\pi)$ at least up to the second moment, where the former is classically efficient by the stabilizer formalism. Alternatively, if a gate or block is fully parametrized, such as the $\mathrm{SU}(4)$ Cartan decomposition of a two-qubit gate, one can also sample the Haar-random Clifford gates instead of sampling Pauli angles. In the next section, we will verify the consistency of the two sampling methods (i.e., Clifford sampling and continuous-angle sampling). By contrast, if the PQC contains correlating tunable parameters over different Pauli rotation gates or non-Clifford gates whose generators are not Pauli strings, two sampling methods will be nonequivalent. Concomitantly, one cannot compute RFs efficiently on classical computers via Clifford sampling, though in these cases RF is still an effective metric and one can estimate RFs efficiently on quantum devices. In particular, for those PQCs that contain non-Clifford gates such as the $T$ gate, one can choose to relax them to Pauli rotations or combinations of Pauli rotations to assess RFs approximately via Clifford sampling.
    \item Secondly, the initial state should be a stabilizer state. For many VQAs such as standard VQE and QAOA, the initial state is just $\ket{0}^{\otimes N}$ or $\ket{+}^{\otimes N}$ which is indeed a stabilizer state and hence RF can be computed classically via the stabilizer formalism given Clifford circuits. However, there are also scenarios in VQAs where the initial state is not a stabilizer state. For example, in quantum machine learning (QML) tasks that involve data encoding in the input~\cite{Cerezo2021}, the initial state (the state after data encoding) is generally not a stabilizer state. In such cases, one can only estimate RFs efficiently on quantum devices.
\end{itemize}
The classical simulability of RF may indicate that though the actual information in the ground state might be difficult to learn, the information on whether the ground state is easy to learn by the given PQC might be much easier to extract. We remark that $M_{\text{eff}}$ can be simply replaced by $M$ in most common cases without loss of effectiveness, especially in underparametrization scenarios. And $M_{\text{eff}}$ can be also efficiently estimated on classical computers through tools from DLA~\cite{Larocca2022, Larocca2023} or the high-order parameter-shift rule~\cite{Mari2021} together with Clifford circuits.

It is also possible to design automatic searching algorithms by taking RF itself as an objective function and utilizing methods like in quantum architecture search~\cite{Ostaszewski2019, Zhang2020b_z, Du2020a_z, Zhang2021_z}. Akin to coarse and fine-tuning, this leads to a promising scheme that first uses classical simulable metrics like RF to search for suitable PQCs for a given problem on classical computers and then trains the selected PQCs on quantum devices to further search for the target state.

\subsection{Gaussian-like condition and other moments}
A natural question is why high RFs can suggest better training performance (smaller training error/higher training accuracy/higher learnability), as witnessed by all the instances in this paper. Here we provide an intuitive explanation. The correlation between RFs and training performance holds true based on the necessary condition that the probability density of quantum states in cost value
\begin{equation}
    \varrho(x) = \opr{Pr}_{\Theta}[C(\bm{\theta})=x],
\end{equation}
is Gaussian-like with zero average, which means that the second moment (variance/landscape fluctuation/RF) can determine the distribution $\varrho$. If setting aside the effect of expressive space dimension for a moment, we can say this condition ensures that a larger variance will lead to a more spread distribution and larger accessibility to the low-energy states of the Hamiltonian, implying a better chance of successful training. For the numerical experiments conducted in this paper, we find that all the corresponding distributions indeed have Gaussian-like shapes with zero averages, which partially explains the prediction power of RFs. In addition, for deep circuits, there are also proofs~\cite{Garcia-Martin2023} stating that $\varrho$ indeed follows Gaussian distributions under reasonable conditions.

On the contrary, the correlation between RFs and training performance may break in the following two scenarios.
\begin{itemize}
    \item The first scenario is that the average of the cost value $\mathbb{E}_{\Theta}[C(\bm{\theta})]$ deviates from zero distinctly. Consider an extreme case where the initial state is already very close to the ground state and the action of the PQC is almost trivial. Then the variance of the cost value should be almost zero, but the average is close to the ground state energy. Hence the final training error would be small, which contradicts the small variance. More generally, if $\varrho$ has a small variance but has a distinct bias to the ground state, the corresponding training performance is not necessarily worse than those with large variances and unbiased averages. In this paper, we directly focus on the variance of the cost value without paying much attention to the average because the average usually vanishes for common PQCs. Given a traceless Hamiltonian, one layer of single-qubit rotations $R_3^{\otimes N}$ is enough to make the average become zero. If the PQC is so special or with so few parameters that the average of the cost value is not close to zero, it is indeed necessary to consider both the average and variance to assess learnability. 
    
    \item The second scenario is that the average is zero, but the probability distribution $\varrho$ deviates from the Gaussian-like symmetric shape distinctly. For example, suppose $\varrho$ is asymmetric regarding the average and has a larger tail on the negative side. In that case, there will be extra accessibility to the low-energy states that does not manifest itself on the second-order moment. This can be detected by computing other higher-order moments such as the third-order moment $\mathbb{E}_{\Theta}[C(\bm{\theta})^3]$, which can quantify the asymmetry of the distribution. That is to say, if $\varrho$ has a relatively small variance but has a distinctly negative third-order moment, the corresponding training performance is not necessarily worse than those with large variances and vanishing third-order moments. Moreover, for those PQCs composed of locally scrambled blocks such as $R_3$ gates and $\mathrm{SU}(4)$ fully parametrized two-qubit gates, the third-order moments can also be efficiently computed by Clifford sampling on classical computers as the Clifford group forms a unitary $3$-design for the corresponding Haar distribution.
\end{itemize}

\section{Additional numerical results and details}\label{appendix:additional}

In this section, we elaborate on the details of the numerical experiments in the main text and provide some additional numerical results to support our conclusions further. Throughout this paper, the numerical experiments are based on the open-source library TensorCircuit~\cite{Zhang2022_z} and PyClifford~\cite{Hu2023} for general circuit simulation with variational training, and Clifford circuit simulation via stabilizer formalism, respectively. 

First, we will clarify some basic details about the experiments shown in the main text. For the RF estimation experiments in Figs.~\textcolor{blue}{2}, \textcolor{blue}{3}, \textcolor{blue}{4}, and \textcolor{blue}{5} of the main text, we typically use $2000$ samples for most cases and use $20000$ samples for cases in Fig.~\textcolor{blue}{2} with large qubit counts ($N=12, 16$) and deep circuit depths ($D=64, 128, 256$). All these samples are then divided into $10$ batches, allowing us to estimate the standard deviations of RF as error bars. The training experiment in Fig.~\textcolor{blue}{3} is conducted using the Adam optimizer with a learning rate $0.1$ and initial parameters uniformly sampled from $[0, 2\pi)$. The training error is averaged over $10$ independent repeated training instances. The error bar represents the standard deviation over the instances. 

\subsection{Verification of Clifford sampling}
In the following, we give a more detailed explanation of the numerical experiments in the main text with additional results. First, we claim in the main text that sampling the Pauli rotation angles from the discrete set $\{0,\pi/2,\pi,3\pi/2\}$ will be equivalent to sampling them from the continuous interval $[0,2\pi)$ in terms of computing RFs so that one can efficiently compute RFs classically via the stabilizer formalism. This can be simply proved by checking the definition of unitary $t$-designs for the $\mathrm{U}(1)$ group, which we dub ``Pauli rotation $t$-design'' here. Take the $R_z$ rotation gate for an example. Denote $P_{t,t}(V)$ as a polynomial of degree at most $t$ in the entries of some matrix $V$ and $V^\dagger$. The continuous integral gives
\begin{equation}\label{eq:continuous}
\begin{aligned}
    \int_{0}^{4\pi} \frac{d\theta}{4\pi} P_{t,t}[R_z(\theta)] &= \int_{0}^{4\pi} \frac{d\theta}{4\pi} \sum_{k=-2t}^{2t}\lambda_{k} e^{-ik\theta/2} = \sum_{k=-2t}^{2t}\lambda_{k} \int_{0}^{4\pi} \frac{d\theta}{4\pi} e^{-ik\theta/2} = \lambda_{0},
\end{aligned}
\end{equation}
where $\lambda_k$ denotes the coefficients in the polynomial. Suppose $4\pi$ is divisible by the step length $\Delta\theta$ and set $\mathbb{S} = \{0,\Delta\theta,2\Delta\theta,\ldots,4\pi-\Delta\theta\}$. The discrete summation gives
\begin{equation}\label{eq:discrete}
\begin{aligned}
    \frac{1}{|\mathbb{S}|}\sum_{\theta\in \mathbb{S}} P_{t,t}[R_z(\theta)] &= \frac{1}{|\mathbb{S}|}\sum_{\theta\in \mathbb{S}} \sum_{k=-2t}^{2t}\lambda_{k} e^{-ik\theta/2} = \frac{1}{|\mathbb{S}|} \sum_{k=-2t}^{2t}\lambda_{k} \sum_{\theta\in\mathbb{S}} e^{-ik\theta/2}.
\end{aligned}
\end{equation}
If one wants get the same result from Eq.\,\eqref{eq:discrete} as in Eq.\,\eqref{eq:continuous}, the terms with $k\neq 0$ should all vanish. The critical condition that gives the maximum $\Delta\theta$ is $2t\cdot\Delta\theta/2=\pi$, i.e., $\Delta\theta=\pi/t$. This indicates that the step length of the discrete Pauli angle set for a Pauli rotation $t$-design should not be shorter than $\pi/t$. When $t=2$, the discrete set is just $\{0, \pi/2, \pi, 3\pi/2, 2\pi, 5\pi/2, 3\pi, 7\pi/2\}$. Furthermore, since the global phase does not affect the expectation value of observables, one just needs to sample from $\{0, \pi/2, \pi, 3\pi/2\}$ to compute landscape fluctuations. We numerically verify this equivalence in Fig.~\ref{fig:clifford_continuous}, where we randomly choose some PQCs and perform the two sampling methods independently. Especially, for the $\mathrm{SU}(4)$ full parametrization, we directly use the two-qubit random Clifford gate that matches the corresponding Haar measure. One can find that all the results from the two sampling methods coincide precisely within the statistical errors. 

\begin{figure}
    \centering
    \includegraphics[width=0.7\linewidth]{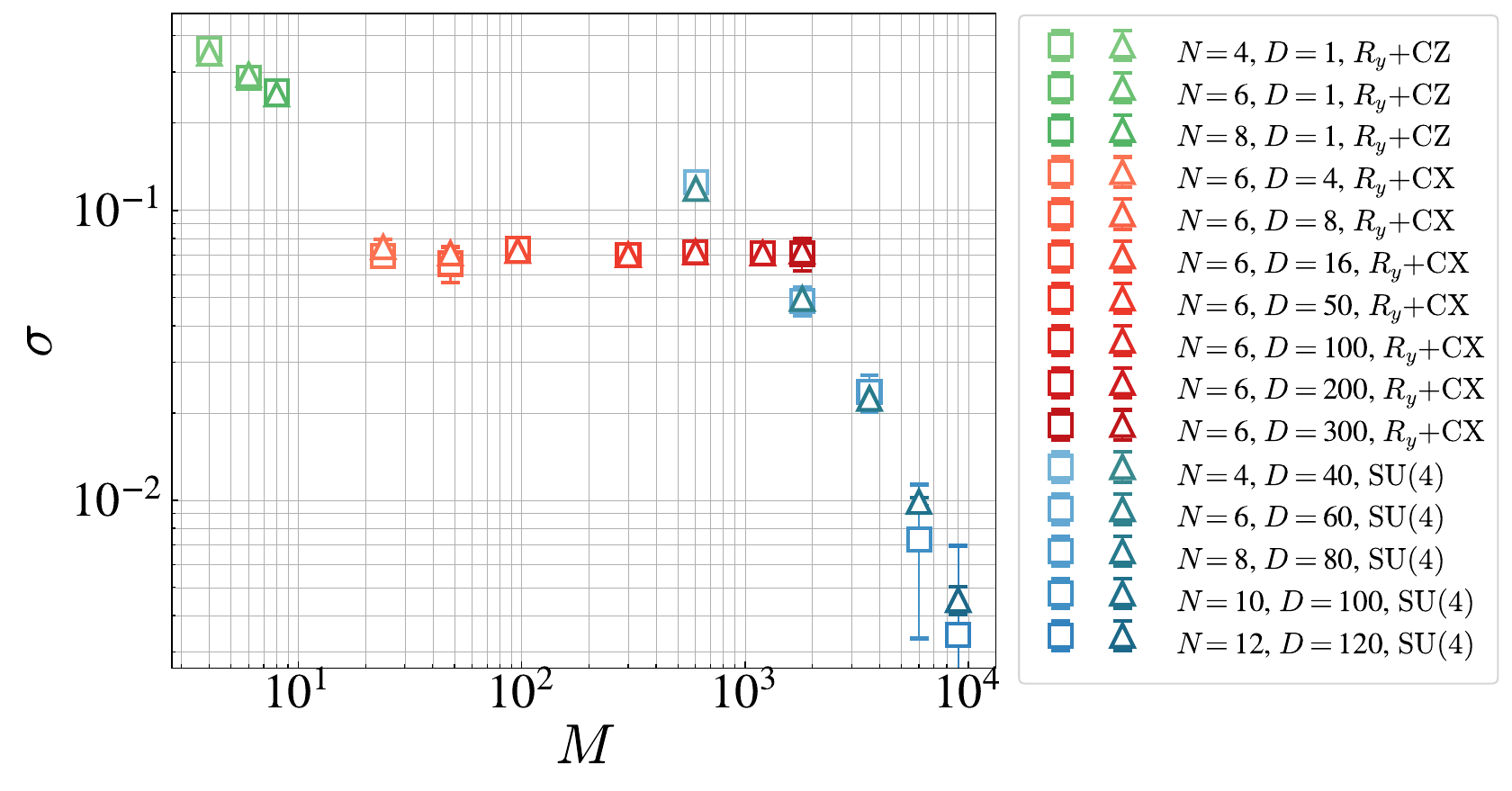}
    \caption{The landscape fluctuation $\sigma$ vs parameter count $M$ for randomly chosen circuits of different depth $D$ and qubit count $N$ obtained from the two sampling methods, i.e., sampling from continuous rotation angles (triangle marker) and the discrete Clifford rotation angles (square marker). The number of samples is $2000$ for all data points. The Hamiltonian is the 1D cluster model and the initial state is $\ket{\bm{0}}$. $R_y$+$\opr{CZ}$ means that each layer contains a layer of $R_y$ gates and a brickwall layer of $\opr{CZ}$ gates. Others are defined similarly. $\mathrm{SU}(4)$ means the brickwall circuit composed of the fully parameterized two-qubit gates.}
    \label{fig:clifford_continuous}
\end{figure}

\subsection{Details and extensions of the experiments on insufficient expressibility}
In Figs.~\textcolor{blue}{2} and \textcolor{blue}{3}, the Hamiltonian is chosen as the 1D transverse-field cluster model
\begin{equation}
    H_{ZXZ} = -\sum_{j=1}^{N} Z_{j-1} X_{j} Z_{j+1} - h\sum_{j=1}^{N} X_j,
\end{equation}
at $h=0$ with open boundary conditions. Namely, the boundary terms are just $X_1Z_2$ and $Z_{N-1}X_N$. The exact ground state in the limit of $h=0$ (the 1D graph state) can be generated via the circuit with a layer of Hadamard gates followed by a layer of CZ gates acting on the initial state $\ket{\bm{0}}$, i.e.,
\begin{equation}
\begin{quantikz}[row sep={0.7cm, between origins}, column sep=0.4cm]
    & \gate{H} & \ctrl{1}& & &\\
    & \gate{H} & \control{} &\ctrl{1} & & \\
    & \gate{H} & \ctrl{1} &\control{}  & & \\
    & \gate{H} & \control{} &\ctrl{1} & & \\
    & \gate{H} & \ctrl{1} &\control{}  & & \\
    & \gate{H} & \control{} & & &
\end{quantikz}\quad.
\end{equation}
Hence the ansatz circuit shape containing the ground state can be chosen as
\begin{equation}
\begin{quantikz}[row sep={0.7cm, between origins}, column sep=0.4cm]
    &\gate[2]{~}& & &\\
    & &\gate[2]{~} & & \\
    &\gate[2]{~}& & & \\
    & &\gate[2]{~} & & \\
    &\gate[2]{~}& & & \\
    & & & &
\end{quantikz}\quad,
\end{equation}
which is exactly the 1D brickwall circuit with open boundary conditions. If the two-qubit block is the $\mathrm{SU}(4)$ Cartan decomposition or $R_y$-$\opr{CZ}$ (the naming rules of ansatzes can be found around Eq.\,\eqref{eq:ansatz_template}), the expressive space of the PQC will contain the exact ground state at $h=0$. As shown in Fig.~\textcolor{blue}{3} of the main text, RF again roughly preserves the order of the learnability of the PQCs (higher RFs herald smaller training errors). In particular, since the true ground state at $h=0.5$ is not guaranteed to be contained by the expressive space of $R_y$-$\opr{CZ}$, the corresponding RF becomes lower and the training error deviates slightly from zero.



\begin{figure}
    \centering
    \includegraphics[width=0.8\linewidth]{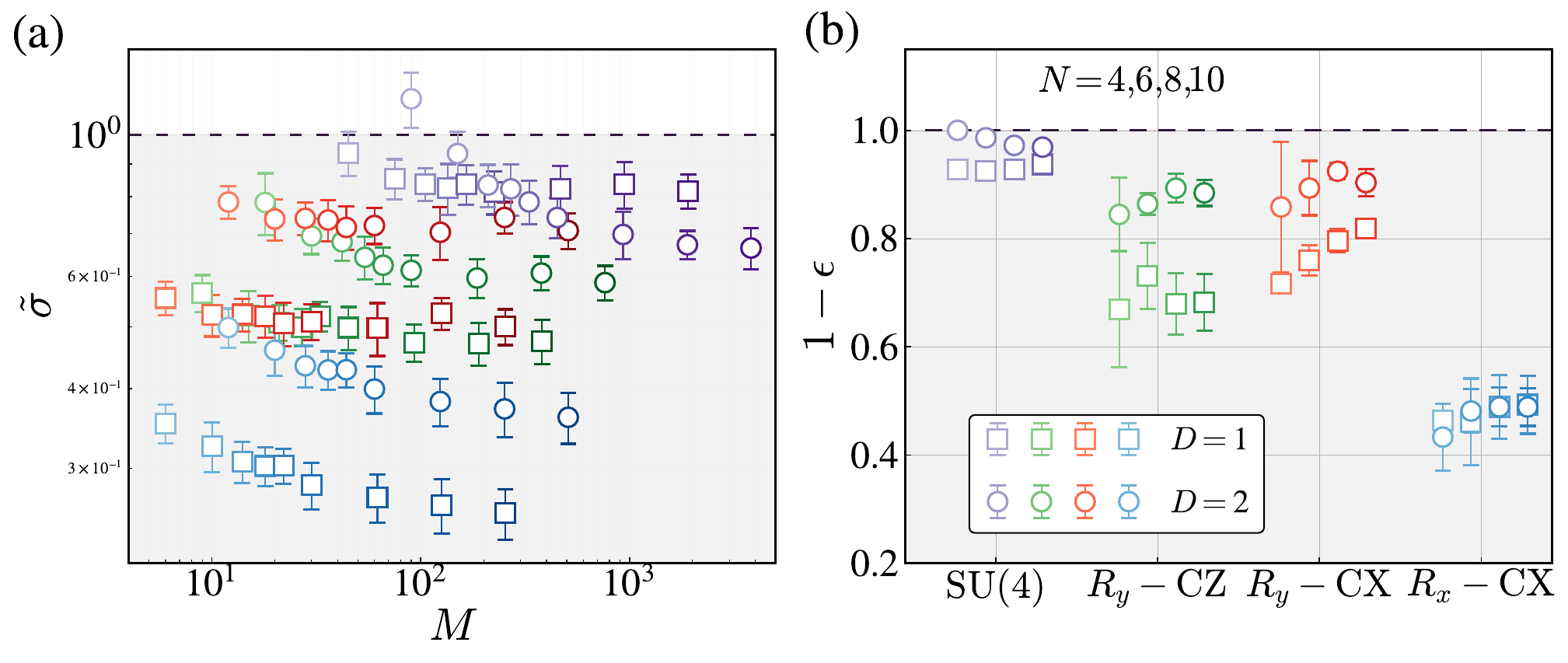}
    \caption{(a) The relative fluctuation $\tilde{\sigma}$ vs parameter count $M$ for 1D brickwall circuits with different gate choices and depths $D$ regarding the 1D antiferromagnetic Heisenberg model. The colors with increasing intensity correspond to $N=4,6,8,10,12,16,32,64,128$. (b) The actual training performance of these circuits with $\epsilon$ representing the training error after convergence. }
    \label{fig:heisenberg0}
\end{figure}

To further showcase the correlation between RFs and actual training accuracies, we use the 1D antiferromagnetic Heisenberg model
\begin{equation}
    H_{\text{Heisenberg}} = \sum_{j=1}^{N-1} X_{j} X_{j+1} + Y_{j} Y_{j+1} + Z_{j} Z_{j+1} ,
\end{equation}
with open boundary conditions. As shown in Fig.~\ref{fig:heisenberg0}, higher RFs again correspond to smaller training errors. In particular, different from the transverse-field cluster model, the performance of $R_y$-$\opr{CZ}$ becomes worse as there is no more guarantee for it to express the ground state even approximately.

We point out that the correlation between RFs and training performance does manifest itself in expressibility like in Fig.~\textcolor{blue}{3} of the main text and Fig.~\ref{fig:heisenberg0} not only explicitly, but also implicitly in other aspects, i.e., barren plateaus, bad local minima and polynomial overparameterization like in Figs.~\textcolor{blue}{2}, \textcolor{blue}{4} and \textcolor{blue}{5} of the main text, where the training errors are not plotted concomitantly but have been thoroughly investigated in the previous literature. For example, in the presence of barren plateaus or bad local minima, we have $\tilde{\sigma}\ll 1$ and the training errors deviate from zero significantly~\cite{McClean2018, Cerezo2021, Anschuetz2022}, while in the case of overparametrization, we have $\tilde{\sigma}\gtrsim 1$ and the training errors are almost zero~\cite{Larocca2022, Larocca2023}.

\subsection{Explanation and extensions for the experiments on local minima}
In Fig.~\textcolor{blue}{4} of the main text, we investigate the influence of the circuit depth on the extent of bad local minima using random Hamiltonians $H_{\text{rand}}$, where Fig.~\textcolor{blue}{4}(b) shows that as $N$ increases, RF stays above $1$ for $D=1$ while falls below $1$ for $D=4$, signifying bad local minima arise for $D=4$ when $N$ become large. Fig.~\textcolor{blue}{4}(a) shows that at a fixed $N$, RF indeed decreases gradually from $D=1$ to $D=4$. The critical depth is around $D_c \approx 2$. Moreover, compared to $H_{\text{rand}}$, the result of the trivial polarization Hamiltonian $H_Z$ shows a larger critical depth around $D_c \approx 4$, i.e., bad local minima come later than those of $H_{\text{rand}}$. This is because $H_Z$ is more local than $H_{\text{rand}}$ on average, resulting in narrower backward causal cones and larger local overparametrization ratio $\gamma$. In fact, this can also be seen from the experimental details. We know that $H_{\text{rand}}$ is obtained by the random back evolution $\mathbf{V}$ of depth $D$ from $H_{Z}$. $\mathbf{V}$ and the PQC $\mathbf{U}$ form a deeper unitary $\mathbf{V}\mathbf{U}$ of depth about $2D$. So mathematically, the landscape fluctuation for $H_{\text{rand}}$ and depth $D$ on average should be almost equal to that for $H_{Z}$ and depth $2D$. This can also be seen in Fig.~\textcolor{blue}{4}(a). However, the parameter counts $M$ for the two cases are different, i.e., the former is only about half of the latter, resulting in different data points on the $\sigma$-$M$ diagram, and concomitantly different RFs.

The phenomenon that more local Hamiltonian results in milder bad local minima can also be shown by directly comparing RFs of different physical models. In Fig.~\textcolor{blue}{3} of the main text, it is observed that RFs of the cluster model fall below $1$ at around $N\geq 64$ for 1D brickwall circuits of $D=1,2$. The locality of the cluster model is $r=3$, i.e., it is a $3$-local Hamiltonian. As a comparison, we conduct the same experiments for the transverse-field Ising (TFI) model 
\begin{equation}
    H_{\text{TFI}} = -\sum_{j=1}^{N-1} Z_jZ_{j+1} - h\sum_{j=1}^{N} X_j,
\end{equation}
of locality $r=2$. As shown in Figs.~\ref{fig:isingh0} and \ref{fig:isingh05}, RFs for the brickwall circuits stay above $1$ even till $N=128$, by contrast to the case of the cluster model. This dependence on the Hamiltonian locality also partially manifests itself on the lower bound in Theorem.~\ref{theorem:var_lower_bound_path_set} in the factor $2^{-r\chi\beta}$ that decays exponentially with the (non-)locality $r$.

The main text mentions that RF is insensitive to good local minima. This is actually intuitive as depicted in Fig.~\textcolor{blue}{1}(b) of the main text. Here we verify it by considering the symmetry-breaking states and the Greenberger–Horne–Zeilinger-type (GHZ) symmetric ground states in the ferromagnetic phase of the TFI model. At $h=0$, the system has two degenerate ground states $\ket{\bm{0}}=\ket{\uparrow}^{\otimes N}$ and $\ket{\bm{1}}=\ket{\downarrow}^{\otimes N}$. Of course, any linear combination of these two states is also a ground state at $h=0$. Especially, the GHZ state
\begin{equation}
    \ket{\text{GHZ}} = \frac{\ket{\bm{0}}+\ket{\bm{1}}}{\sqrt{2}},
\end{equation}
is also a ground state at $h=0$. The difference is that the GHZ state preserves the global symmetry $\prod_j X_j$ of the original Hamiltonian, while the two product states are not, hence termed as ``symmetry-breaking'' states. If we apply a small external field $h\neq 0$, the strict ground state will not be two-fold degenerate anymore for finite-size systems. Instead, the symmetric GHZ-type state is the unique ground state, though not in the exact form of $\ket{\text{GHZ}}$. However, the energy difference between the symmetric and symmetry-breaking states is small and even decreases exponentially with $N$. In the language of the training landscape, the symmetry-breaking states can be seen as two good local minima with respect to the global minimum corresponding to the true symmetric ground state. Hence we conduct additional experiments on good local minima using this TFI model. Specifically, the GHZ state can be generated by the circuit
\begin{equation}
\begin{quantikz}[row sep={0.7cm, between origins}, column sep=0.4cm]
    &\gate{H}&\ctrl{1}& & & & &\\
    &        &\targ{} &\ctrl{1} & & & &\\
    &&        &\targ{}  &\ctrl{1} & & &\\
    &&&        &\targ{}  &\ctrl{1} & &\\
    &&&&        &\targ{}  &\ctrl{1} &\\
    &&&&&        &\targ{}  &
\end{quantikz}\quad.
\end{equation}
Similarly, the GHZ-type symmetric state can be generated by the circuit of shape
\begin{equation}
\begin{quantikz}[row sep={0.7cm, between origins}, column sep=0.4cm]
    &\gate[2]{~} & & & & &\\
    & &\gate[2]{~} & & & &\\
    & & &\gate[2]{~} & & &\\
    & & & &\gate[2]{~} & &\\
    & & & & &\gate[2]{~} &\\
    & & & & & &
\end{quantikz}\quad,
\end{equation}
which is just a typical instance of 1D finite local-depth circuits (FLDC)~\cite{Zhang2024}. We call it the ``ladder'' ansatz of one layer. We use such a PQC instead of common brickwall circuits because GHZ-type states are long-range entangled, in the sense that they cannot be prepared from product states by applying any constant-depth local circuits. As shown in Figs.~\ref{fig:isingh0} and \ref{fig:isingh05}, it still holds that $\tilde{\sigma}\gtrsim 1$ in the presence of good local minima. The brickwall circuits even get a larger RF in spite that their expressive space does not contain the true ground state, which might be attributed to the fact that their smaller causal cones lead to larger local overparametrization ratios and their expressive space at least contains the approximate symmetry-breaking ground states.


\begin{figure}
    \centering
    \includegraphics[width=0.8\linewidth]{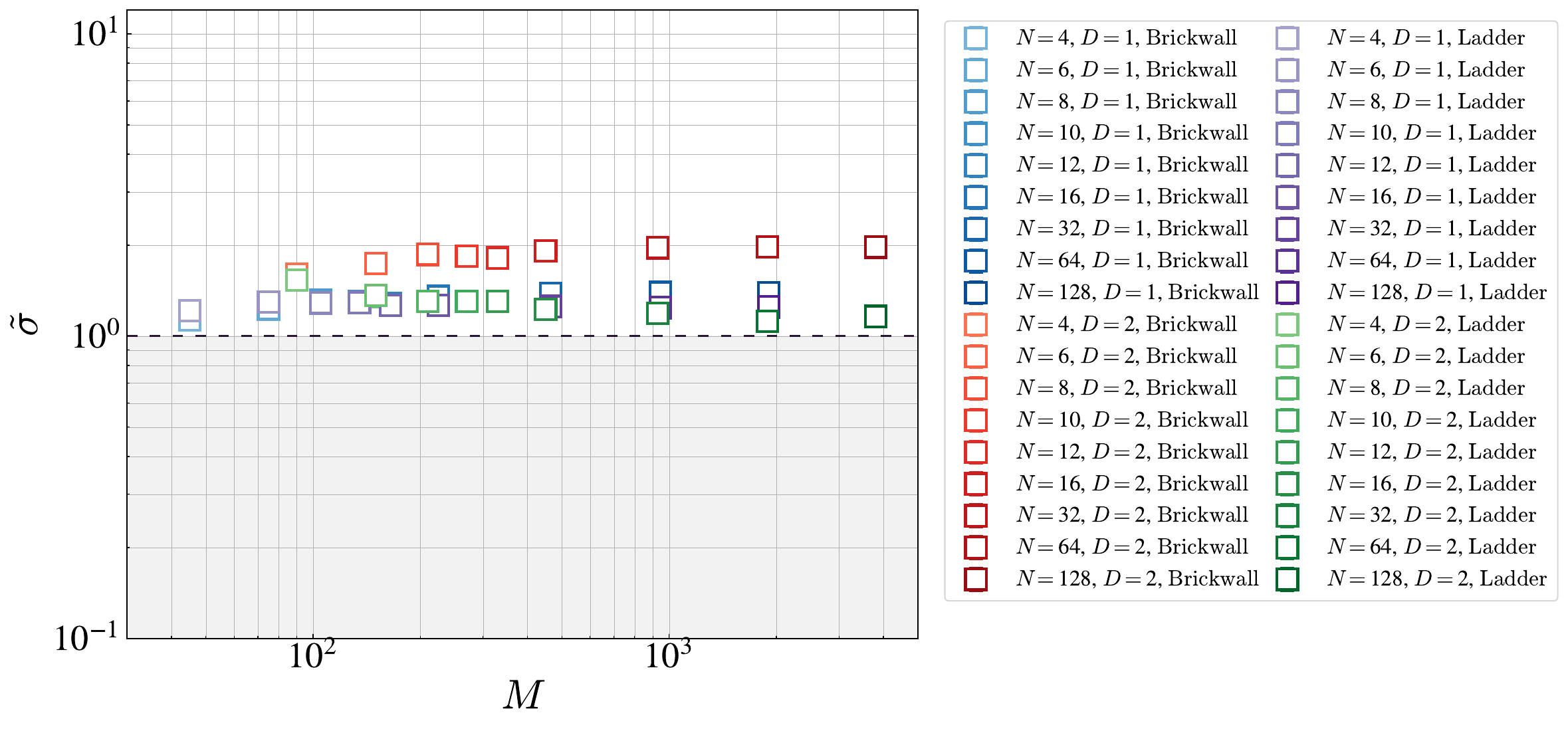}
    \caption{The relative fluctuation $\tilde{\sigma}$ vs parameter count $M$ for the 1D TFI model with $h=0$ using 1D brickwall and ladder circuits of depth $D=1,2$ respectively. The colors with increasing intensity correspond to $N=4,6,8,10,12,16,32,64,128$. 
    }
    \label{fig:isingh0}
\end{figure}

\begin{figure}
    \centering
    \includegraphics[width=0.8\linewidth]{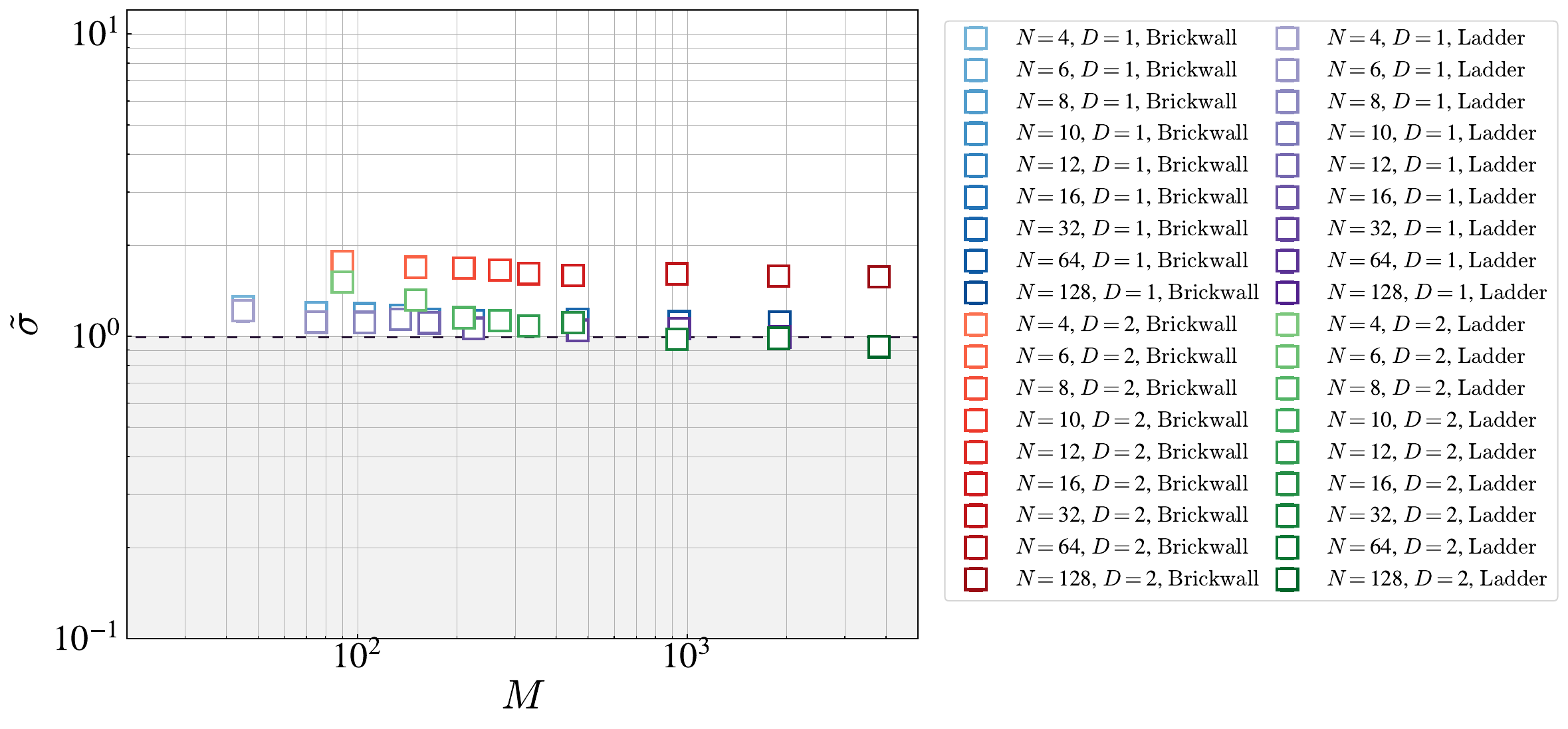}
    \caption{The relative fluctuation $\tilde{\sigma}$ vs parameter count $M$ for the 1D TFI model with $h=0.5$ using 1D brickwall and ladder circuits of depth $D=1,2$ respectively. The colors with increasing intensity correspond to $N=4,6,8,10,12,16,32,64,128$.}
    \label{fig:isingh05}
\end{figure}

\subsection{Details of the experiments on polynomial overparametrization}
In Fig.~\textcolor{blue}{5} of the main text, we investigate the behavior of RFs in the presence of overparametrization with polynomial parameters. The PQC used is the 1D HVA circuit like in QAOA with each layer consisting of a layer of $R_x$ gates and a layer of nearest-neighbor $R_{zz}$ gates with correlating parameters in each sub-layer, i.e.,
\begin{equation}
    \mathbf{U}(\bm{\theta}) = \prod_{l=1}^{D} \exp\left(-i \sum_{j=1}^{N} X_j \theta_{l,2}/2 \right) \exp\left(-i \sum_{j=1}^{N-1} Z_j Z_{j+1} \theta_{l,1}/2\right).
\end{equation}
The initial state is $\ket{+}^{\otimes N}$. A key step is to determine the effective dimension $M_{\text{eff}}$ for a given PQC. As discussed in the last section, we obtain $M_{\text{eff}}$ by computing the rank of the QFI matrix at several random parameter points. The results are consistent with those in Ref.~\cite{Larocca2022}, which show that $M_{\text{eff}}=N^2-N$ for this 1D QAOA-type ansatz with open boundary conditions.

\end{document}